\documentclass[12pt,a4paper]{article}
\usepackage[utf8]{inputenc}
\usepackage[english]{babel}
\usepackage[T1]{fontenc}
\usepackage{amsmath}
\usepackage{amsfonts}
\usepackage{amssymb}
\usepackage{amsthm}
\usepackage{graphicx}
\usepackage[usenames,dvipsnames]{xcolor}
\usepackage{subfig} 
\usepackage{lmodern}
\usepackage[left=2cm,right=2cm,top=2cm,bottom=4cm]{geometry}
\usepackage{float} 
\usepackage{dsfont} 
\usepackage[pdftex,colorlinks,citecolor=blue, urlcolor=blue]{hyperref}
\usepackage{natbib}

\newcommand{\cov}{\mathrm{Cov}}

\newcommand{\indic}[1]{\mathds{1}_{#1}}
\newcommand{\Nx}{N_x}
\newcommand{\Nt}{N_t}
\newcommand{\Dom}{D}
\newcommand{\ta}{\tau}
\newcommand{\suchthat}{\;\ifnum\currentgrouptype=16 \middle\fi|\;}

\usepackage[normalem]{ulem}

\newtheorem{theorem}{Theorem}[section]

\newtheorem{pro}[theorem]{Proposition}

\newtheorem{rem}[theorem]{Remark}

\def \R{{\mathbb R}}

\def \Cxi{{\mathcal C}}
\def \Ixi{{\mathcal I}}

\def\annu#1{_{%
  \vbox{\hrule height .2pt 
    \kern 1pt 
    \hbox{$\scriptstyle {#1}\kern 1pt$}%
  }\kern-.05pt 
  \vrule width .2pt 
}}
\usepackage{ifpdf}
\ifpdf
\else
  \usepackage{graphicx}
  \DeclareGraphicsExtensions{.eps,.jpg,.png}
\fi
\usepackage{epstopdf}

\title{Yield curve interpolation using constraint kriging techniques}

\begin{document}
\begin{center}
\Large{\textbf{Kriging of financial term-structures}}
\end{center}

\vspace{0.4cm}

\begin{center}
 Areski Cousin$^{\dagger}$\footnote{areski.cousin@univ-lyon1.fr}, Hassan Maatouk$^{\ddagger }$\footnote{hassan.maatouk@mines-stetienne.fr} and Didier Rulli\`ere$^{\dagger}$\footnote{didier.rulliere@univ-lyon1.fr}
\end{center}
 
\vskip0.2cm

\begin{center} 
($\dagger$) 
{\small Universit\'{e} de Lyon, Universit\'{e} Claude Bernard Lyon 1, ISFA, Laboratoire SAF  EA2429, 50 av. Tony Garnier, 69366 Lyon, France\\}

($\ddagger$) {\small Mines de Saint-Etienne, 158 Cours Fauriel, 42023 Saint-\'Etienne, France\\}
\end{center}
%

\begin{abstract}
Due to the lack of reliable market information, building financial term-structures may be associated with a significant degree of uncertainty.
In this paper, we propose a new term-structure interpolation method that extends classical spline techniques by additionally allowing for quantification of uncertainty. 
The proposed method is based on a  generalization of kriging models with linear equality constraints (market-fit conditions) and  shape-preserving conditions such as monotonicity or positivity (no-arbitrage conditions). 
We define the most likely curve and show how to build confidence bands.
The Gaussian process covariance hyper-parameters under the construction constraints are estimated using cross-validation techniques.
Based on observed market quotes at different  dates, we demonstrate the efficiency of the method by building curves together with confidence intervals for 
term-structures of OIS discount rates, of zero-coupon swaps rates and of CDS implied default probabilities. We also show how to construct interest-rate surfaces or default probability surfaces by considering time (quotation dates) as an additional dimension.
\end{abstract}




\section{Introduction}


%
%
Constructing  term-structures is at the heart of asset pricing and risk management. 
A term-structure is a curve which describes the evolution of some financial or economic quantities as a function of time horizon. Typical examples are the term-structure of risk-free interest-rates, the term-structure of bond yields or credit spreads, the term-structure of default probabilities or the term-structure of stock return implied volatilities. These curve are typically not directly observed in the market. Thus, the curve construction is based on a  benchmark set of contingent financial instruments whose values explicitly depend on some part of the curve.  In practice, market quotes of these products only provide a partial information on the term-structure since they can only be considered to be reliable  for a small set of liquid maturities.  
The problem is then to transform a small set of market quotes into a continuum set of values representing the evolution of the underlying quantity of interest with respect to time horizon.\\

On practical grounds, the curve is assumed to belong to a family of parametric functions (Nielson-Siegel functional~\cite{nelson1987parsimonious}, polynomial splines, \citet{SW:2001}) and its construction consists in finding the underlying parameters that best fits observed market quotes for all available maturities.
In \citet{deAndresGomez2004}, the interest-rate term-structure is estimated from bid-ask spreads of underlying instruments using fuzzy regression techniques.
 \citet{Hagan2006} provide a review of different interpolation techniques for curve construction. They introduce a monotone convex method and postulate a series of quality criterion such as ability to fit market quotes, arbitrage-freeness, smoothness, locality of interpolation scheme, stability of forward rate and consistency of hedging strategies. \citet{Andersen2007} analyzes the use of hyperbolic tension splines for construction of interest-rate term structures. The underlying optimization allows the user to control the relative importance of fit precision with respect to shape preservation (smoothness of the curve, penalization of oscillations and excess convexity/concavity).  In the same vein, \citet{Chiu_el2008} shows that $L_{1}$ cubic splines minimizes the curve oscillation without sacrificing good approximation of the data.
\citet{Iwashita2013} makes a survey of non-local spline interpolation techniques which preserve stability of forward rates. 
Other papers such as  \citet{Ametrano2009}, \citet{Chibane2009}, \citet{Kenyon2012} or \citet{Fries2013} are concerned with the adaptation of curve construction methods in a multi-curve interest-rate environment. 
Note that, in terms of interpolation scheme, there is no consensus towards a particular best practice method in all circumstances. In addition, the previous  approaches does not account for the uncertainty embedded in the process of curve construction. This could be of primary importance given that the market inputs may be  unreliable or even inexistent for some  maturities.
\\

This issue is related to the study of model uncertainty and its impact on risk management. This topic has been studied since a certain time period and, following the recent financial crisis, has received a particular interest. Impact of model risk on valuation and hedging of financial derivatives have been treated by, among others,  \citet{Derman1996}, \citet{Eberlein1997}, \citet{ElKaroui1998}, \citet{Green1999}, \citet{Branger2004}, \citet{Cont2006}, \citet{Davis2007}, \citet{Henaff2010}, \citet{Morini2011}. In most papers, the question of model risk is restricted to the class of  derivative products. One of the main objective is to quantify model uncertainty, for instance  to obtain bounds for the arbitrage free value of some derivative instruments, given some information on the underlying securities, such as marginal distribution of its price at some particular time horizons.
In contrast, the question of model risk embedded in the construction of marginal distribution or term-structure function themselves has not been investigated as a main object, whatever it  may concern discount curves, zero-coupon curves, swap basis curves, bond term structures or CDS-implied survival curves. \\


From \citet{KW:1970} and \citet{MKGL:1996}, it is well-known that spline fitting is a special case of kriging  \citep[see also][]{bay:hal-01136466, 2016arXiv160202714B}. In addition, kriging allows to account for quantification of uncertainty. Kriging has been developed in geosatistics to estimate the density of some mineral resource in the ground given a relatively small set of borehole, see \cite{krige1951statistical}, \cite{matheron1963principles}, \cite{cressie1990origins}.  Its principle relies on the determination of the conditional distribution of a spatial random field given a set of observed values of the field. The main interest of this method is that it allows to build a predictor of quantities of interest at other locations, as well as uncertainties relying on this prediction.

Kriging is now widely used in many fields like hydrology, air pollution, epidemiology, weather prediction, etc. to interpolate some quantity of interest given some known values at different locations. Despite its popularity, there are relatively few works concerning kriging in actuarial sciences or in finance. Many reference academic journal of these fields give only few or even no entry corresponding to the word ``kriging''. The method is however sometimes referred to using the terms ``Gaussian Processes'' and ``machine learning''. Some existing works using kriging methodology in actuarial sciences and finance concern for example dynamic lifetime adjustments \citep{Debon2010327}, variable annuities valuation \citep{Gan2013795, Gan2015138}, nested simulation of expected shortfall \citep{liu2010stochastic}, Vasicek model calibration \citep{sousa2012machine}, stock market linkages \citep{Asgharian20134738} or credit scoring \citep{FernandesArtes2015}. Other works using spatial techniques are~\cite{Kanevski2008} on interest rates,~\cite{Benth:2015} for energy futures prices.
Some preprints or conference papers also mention the fit of some financial models \citep{stutvoet2007fitting}, spatial insurance \citep{paulson2006spatial}, trading and hedging strategies \citep{baysal2008response}, valuation of Bermudan options \citep{Ludkovski:2015}. Kriging methods naturally rely on some assumptions on the underlying random fields, and one must carefully consider all conditions that must be satisfied before constructing a kriging model.\\

In practice, the term-structure under construction has to satisfy several type of conditions.
One of the most important condition is the compatibility of the curve with market data, i.e., if the curve is used to value a benchmark set of instruments (under a specific pricing rule), the resulting values shall be as close as possible to the observed market quotes. In many classical situations, the market-fit condition translates into a system of linear constraints which can be easily incorporated in kriging techniques.
In addition, kriging can also handle the presence of noisy observations (using the so-called \textit{nugget effect}). This may be relevant in situation where, due to the lack of liquidity, market quotes cannot be considered to be reliable. It is then possible to incorporate an additional level of uncertainty (degree of confidence) associated with market observations.
Monotonicity constraints also appears to be important in many applications. For instance, the price of default-free zero-coupon bonds (or risk-free discount factors) is a non-increasing  function of time-to-maturities under no-arbitrage assumption. Survival functions inferred from CDS spread term-structures are $[0,1]$-valued non-increasing functions.\\

Recently, some authors have studied the integration of monotonicity constraints into Gaussian process emulators, see e.g. \cite{doi:10.1137/140976741} and \cite{kleijnen2012monotonicity}. However, these methods do not guaranty monotonicity constraints in the entire domain. The article~\cite{abrahamsen2001kriging} also deals with the introduction of constraints at some location of a Gaussian process. In \cite{maatouk:hal-01096751}, classical kriging has been improved to tackle monotonicity, positivity constraints or bounds constraints on the curve values. In the present paper, we show how ``constrained'' kriging techniques can be used to extend the classical spline interpolation approaches by additionally quantifying the uncertainty in some illiquid part of the curve. \\

The paper is organized as follows. Section~\ref{Sec:term-structure} states the term-structure construction problem and gives specific examples of market-fit conditions and shape preserving constraints. In Section~\ref{IntModels}, we briefly recall Gaussian process modeling with interpolation conditions or with more general linear equality constraints. In Section~\ref{IntModIConst}, we present the model defined in \cite{maatouk:hal-01096751} to incorporate monotonicity constraints into a Gaussian process emulator. We then study the associated properties such as convergence to the constrained interpolation spline and the estimation of the covariance hyper-parameters.
 In Section~\ref{EmpInv}, based on real market data, we construct curves together with confidence intervals  for different financial term-structures such as OIS discount curves, zero-coupon swap curves and CDS-implied default distributions. 

\section{The term-structure construction problem}
\label{Sec:term-structure}

The main ingredients in the construction of a term-structure function is a set of  financial products whose value depends on some points of the curve. Then, observing the price of these products provides an indirect (and partial) information on the curve. The first step is then to specify the relation between the value of these products and the values of the curve at different time horizons. In this paper, we restrict ourselves to situations where this relation is linear. As we will see, this is the case in many practical situations such as the construction of corporate or sovereign bond yield curve, the construction of OIS discount curves, the construction of forward curves based on fixed-vs-floating interest-rate swaps or  the construction of implied default rates based on CDS spreads.\\

\subsection{Market-fit  and shape-preserving conditions}

The aim is to construct at some quotation date $t$ a term-structure function $T \rightarrow P(t, T)$, based on the observation  of a series of market quotes $S_1(t) \ldots, S_n(t)$ corresponding to the market value of $n$  financial instruments with time-to-maturities $T_1, \ldots, T_n$. In what follows, the quantity $T$ denotes a time length (as opposed to a calendar date), so that $P(t,T)$ corresponds to the value of the curve at time horizon $T$ or at calendar date $t+T$. The observation of market quotes at time $t$ provides a partial information on the curve at a set of time horizons or points $X =  \left(\tau_1, \ldots, \tau_m\right)$, i.e., at some calendar dates $t+\tau_1, \ldots, t+\tau_m$ which typically correspond to  payment dates of cash-flows.\\

The curve is (fully) compatible at time $t$ with market observations if the vector $P(t, X):=\left(P(t,\tau_1), \ldots, P(t,\tau_m)\right)^{\top}$ satisfies a linear system of the form 
\begin{equation}
\label{eq:market_fit}
A_t \cdot P(t, X)=\boldsymbol{b}_t,
\end{equation}
where $A_t$ is $n\times m$ real-valued matrix and $\boldsymbol{b}_t$ is a $n$-dimensional column vector. Of course, $A_t$ and $\boldsymbol{b}_t$ may depend on market quotes $S_1(t) \ldots, S_n(t)$, on the characteristics of the product cash-flows but also on the hypothesis made for assessing the value of these products at time $t$.
 Note that the number of observations $n$ is typically smaller than the number of points $m$, so that the solution of  system~\eqref{eq:market_fit} lives in a linear space  with dimension $m-n$.
\\

When constructing a financial term-structure, one may consider some additional information on the shape of the curve. For instance, the function $T \rightarrow P(t, T)$ may be known to be decreasing with respect to time horizon $T$ and its values may be bounded and belong to the interval $[0,1]$. This is  typically the case when one wants to construct a curve of discount factors (default-free zero-coupon bond prices) or an implied survival function (survival probabilities of a CDS reference entity). Violating these kinds of shape-preserving conditions results in term-structure functions that are typically not arbitrage-free.
\\

\begin{rem}
\label{rem:surface}
The term-structure construction problem can be stated in a two-dimensional setting when the evolution of time (quotation dates) is added as a second dimension. In that case, the aim is to construct a surface $(t,T) \rightarrow P(t, T)$ based on a series of market quotes $S_1(t) \ldots, S_n(t)$ observed at several quotation dates $t=t_1, \ldots, t_N$. Note that the cash-flows characteristics of these products may depend on time $t$. In particular, the underlying maturity dates could be different at every time $t$. If for any date $t=t_1, \ldots, t_N$, the market-fit condition translates into a linear system, then this condition can again be expressed by a single system by concatenating  for every time $t=t_1, \ldots, t_N$ the $N$ systems  given as in Equation~\eqref{eq:market_fit}. The shape-preserving condition can be expressed as the intersection of shape-preserving conditions at time $t=t_1, \ldots, t_N$.
\end{rem}



In the following, we give some examples where  the construction of term-structures  involves combination of linear equality constraints and shape preserving conditions such as monotony or positivity.

\subsection{Classical examples of term-structures}


In what follows, $t$  denotes a particular quotation date, i.e., the date at which the market quotes are observed for the underlying family of contingent products. 
For each example, we only consider a single financial product in this family and we provide the linear relation which characterize its value. These relation then corresponds to one particular row of the linear system~\eqref{eq:market_fit}.\\

\noindent \textbf{Corporate or sovereign bond yield curves}\\

Let $S$ be the observed market price  of a corporate or a sovereign bond with time-to-maturity  $T$ and with a fixed coupon rate $c$. The price $S$ and the coupon rate $c$ are expressed in percentage of invested nominal. The set of coupon payment dates is given by $(t+\ta_1,\ldots, t+\ta_p)$ where $\ta_1<\ldots<\ta_p=T$. The  year fraction $\delta_k$ represents the  time length $\ta_{k}- \ta_{k-1}$, $k=1, \ldots, p$ where  $\ta_0=0$. The present value of this bond can be defined as a linear combination of some default-free zero-coupon bonds, i.e., 
\begin{equation}
\label{eq:bond}
 c\sum_{k=1}^{p}{\delta_k P^{B}(t,\ta_k)} + P^{B}(t, \tau_p) = S,
\end{equation}
where $P^{B}(t,\ta)$ represents the price at time $t$ of a default-free zero-coupon bond with time-to-maturity $\ta$. Note that, even if representation~\eqref{eq:bond} obviously relies on a default-free assumption, it is commonly employed as an intermediary step in the computation of the so-called bond yield-to-maturity.\footnote{The bond yield associated with time-to-maturity $T$ is defined as the constant rate of return $Y(t,T)$ such that the present value relation~\eqref{eq:bond} holds exactly when all the involved ZC bonds have this rate of return.} In this example, the curve $T\rightarrow P^B(t,T)$ will be inferred from a set of market fit conditions similar to~\eqref{eq:bond}, each of them corresponding to specific debt products but with different maturities.
As a result, this set of conditions can be easily represented in the form of a linear system as in Equation~\eqref{eq:market_fit}. In addition, the default-free assumption implies that the curve  $T\rightarrow P^B(t,T)$ is decreasing if arbitrage opportunities are precluded.\\

%
%

\noindent \textbf{Discount curves based on Overnight Indexed Swaps par rates}\\

Due to legal terms of standard collateral agreements, a possible choice to build discount curves is to use market quotes of OIS-like instruments \citep[see, for instance][for more details]{Hull2013}.
Let $S$ be the par swap rate of an overnight indexed swap with maturity $T$ and fixed leg payment dates $\ta_1<\ldots<\ta_p = T$. The  year fraction $\delta_k$ represents the  time length $\ta_{k}- \ta_{k-1}$, $k=1, \ldots, p$ where  $\ta_0=0$. For overnight-index swaps this time length is typically equal to one year.
The swap equilibrium relation takes the following linear form
\begin{equation}
\label{eq:OIS}
S \sum_{k=1}^{p}  \delta_k P^{D}(t,\ta_k)  = 1 - P^{D}(t, T), 
\end{equation}
where $P^{D}(t,\ta)$ is the discount factor at time $t$ associated with a time horizon $\tau$. In the previous equation, the left hand side represents the fixed leg present value whereas  the right hand side corresponds to the floating leg present value. For more details on the derivation of (\ref{eq:OIS}), the reader is referred to \citet{Fujii2010}.
\\

Under some circumstances, discount cuves can also be extracted from par rates of fixed versus Euribor swaps with Euribor tenor of 3 months or 6 months. This is the case for instance in the LTGA framework of Solvability 2  prudential regulation where the basis risk-free rates used for euro are constructed from the euro swap rates with a small adjustment for credit spread \citep[see, e.g.][]{CFO-CRO-2010}.
The resulting market-fit condition will exactly have the same form as~\eqref{eq:OIS}.\\

Note that, in the banking industry, discount curves are now understood as OIS based curves, see e.g. \cite{pallavicini2010interest}. The next example explains how to infer Euribor forward rates from quoted OIS rates and Euribor Swap rates.\\


%
%

\noindent \textbf{Forward curves based on OIS and fixed versus Ibor-floating interest-rate swaps}\\


 Let $S$ be the observed par rate of an interest rate swap with maturity time $T$ and floating payments linked to a Libor or an Euribor rate associated with a tenor $j$ (typically, $j=3$ months or $j=6$ months). The fixed-leg payment scheme is given by $\ta_1<\cdots<\ta_p=T$ and  the floating-leg payment scheme is given by $\tilde{\ta}_1<\cdots<\tilde{\ta}_{q}=T$. 
  For most liquid products, payment on the fixed  leg are made with an annual frequency, so that $\ta_k$ corresponds to $k$ years ahead from the current date $t$. 
The  year fraction $\delta_k$ represents the  time length $\ta_{k}- \ta_{k-1}$, $k=1, \ldots, p$ ($\ta_0=0$) whereas the year fraction  $\tilde{\delta}_i$ represents the time length  $\tilde{\ta}_{i}-\tilde{\ta}_{i-1}$, $i=1, \ldots, q$ ($\tilde{\ta}_0=0$).   
Note that the length between two consecutive dates on the floating leg should correspond to the Libor or Euribor tenor, i.e., depending on the case, $\tilde{\delta}_i \simeq$ 3 months or $\tilde{\delta}_i \simeq$ 6 months. As a result, given an OIS discount curve $P^{D}$, the swap equilibrium relation can be represented in a linear form with respect to some forward Libor or Euribor rates, i.e.,  
\begin{equation}
\label{eq:IRS}
S \sum_{k=1}^{p}  \delta_k P^{D}(t,\ta_k) = \sum_{i=1}^{q} P^{D}(t, \tilde{\ta}_i) \tilde{\delta}_i F_j(t, \tilde{\ta}_{i} ),
\end{equation}
where $P^{D}(t, \tilde{\ta})$ is a risk-free discount factor at time $t$ for maturity $\ta$ and  $F_j(t, \tilde{\ta}_{i} ) := F(t,\tilde{\ta}_{i-1}, \tilde{\ta}_{i} )$ is the forward Libor or Euribor rate defined as the fixed rate to be exchanged at time $\tilde{\ta}_i$ against the $j$-tenor Libor or Euribor rate established at time $\tilde{\ta}_{i-1}$ so that the swap has zero value at time $t$. As in the previous example, the left hand side of~\ref{eq:IRS} represents the fixed leg present value whereas the right hand side corresponds to the floating leg present value. For more details, see, for instance \citet{Chibane2009}. Given a (pre-constructed) discount curve $P^D$ and  a set of swap par rates $S = S_1, \ldots, S_n$ corresponding to time-$t$ market quotes of Euribor or Libor swaps with maturities $T=T_1, \ldots, T_n$, the forward curve $T\rightarrow F_j(t, T)$ under construction has to satisfy a linear market-fit condition. This condition takes the form of a linear system whose each line is given by~\eqref{eq:IRS}. In addition, one can  require that forward rates are positive, so that the underlying pseudo zero-coupon prices form a decreasing function of time horizon. This particular shape preserving property (positivity) can also be enforced in the interpolation procedure we will described in the next section.\\

\noindent \textbf{Credit curves based on Credit Default Swaps spreads}\\

Let $S$ be the fair spread of a credit default swap with protection time  horizon $T$ and with premium payment dates $\ta_1<\cdots<\ta_p=T$. If we denote by $R$ the  expected recovery rate of the reference entity  and by 
$\delta_k$ the year fraction corresponding to time length $\ta_{k}- \ta_{k-1}$ ($\ta_0 = 0$), then 
 the CDS swap equilibrium relation can be expressed as
\begin{equation}
\label{eq:CDS}
  S \sum_{k=1}^p \delta_k P^D(t, \ta_k) Q(t, \ta_k)  = -(1-R)\int_{0}^T P^D(t, \ta) dQ(t,\ta),
\end{equation}
where $P^{D}(t, \ta)$ is the risk-free discount factor at time $t$ for time horizon $\ta$ and where $Q(t,\ta)$ is the probability (at time $t$) that the underlying reference entity has not defaulted before time horizon $\ta$. Then, $Q(t,\ta)$ is the survival probability of the  debt issuer in the time horizon $\ta$.  The left hand side of~\eqref{eq:CDS} represents the premium leg present value whereas the right hand side corresponds to the protection leg (or default leg) present value.
We implicitly assume here that recovery, default and interest rates are stochastically independent.
Using an integration by parts, it is straightfoward to show that survival probabilities $Q(t,\ta)$, $0 \leq \ta \leq T$,  are linked through the following linear relation~:
\begin{equation}
\label{eq:linear_CDS}
\begin{split}
S &\sum_{k=1}^p \delta_k P^D(t,\ta_k) Q(t, \ta_k) + (1-R) P^D(t,T)Q(t,T)\\ 
& \hspace{0.5 cm} + (1-R)\int_{0}^T f^D(t,\ta)P^D(t, \ta)Q(t,\ta)d\ta = 1 - R
\end{split}
\end{equation}
where $f^D(t,\ta)$ is the instantaneous forward rate\footnote{Instantaneous forward rates can be derived from discount factors through the following relation~: $f^D(t, \ta)P^D(t, \ta) = - \frac{\partial P}{\partial \ta} (t, \ta)$.} at time $t$ for  time horizon $\ta$.
In practice, the integral involved in the expression of the protection leg present value is classically discretized on the premium time grid $\ta_1<\cdots<\ta_n=T$, so that the continuous linear condition~\eqref{eq:linear_CDS} can be stated as a discrete one given by
\begin{equation}
\label{eq:discrete_linear_CDS}
\begin{split}
 &\sum_{k=1}^{p-1} \left(S\delta_k P^D(t,\ta_k) + (1-R)(P^D(t,\ta_{k-1}) - P^D(t,\ta_k))\right) Q(t, \ta_k) \\ 
& \hspace{0.5 cm} + \left(S\delta_p P^D(t,T) + (1-R) P^D(t,\ta_{p-1})\right) Q(t, T) = 1 - R.
\end{split}
\end{equation}
At some fixed quotation date $t$, CDS protection is usually available for a set of liquid maturities $T_1, \ldots, T_n$. Then the construction of an implied survival function $T\rightarrow Q(t,T)$ consists in building a $[0,1]$-valued decreasing function that satisfies a system of $n$  linear equality constraints as of~\eqref{eq:discrete_linear_CDS}.



%
%
%
%

%

\section{Kriging under linear equality constraints}

\label{IntModels}

The term-structure construction approach we propose relies on an extension of kriging to linear equality and shape preserving constraints. In this section, we give a formal presentation of this interpolation technique when only linear equality constraints are considered. Section~\ref{IntModIConst} explains how this technique can be generalized when some monotonicity constraints are added.\\ 

Kriging or Gaussian process regression is a method of interpolation for which the interpolated values are modeled by a Gaussian Process (GP) with a prior covariance function. This method is widely used in the domain of spatial analysis and computer experiments \citep[see, e.g.][]{Rasmussen:2005:GPM:1162254}. More formally, we consider the model $y=f(\boldsymbol{x})$, where $f$ is an unknown real-valued function of a $d$-dimensional input variable $\boldsymbol{x}\in\mathbb{R}^d$. In the case where the computation of $f$ is expensive and time-consuming or in the case where $f$ is known only at some input locations, this function can be estimated using so-called \emph{kriging} techniques. 
In that case, $f$  is seen as a realization of a Gaussian Process (GP) $Y$ defined as
\begin{equation*}
Y(\boldsymbol{x}):=\mu(\boldsymbol{x})+Z(\boldsymbol{x}),
\end{equation*}
where the deterministic function $\mu \ : \ \boldsymbol{x}\in \mathbb{R}^d \longrightarrow \mu(\boldsymbol{x})\in\mathbb{R}$
is the mean of $Y$ and $Z$ is a zero-mean GP with covariance function
\begin{equation*}
K \ : \ (\boldsymbol{x},\boldsymbol{x}')\in \mathbb{R}^d\times\mathbb{R}^d \longrightarrow K(\boldsymbol{x},\boldsymbol{x}')=\cov(Y(\boldsymbol{x}),Y(\boldsymbol{x}'))\in\mathbb{R}.
\end{equation*}
The quantity $K(\boldsymbol{x},\boldsymbol{x}')$ is thus the covariance of the values of the random field $Y$ at two input locations $\boldsymbol{x}$ and $\boldsymbol{x}'$.
We assume that the covariance function $K$ is such that the random field $Y$ have continuous and differentiable sample paths with probability one, see \cite{abrahamsen1997review}. Figure~\ref{GPsimkernels} (right) gives an illustration of some  Gaussian processes sample paths.  In numerical illustrations of Section~\ref{EmpInv}, we consider Gaussian processes with $d$-dimensional covariance functions given as a tensor product, i.e., for $\boldsymbol{x}=(x_1,\ldots,x_d)$ and $\boldsymbol{x}'=(x'_1,\ldots,x'_d)$~:
\begin{equation*}
K(\boldsymbol{x},\boldsymbol{x}')=\sigma^2\prod_{i=1}^dC_i(x_i-x'_i,\theta_i),
\end{equation*}
where $\boldsymbol{\theta}=(\theta_1,\ldots,\theta_d)\in\mathbb{R}^d$ and $\sigma^2$ are respectively called \textit{length} and \textit{variance hyper parameters}.
The function $C_i$ are kernel correlation functions which depend on the length parameter $\theta_i$ and on $x_i-x'_i, \ i=1,\ldots,d$, see Table~\ref{kernel} for some popular kernel correlation functions. Note that the length parameter $\theta$ can be interpreted as a correlation parameter\footnote{For a Gaussian covariance kernel, it can be shown that, for any $n$,  increasing $\boldsymbol{\theta}$ yields an increase of the vector $(Y(x^{1}), \ldots, Y(x^{n}))$ with respect to the supermodular order, a stochastic order which is well-known for comparing the degree of dependence.} as it controls the degree of dependence amongst the values of the Gaussian process at any two points. The parameter $\sigma$ controls the initial Gaussian process variance.

\subsection{Classical kriging}
In classical kriging, the real function $f$ is known to take some values $y_1,\ldots,y_n$ at some $d$-dimensional design points $x^{(1)},\ldots,x^{(n)}$, so that $f(X)=\boldsymbol{y}$, where the design points are given as the rows of the $n\times d$ matrix $X=\left(x^{(1)},\ldots,x^{(n)}\right)^\top\in\mathbb{R}^{n\times d}$, $f(X)=\left(f(x^{(1)}),\ldots,f(x^{(n)})\right)^\top\in\mathbb{R}^n$ and $\boldsymbol{y}=(y_1,\ldots,y_n)^\top\in\mathbb{R}^n$. One advantage of using a GP emulator is that, conditionally to the observation data $\boldsymbol{y}$, the conditional process $Y \suchthat Y(X)=\boldsymbol{y}$ is still a GP. This process is characterized by its (marginal) mean 
\begin{equation}
\label{eq:kriging_mean}
\eta(\boldsymbol{x})=\mu(\boldsymbol{x})+\boldsymbol{k}(\boldsymbol{x})^\top\mathbb{K}^{-1}(\boldsymbol{y}-\boldsymbol{\mu}), \quad  
\boldsymbol{x}\in\mathbb{R}^d
\end{equation}
and its covariance function $\tilde{K}$ given by
\begin{equation}
\tilde{K}(\boldsymbol{x},\boldsymbol{x}')= K(\boldsymbol{x},\boldsymbol{x}')-\boldsymbol{k}(\boldsymbol{x})^\top\mathbb{K}^{-1}\boldsymbol{k}(\boldsymbol{x}'),\quad  
\boldsymbol{x},\boldsymbol{x}' \in\mathbb{R}^d
\end{equation}
where $\boldsymbol{\mu}=\mu(X)=\left(\mu(x^{(1)}),\ldots,\mu(x^{(n)})\right)^\top\in\mathbb{R}^n$ is the trend vector at the design points, $\mathbb{K}$ is the covariance matrix
of $Y(X)$ and $\boldsymbol{k}(\boldsymbol{x})=\left(K\left(\boldsymbol{x},x^{(1)}\right),\ldots,K\left(\boldsymbol{x},x^{(n)}\right)\right)^\top$ is the vector of covariance between $Y(\boldsymbol{x})$ and $Y(X)$. The conditional mean $\eta(\boldsymbol{x})$ given the observation data $Y(X)=\boldsymbol{y}$ is the Best Linear Unbiased Estimator (BLUE) of $Y(\boldsymbol{x})$, which is known as \textit{kriging mean} \citep[see][]{jones1998}. 
One remarkable property is that the covariance function $\tilde{K}$ of the conditional Gaussian process does not depend on the observation data $\boldsymbol{y}$. 
In addition, the regularity of the kriging mean predictor function $\eta$ inherits from  the regularity of the mean function $\mu$ and from the regularity of the
covariance function $K$ of the original GP $Y$. 
Then, the choice of this covariance function  is essential because it drives the smoothness of the kriging metamodel. Table~\ref{kernel} gives some popular  kernel correlation functions, ordered by decreasing degree of smoothness.
Figure~\ref{GPsimkernels} shows three alternative covariance functions with  associated Gaussian process sample paths.  
\begin{figure}[hptb]
\begin{minipage}{.5\linewidth}
\centering
\includegraphics[scale=.4]{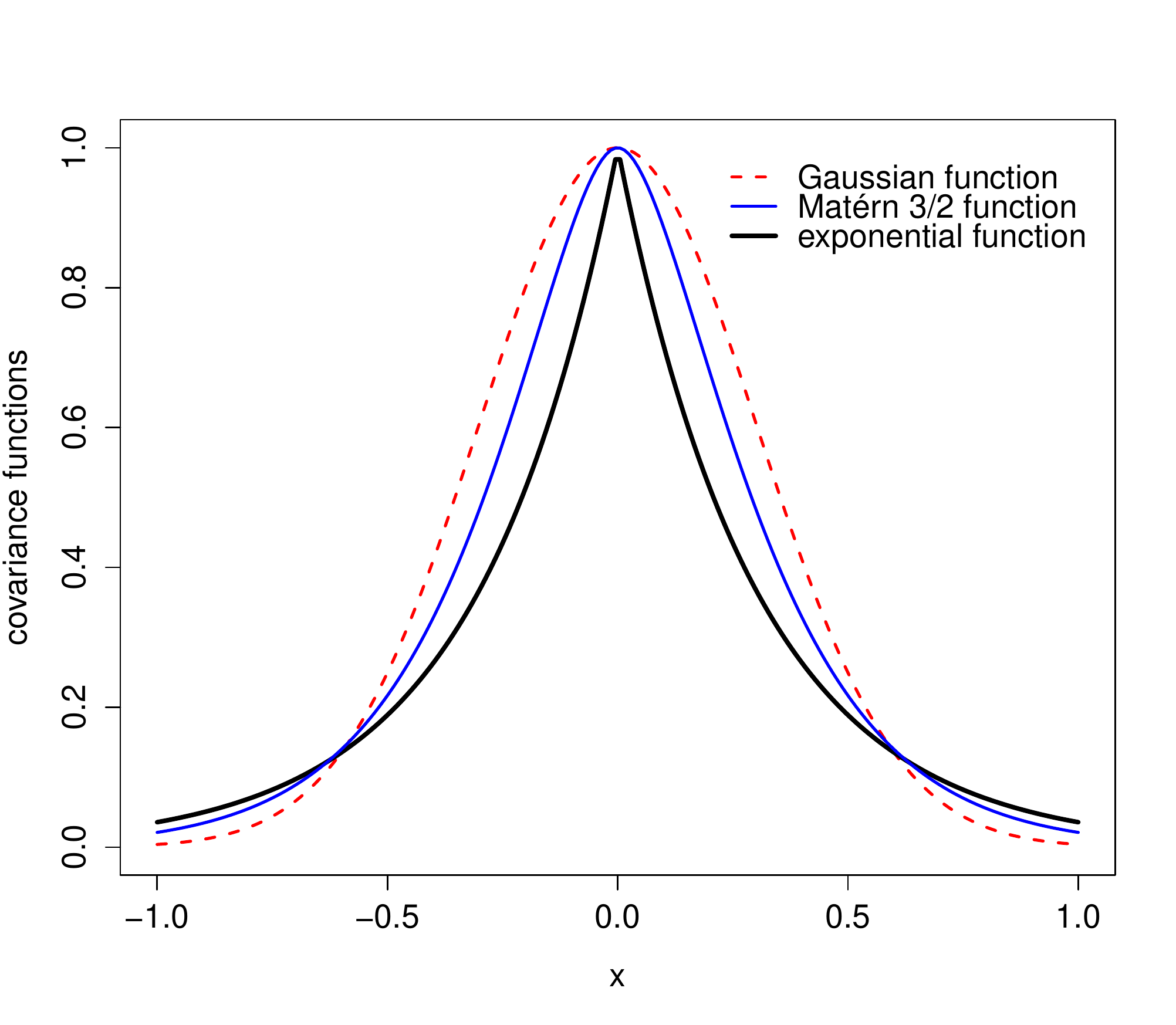}
\end{minipage}%
\begin{minipage}{.5\linewidth}
\centering
\includegraphics[scale=.4]{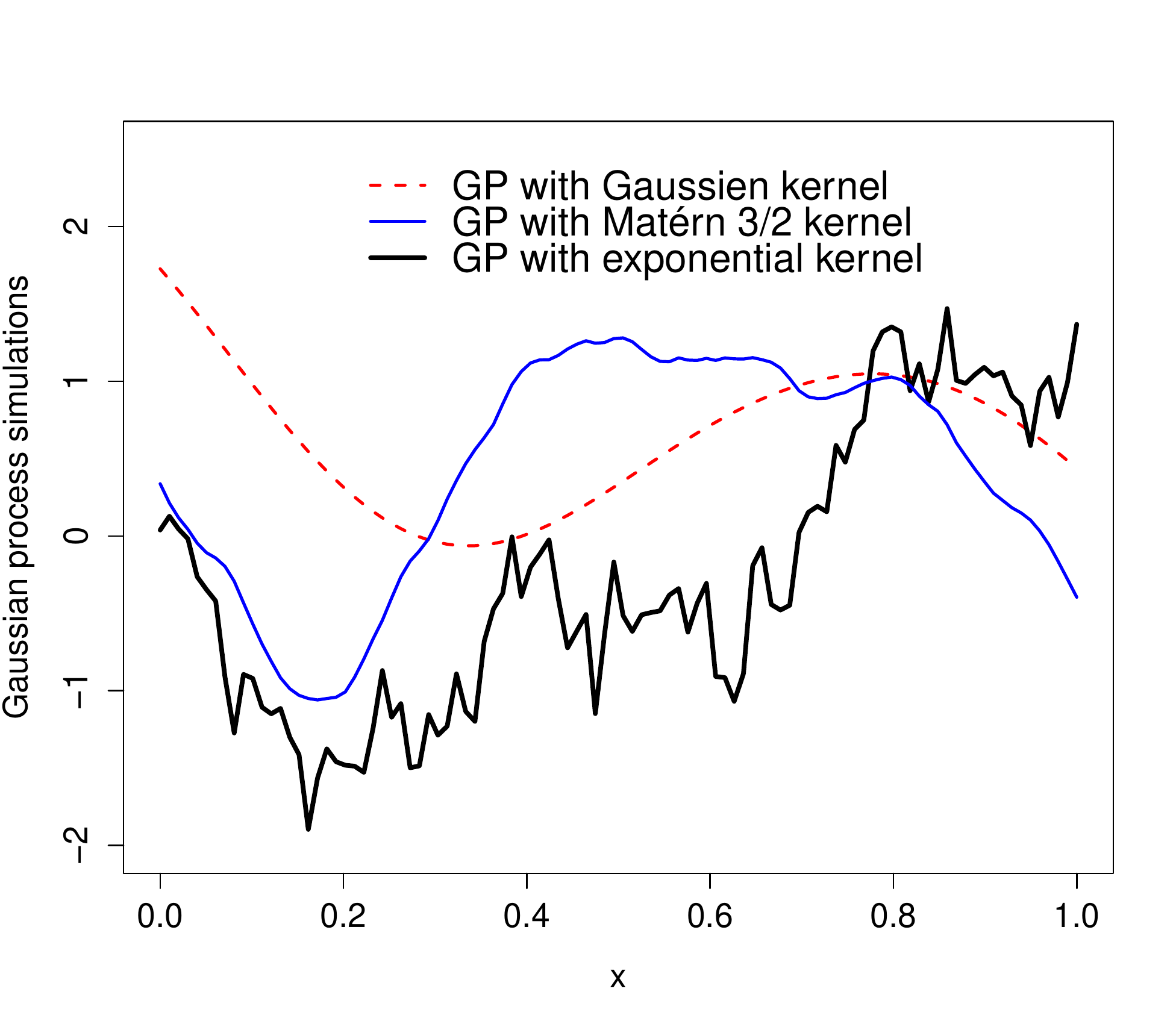}
\end{minipage}
\caption{Some Gaussian process covariance functions (left) and associated sample paths (right).
The covariance parameters are fixed to $(\theta, \sigma)=(0.3,\, 1.0)$.}
\label{GPsimkernels}
\end{figure}

\begin{table}[hptb]
\centering
\caption{Some popular kernel correlation functions $C(x-x',\theta)$ used in kriging methods.}
\label{kernel}
\begin{tabular}{ccc}
\hline \hline
Name  &  Expression &  Class   \\   
\hline
Gaussian       & $\exp\left(-\frac{(x-x')^2}{2\theta^2}\right)$ 	& $\mathcal{C}^{\infty}$       \\
Mat\'ern 5/2   & $\left(1+\frac{\sqrt{5}| x-x' |}{\theta}+\frac{5(x-x')^2}{3\theta^2}\right)\exp\left(-\frac{\sqrt{5} | x-x'|}{\theta}\right)$
& $\mathcal{C}^2$      \\
Mat\'ern 3/2   &  $\left(1+\frac{\sqrt{3}| x-x' |}{\theta}\right)\exp\left(-\frac{\sqrt{3}|x-x'|}{\theta}\right)$   & $\mathcal{C}^1$      \\
Exponential    &  $\exp\left(-\frac{|x-x'|}{\theta}\right)$ & $\mathcal{C}^0$      \\   
\hline
\end{tabular}
\end{table}

\subsection{Extension to linear equality constraints}\label{ILC}

The previous setting can be generalized by considering linear equality constraints instead of pure interpolation constraints. This is of primary importance if one wants to construct term-structures which are compatible with linear market-fit conditions. We are under the situation where $n$ relevant financial products are considered to construct the curve and their market quotes provide information on the curve at the $m$ points $x^{(1)},\ldots,x^{(m)}$. Then, the (unknown) real function $f$ satisfies some linear constraints of the form
\begin{equation}\label{initLinConst}
A\cdot f(X)=\boldsymbol{b},
\end{equation}
where $A$ is a given matrix of dimension $n\times m$, $n,m\in\mathbb{N}$, $X=\left(x^{(1)},\ldots,x^{(m)}\right)^\top\in\mathbb{R}^{m\times d}$, $f(X)=\left(f(x^{(1)}),\ldots,f(x^{(m)})\right)^\top \in \mathbb{R}^m$ and $\boldsymbol{b}\in\mathbb{R}^n$. 
In that case, the conditional process $Y \suchthat A\cdot Y(X)=\boldsymbol{b}$ is still a Gaussian process with mean
\begin{equation}
\label{eq:mean}
\eta(\boldsymbol{x})=\mu(\boldsymbol{x})+\left(A\boldsymbol{k}(\boldsymbol{x})\right)^\top \left(A\mathbb{K}A^\top\right)^{-1}(\boldsymbol{b}-A\boldsymbol{\mu}), \quad  
\boldsymbol{x}\in\mathbb{R}^d
\end{equation}
and covariance function 
\begin{equation}
\label{eq:cov}
\tilde{K}(\boldsymbol{x},\boldsymbol{x}')= K(\boldsymbol{x},\boldsymbol{x}')-\left(A\boldsymbol{k}(\boldsymbol{x})\right)^\top \left(A\mathbb{K}A^\top\right)^{-1}A\boldsymbol{k}(\boldsymbol{x}'),\quad  
\boldsymbol{x},\boldsymbol{x}' \in\mathbb{R}^d
\end{equation}
where $\boldsymbol{\mu}=\mu(X)=\left(\mu(x^{(1)}),\ldots,\mu(x^{(m)})\right)^\top\in\mathbb{R}^m$ is the trend vector at the design points, $\mathbb{K}$ is the covariance matrix
of $Y(X)$ and $\boldsymbol{k}(\boldsymbol{x})=\left(K\left(\boldsymbol{x},x^{(1)}\right),\ldots,K\left(\boldsymbol{x},x^{(m)}\right)\right)^\top$ is the vector of covariance between $Y(\boldsymbol{x})$ and $Y(X)$.
Note that when $A$ is a square identity matrix, the linear constraints become interpolation constraints.

\subsection{Kriging in the presence of  noisy observations}\label{subsec:noisy}

The previous framework can be further extended to situations where market observations cannot be considered to be fully reliable due to, e.g., market microstructure effect. We assume that this uncertainty blurred the market information on $f$ in such a way that the previous measurement equation \eqref{initLinConst} is flawed by an additive error term $\epsilon$ :
$$
 \boldsymbol{b} = A\cdot f(X) + \epsilon.
$$
We consider that this error term is one realization of an independent zero-mean Gaussian noise $\boldsymbol{\varepsilon}$ in $\R^{m}$ with  covariance matrix $\Sigma$. 
If the measurement error $\boldsymbol{\varepsilon}$ is assumed to be  independent of the Gaussian process $Y$,  the conditional process $Y \suchthat  \boldsymbol{b} = A\cdot Y(X) + \boldsymbol{\varepsilon}$ is still a Gaussian process with mean
\begin{equation}
\label{eq:mean_noisy}
\eta(\boldsymbol{x})=\mu(\boldsymbol{x})+\left(A\boldsymbol{k}(\boldsymbol{x})\right)^\top \left(A\mathbb{K}A^\top + \Sigma\right)^{-1}(\boldsymbol{b}-A\boldsymbol{\mu}), \quad  
\boldsymbol{x}\in\mathbb{R}^d
\end{equation}
and covariance function 
\begin{equation}
\label{eq:cov_noisy}
\tilde{K}(\boldsymbol{x},\boldsymbol{x}')= K(\boldsymbol{x},\boldsymbol{x}')-\left(A\boldsymbol{k}(\boldsymbol{x})\right)^\top \left(A\mathbb{K}A^\top + \Sigma\right)^{-1}A\boldsymbol{k}(\boldsymbol{x}'),\quad  
\boldsymbol{x},\boldsymbol{x}' \in\mathbb{R}^d
\end{equation}
where $\boldsymbol{\mu}=\mu(X)=\left(\mu(x^{(1)}),\ldots,\mu(x^{(m)})\right)^\top\in\mathbb{R}^m$ is the trend vector at the design points, $\mathbb{K}$ is the covariance matrix
of $Y(X)$ and $\boldsymbol{k}(\boldsymbol{x})=\left(K\left(\boldsymbol{x},x^{(1)}\right),\ldots,K\left(\boldsymbol{x},x^{(m)}\right)\right)^\top$ is the vector of covariance between $Y(\boldsymbol{x})$ and $Y(X)$.
Note that the only difference compared to noiseless kriging equations \eqref{eq:mean} and  \eqref{eq:cov} is the replacement, at every occurrence, of the covariance matrix $A\mathbb{K}A^\top$ by an inflated matrix $A\mathbb{K}A^\top + \Sigma$. In addition, the curve associated with the kriging mean function \eqref{eq:mean_noisy} does not satisfy the noiseless market fit condition.\\


The possibility to handle the presence of noisy observations is quite important. The presence of measurement errors and observations with contaminations may have important consequences on the dynamics of yield curves, as noticed for example in~\cite{LauriniOhashi2015}. The previous framework can be used to construct term structure functions in the presence of illiquid securities. Assume that the vector $\boldsymbol{b}$ represents mid-prices of some reference products in a market situation where bid-ask spreads are large. For instance, the covariance matrix $\Sigma$ can be chosen in such a way that the range of $(\boldsymbol{b}-\boldsymbol{\varepsilon})$ is concentrated on the bid-ask interval.
One possible way to do that is to define $\Sigma$ as a diagonal matrix where each positive term is the square difference between ask and mid prices.\\ 

Note moreover that this extension can also be useful when one wants to explicitly look for an attractive trade-off between precision and curve smoothness. Indeed, considering that observed information is uncertain allows to weaken the market-fit condition which in turn benefits to curve smoothness.
Our methodology can be then adapted both for best-fitting of noisy market prices and for the construction of an exact interpolatory term structure to a set of liquid instruments.

\section{Kriging under additional monotonicity constraints}\label{IntModIConst}

As mentioned in Section~\ref{Sec:term-structure}, the studied real function $f$ may be known to satisfy some shape-preserving constraints such as monotonicity or positivity. 
\emph{No-Arbitrage constraints} lead to a theoretical need of monotonicity in curve constructions. For example, quantities such as prices of default-free zero-coupon bond, discount factors or implied survival probabilities shall be non-increasing with respect to time horizons under the no-arbitrage condition.  In this section, we propose to extend constrained spline techniques to constrained kriging in order to build term-structure functions.\\ 

The use of monotonic splines  for term-structure construction has been a subject of great interest in the literature. This has been investigated by among other \cite{Barzanti-Corradi99}, \cite{ramponi2003adaptive},  \cite{Chiu_el2008}. The relationship of shape and monotonicity restrictions with no-arbitrage conditions is detailed for example in the recent literature, see~\cite{Laurini:2010} and~\cite{Fengler:2015}. 
Additional motivations of the use of monotonicity constraints will be developed in the numerical Section~\ref{sec:motivMonotonie}.\\

A natural extension of kriging under monotonicity constraints is to consider conditional Gaussian processes with monotone paths. However, the difficulty is that the conditional monotone process is not a Gaussian process any more. We adopt here the approach introduced by \cite{maatouk:hal-01096751} where Gaussian processes are approximated by finite-dimensional versions, so that the monotonicity constraints can  be checked very efficiently in the entire domain. \\ 

In the following, we consider linear equality constraints as described in Section~\ref{ILC} and we explain how to incorporate supplementary monotonicity constraints. Note that, even if we focus on monotonicity constraints, other shape-preserving constraints can be incorporated using similar ideas \citep[see][]{maatouk:hal-01096751}.
We first introduce monotonicity constraints under 
 the  one-dimensional case, where one has to retrieve a monotonic curve at a given quotation date (Section~\ref{ODC}). We then explain how to extend the one-dimensional kriging construction method to dimension two, where information at several quotation dates can be jointly used  (Section~\ref{TDC}).


\subsection{One dimensional case}\label{ODC}

In this section, we assume that the input variable $x$ belongs to an interval $\Dom = [\underline{x}, \overline{x}]$  of $\mathbb{R}$ and we consider an original Gaussian process $Y$ with covariance function $K$. For simplicity, we assume that $Y$ is a zero-mean GP. 
The aim of this section is to explain how to construct a process that both satisfies linear equality constraints  and monotonicity constraints.\\

 The first step is to approximate the original Gaussian process $Y$ by a finite-dimensional version which is monotonic under some simple finite-dimensional linear inequality conditions.
  We begin by discretizing the input interval $\Dom$ into  a regular subdivision $u_0<\ldots < u_N$ with $u_0=\underline{x}$, $u_N=\overline{x}$ and with a constant mesh $\delta$, so that $u_j=u_0+j\delta$, $j=0, \ldots, N$. We then consider an associated set of basis functions  $\phi_j$,  $j=0, \ldots, N$ defined as 
\begin{equation*}
\phi_j(x)=\int_{\underline{x}}^xh_j(u)du, \qquad x\in \Dom,
\end{equation*}
where $h_j(x):=\max\left(1-\frac{|x-u_j|}{\delta},0\right)$ is a hat function centered at the $j^{\text{th}}$ knot $u_j$ of the input subdivision.  Note that the basis functions  $\phi_j$,  $j=0, \ldots, N$ are increasing as primitives of positive functions.
We then define the finite-dimensional approximation of $Y$ on $\Dom$ as the process $Y^N$ such that 
%
%
\begin{equation}\label{propModel}
Y^N(x)=\eta+\sum_{j=0}^N\xi_j\phi_j(x),\qquad x\in\Dom,
\end{equation} 
where $\boldsymbol{\xi}=(\eta,\xi_0,\ldots,\xi_N)^\top$ is a zero-mean Gaussian vector.
If the Gaussian process $Y$ has almost surely  differentiable paths, choosing $\eta = Y(u_0)$ and $\xi_{j} = Y'(u_j)$, $j=0,\ldots,N$ guarantees that the finite-dimensional process $Y^N$ uniformly converges on $D$ towards $Y$ almost surely as $N$ tends to infinity \citep[see][]{maatouk:hal-01096751}. In that case, the
 covariance matrix $\Gamma^N$  of $\boldsymbol{\xi}$ is given as
\begin{equation}
\label{eq:cov_mat_one_dim}
\Gamma^N=\left[\begin{matrix}
K(u_0,u_0) && \frac{\partial K}{\partial x'}(u_0,u_j) \\\\
\frac{\partial K}{\partial x}(u_i,u_0) && \frac{\partial^2K}{\partial x\partial x'}(u_i,u_j)\\
\end{matrix}
\right]_{0\leq i,j\leq N},
\end{equation}
where $K$ is the covariance matrix of the original GP $Y$ and $u_j, \ j=0,\ldots,N$ are the knots of the input subdivision. 

\begin{pro}[Monotonicity]
\label{pro:monotone_1D}
The process $Y^N$  defined in Equation~\eqref{propModel} is non-increasing (resp. non-decreasing) \textit{if and only if} all the coefficients $\xi_j, \ j=0,\ldots,N$ are nonpositive (resp. nonnegative).
\end{pro}

\begin{proof}
If $\xi_j, \ j=0,\ldots,N$ are nonpositive then, since $\phi_j$ are increasing, then $Y^N$ is increasing. For the reverse implication, let us first notice that
the derivative of the basis functions $\phi_j, \ j=0,\ldots,N$ are such that
\begin{eqnarray*}
\phi_j'(u_k)=h_j(u_k)=\indic{j=k}=\left\{\begin{array}{ll}
1 & \mbox{if} \ j=k\\
0 & \mbox{if} \ j\neq k 
\end{array}\right.
\end{eqnarray*}
Thus, the derivative of the process $Y^N$ at any knots $u_k, \ k=0,\ldots,N$ is
\begin{equation*}
\left(Y^N\right)'(u_k)=\sum_{j=0}^N\xi_j\phi'_j(u_k)=\xi_k,
\end{equation*}
which concludes the proof.
\end{proof}
The choice of the basis functions $\phi_j$ and of $\Gamma^N$ depends on the type of shape-preserving constraints. Other type of basis function can be used for other constraints  \citep[for more details, see][]{maatouk:hal-01096751}.\\

In order to construct curves which are compatible with market quotes, linear equality constraints like the one given in Subsection~\ref{ILC} have to be imposed on the process $Y^N$ at some points $x^{(1)},\ldots,x^{(m)}$ in $\Dom$. Then, if $Y^N(X)=\left(Y^N\left(x^{(1)}\right), \ldots,Y^N\left(x^{(m)}\right)\right)^\top$ denotes the vector of values involved in the curve construction and given Equation~\eqref{propModel}, the condition $A\cdot Y^N(X) = \boldsymbol{b}$ translates into the following linear equality constraint on the Gaussian vector $\boldsymbol{\xi}$~:
\begin{equation}\label{EqCond}
A \cdot \Phi \cdot \boldsymbol{\xi}=\boldsymbol{b},
\end{equation}
where $\Phi$ is a $m\times (N+2)$ matrix defined as
\begin{eqnarray*}
\Phi_{i,j}:=\left\{\begin{array}{ll}
1 & \mbox{for} \ i=1,\ldots,m \ \text{and} \ j=1,\\
\phi_{j-2}\left(x^{(i)}\right) & \mbox{for} \ i=1,\ldots,m \ \text{and} \ j=2,\ldots,N+2.
\end{array}\right.
\end{eqnarray*}
Note that, generally speaking,  the linear equality condition \eqref{EqCond} on $\boldsymbol{\xi}$ admits solutions only when $N+2\geq n$ as  $A \cdot \Phi$ is a matrix of dimension $n\times (N+2)$.

\paragraph{Curve simulation.}  Conditional GP satisfying both monotonicity and linear equality constraints can be sampled by  generating truncated Gaussian vector $\boldsymbol{\xi}$ restricted to~: 
\begin{eqnarray*}\left\{\begin{array}{ll}
B \cdot \boldsymbol{\xi}=\boldsymbol{b} & \qquad \mbox{linear equality condition}\\
\boldsymbol{\xi}\in \Cxi & \qquad \mbox{monotonicity constraint}
\end{array}\right.
\end{eqnarray*}  
where,  $B  = A \cdot \Phi$ and for instance $\Cxi=\left\{\boldsymbol{\xi}\in\mathbb{R}^{N+2} \ : \ \xi_j\leq 0, \ j=0,\ldots,N\right\}$ for non-increasing constraints. 
 Then, simulated paths can be sampled in two steps.  First, the conditional  distribution of the vector $\boldsymbol{\xi}$ given $B \cdot \boldsymbol{\xi}=\boldsymbol{b}$ is still Gaussian with mean 
 $$
 (B\Gamma^N)^\top\left(B\Gamma^N  B^\top\right)^{-1}\boldsymbol{b}
 $$
 and covariance matrix
 $$
 \Gamma^N-\left(B\Gamma^N\right)^\top\left(B\Gamma^N  B^\top\right)^{-1}B\Gamma^N,
 $$
so that it can be simulated very efficiently.
Then, the simulation of truncated Gaussian vectors restricted to, for instance, negativity of the components (here $\boldsymbol{\xi}\in \Cxi$) can be done by using improved rejection sampling algorithm such as the one described in \cite{MaatoukMCQMC2014} and \cite{Robert}.
 By Equation~\eqref{propModel} we get sample paths that fulfil both constraints. These simulations can be used to construct confidence intervals for the value of the curve at each point $x$ in $\Dom$.

\paragraph{Most likely curve.} 

Given a covariance kernel $K$ and its estimated parameters, it is possible to determine the most likely path of the conditional Gaussian process under both linear and monotonicity constraints. This most likely curve corresponds to the mode\footnote{The maximum of the probability density function.} of the finite-dimensional truncated Gaussian vector $\boldsymbol{\xi}$ \citep[see also][for a discussion about the mode]{abrahamsen2001kriging}.
In bayesian statistics, it is known as the Maximum A Posteriori (MAP) estimator \citep[see][for more details]{maatouk:hal-01096751}. Its  expression is given by~:
\begin{equation}\label{propEstimator}
M^N_{K}\left(x \suchthat A,\boldsymbol{b}\right)=\nu+\sum_{j=0}^N\nu_j\phi_j(x),
\end{equation}
where $\boldsymbol{\nu}=(\nu,\nu_0,\ldots,\nu_N)^\top\in\mathbb{R}^{N+2}$ is the solution of the following convex optimization problem~:
\begin{equation}\label{eqNU}
\boldsymbol{\nu}=\arg\min_{\boldsymbol{c}\in \Cxi\cap \Ixi(A,\boldsymbol{b})}\left(\frac{1}{2}\boldsymbol{c}^\top\left(\Gamma^N\right)^{-1}\boldsymbol{c}\right),
\end{equation}
and where $\Gamma^N$ is the covariance matrix of the Gaussian vector $\boldsymbol{\xi}$ defined in \eqref{eq:cov_mat_one_dim}. The vector $\boldsymbol{\nu}$ can be seen as the mode of the Gaussian vector $\boldsymbol{\xi}$ restricted to $\Cxi\cap \Ixi(A,\boldsymbol{b})$, where $\Cxi$ is the set of vectors satisfying  monotonicity constraints and $\Ixi(A,\boldsymbol{b}) = \left\{\boldsymbol{\xi}\in\mathbb{R}^{N+2} \ : \ A\cdot \Phi\cdot\boldsymbol{\xi}=\boldsymbol{b}\right\}$ is the set of vectors which are compatible with the linear equality constraint.
Obviously, this curve satisfies  both constraints. 
Additionally, it does not depend on the variance  $\sigma^2$ of the  covariance function $K$ since $\sigma^2$ is a multiplicative constant in the matrix $\Gamma^N$ and then does not affect the $\arg\min$ in Equation~\eqref{eqNU}. In \cite{bay:hal-01136466} and \cite{2016arXiv160202714B}, the convergence of  the proposed estimator~(\ref{propEstimator})  as $N$ tends to infinity is studied and its limit corresponds to a constrained spline function (which depends on the underlying kernel function). By this methodology, one can thus retrieve classical spline interpolation with the additional possibility of getting confidence intervals.

\subsection{Two dimensional case}\label{TDC}

As explained in Remark~\ref{rem:surface}, the term-structure construction can be stated in a two-dimensional setting. The aim is to incorporate in the curve construction process  market information at several quotation dates. Contrary to the previous section, the output is now a surface which may represent the evolution of a reference quantity with respect to time-to-maturities and quotation dates. More formally, we consider a two-dimensional input variable $\boldsymbol{x}=(x,t)$ which is assumed to belong to a rectangle $\Dom = [\underline{x}, \overline{x}]\times [\underline{t}, \overline{t}]$ of $\R^2$. The variable $x$ may represent time-to-maturities whereas the variable $t$ may represent evolution of time or quotation dates. For the sake of simplicity, we assume that the unknown bivariate real function $f$ is monotone, say non-increasing, with respect to the first input variable only~:
\begin{equation}
\label{eq:bivariate_monotone_constraint}
x_{a}\leq x_{b} \quad \Rightarrow \quad f(x_{b},t)\leq f(x_{a},t), \qquad  \mbox{ for all }  t \in[\underline{t}, \overline{t}],\;  x_{a}, x_{b} \in[\underline{x}, \overline{x}].  
\end{equation}

More general inequality conditions can be considered in this bivariate setting \citep[see][for more details]{maatouk:hal-01096751}. The aim of this section is to explain how to construct a process that simultaneously  satisfies a series of linear equality constraints (one for each considered quotation date) and a monotonicity constraint as described in~\eqref{eq:bivariate_monotone_constraint}. As in the one-dimensional setting, we start with an original bivariate Gaussian process $Y$ with zero-mean and with a covariance function $K$. 
\\

The idea is the same as the one dimensional case presented in Section~\ref{ODC}. We begin by discretizing the input rectangle $\Dom$ in a $(N_x+1)\times (N_t+1)$ grid which for simplicity is assumed to be regular. 
The subdivision of the $x$-axis is $u_0<\ldots < u_{N_x}$ with $u_0=\underline{x}$, $u_{N_{x}}=\overline{x}$ and a constant mesh $\delta_{x}$, so that $u_i=u_0+i\delta_{x}$, $i=0, \ldots, N_{x}$. The subdivision of the $y$-axis is $v_0<\ldots < v_{N_t}$ with $v_0=\underline{t}$, $v_{N_{t}}=\overline{t}$ and a constant mesh $\delta_{t}$, so that $v_j=v_0+j\delta_{t}$, $j=0, \ldots, N_{t}$.
The following developments  can be easily extended to irregular grids. As in the one-dimensional case, we consider an original Gaussian  process $Y$ with covariance function $K$. We then  define the finite-dimensional approximation of $Y$ as the process $Y^N$ such that 
\begin{equation}\label{2Dmodel}
Y^N(x,t)=\sum_{i=0}^{\Nx}\sum_{j=0}^{\Nt}\xi_{i,j}g_i(x)h_j(t), \qquad  \mbox{ for all }  (x,t)\in\Dom,
\end{equation} 
where $g_i(x)=\max\left(1-\frac{|x-u_i|}{\delta_x},0\right)$ and $h_j(t)=\max\left(1-\frac{|t-v_j|}{\delta_t},0\right)$ are  hat functions centered at the knots $u_i$ and $v_j$, for $i=0, \ldots, N_x$ and $j=0, \ldots, N_t$ and $\boldsymbol{\xi}=(\xi_{0,0},\xi_{0,1},\ldots,\xi_{i,j},\ldots,\xi_{\Nx,\Nt})^\top$ is a zero-mean Gaussian vector with components $(\boldsymbol{\xi})_{\rho_{ij}}=\xi_{i,j}$ where $\rho_{ij}=(\Nt+1)i+j+1$, $i=0,\ldots,N_x$ and $j=0,\ldots,N_t$.  Let $N_{tot}=(N_x+1) (N_t+1)$ be the size of the column vector $\boldsymbol{\xi}$. Choosing  $\xi_{i,j} = Y(u_i,v_j)$, $i=0,\ldots,N_x$ and $j=0,\ldots,N_t$ guarantees that the finite-dimensional  process $Y^N$ converges on $\Dom$ towards $Y$ almost surely as $N_x$ and $N_t$ tends to infinity \citep[see][]{maatouk:hal-01096751}. In that case, the covariance matrix $\Gamma^N\in\mathbb{R}^{N_{tot}^2}$ of the Gaussian vector $\boldsymbol{\xi}$ can be written as~:
\begin{equation*}
\Gamma^N_{\rho_{ij},\rho_{i'j'}}=\cov(\xi_{i,j},\xi_{i',j'})= K\left((u_i,v_j),(u_{i'},v_{j'})\right),
\end{equation*}
where $i,i'=0,\ldots,\Nx$ and $j,j'=0,\ldots,\Nt$. As can be shown in the next proposition, the monotonicity constraint~\eqref{eq:bivariate_monotone_constraint} reduces to a linear inequality condition on the vector $\boldsymbol{\xi}$. 
\begin{pro}[Monotonicity]
\label{pro:monotone_2D}
The process $Y^N$ defined in Equation~\eqref{2Dmodel} is non-increasing (resp. non-decreasing) with respect to the first variable $x$ \textit{if and only if} all the coefficients $\xi_j, \ j=0,\ldots,N$ are such that
\begin{equation}\label{Dim2LinConst}
\xi_{i-1,j}\leq  \xi_{i,j},\; (\text{resp. } \xi_{i-1,j}\geq  \xi_{i,j}) \qquad i=1,\ldots,\Nx \ \mbox{and} \ j=0,\ldots,\Nt.
\end{equation}
\end{pro}
\begin{proof}
The proof is similar to the one of Proposition~\ref{pro:monotone_2D}.
\end{proof}

In order to construct surfaces which are compatible with market quotes observed at times  $t=t_1, \ldots, t_{I}$, $I$ different linear equality constraints like the one given in Section~\ref{ILC} has to be imposed on the bivariate process $Y^N$. Let us consider that the market-fit conditions involve $m$ points $x^{(1)},\ldots,x^{(m)}$ for any quotation time $t=t_1, \ldots, t_I$. Then, if $Y^N(X,t)$ $=$ $\left(Y^N\left(x^{(1)},t\right), \ldots,\right.$ $\left.Y^N\left(x^{(m)},t\right)\right)^\top$ denotes the vector of values involves in the curve construction at time $t$ and given Equation~\eqref{2Dmodel}, the condition $A_t\cdot Y^N(X,t)=\boldsymbol{b}_t$ translates into the following linear equality constraint on the Gaussian vector $\boldsymbol{\xi}$~:
\begin{equation}
\label{EqCond_2D}
A_t\cdot H_t\cdot \boldsymbol{\xi}=\boldsymbol{b}_t
\end{equation}
where the $m\times N_{tot}$ matrix $H_t$ has components $(H_t)_{k,\rho_{ij}}=g_i\left(x^{(k)}\right)h_j(t)$, $k=1, \ldots, m$ and $\rho_{ij}=(\Nt+1)i+j+1$, $i=0,\ldots,N_x$ and $j=0,\ldots,N_t$. The former condition has to hold simultaneously for every time $t=t_1, \ldots, t_I$, which can be summarized as one single linear equality constraint  
\begin{equation}
\label{EqCond_2D_bis}
B \cdot \boldsymbol{\xi} = \boldsymbol{b},
\end{equation}
where the $nI \times N_{tot}$ matrix $B$ (resp.  $\boldsymbol{b}$) is formed by  vertical concatenation of matrices $A_t \cdot H_t$ (resp. $\boldsymbol{b}_t$) for $t=t_1, \ldots, t_I$.
Notice that, generally speaking, this linear system  admits  solutions  when $N_{tot} \geq nI$.



\paragraph{Surface simulation.} 
As in the one-dimensional setting, the simulation of the bivariate Gaussian process $Y^N$ conditionally to both linear equality  and monotonicity constraints reduces to the simulation of the Gaussian vector $\boldsymbol{\xi}$ restricted to 
\begin{eqnarray*}\left\{\begin{array}{ll}
B\cdot \boldsymbol{\xi}=\boldsymbol{b} & \qquad \mbox{linear equality condition}\\
\boldsymbol{\xi}\in \Cxi & \qquad \mbox{monotonicity constraint}
\end{array}\right.
\end{eqnarray*}
where, for instance, $\Cxi=\left\{\boldsymbol{\xi}\in\mathbb{R}^{N_{tot}} \ : \ \xi_{i-1,j}\leq \xi_{i,j}\right\}$ for a non-increasing constraint. 
The simulation procedure is the same as the one presented in Section~\ref{ODC}. 
These simulations can be used to construct confidence intervals for the surface value at each point $(x,t)$ in $\Dom$.



\paragraph{Most likely surface.} The methodology can be easily extended from Section~\ref{ODC}.

\subsection{Parameters estimation}\label{PE}

The most likely path of the constrained process (mode estimator) 
depends  on the choice of the underlying Gaussian process or equivalently on its covariance function $K$. In this section, we investigate the estimation of its parameter, i.e., the length and the variance hyper parameters. 
In the literature,  estimation of the covariance function hyper parameters is usually done using two types of methods~: Maximum Likelihood (ML) estimators as in \cite{sant:will:notz:2003} and Cross Validation (CV) methods  as in \cite{Bachoc201355}, \cite{Cressie1993Wiley} and \cite{Roustant:Ginsbourger:Deville:2012:JSSOBK:v51i01}. Both methods ML and CV are not suited to monotonicity constraints.\\

Recently, an Adapted Cross-Validation (ACV) technique has be proposed by  \cite{Maatouk201538} to estimate covariance hyper-parameters of Gaussian processes in the presence of inequality constraints. The main idea is to consider the mode curve (as opposed to the mean curve) as the estimator to be used in the cross-validation method. 
The principle of cross-validation is to select the set of parameters that minimizes a distance between observed values and their estimates while successively omitting some set of observations. 
Thus, cross-validation is similar in spirit to backtesting \citep[see][]{kerkhof2004backtesting}, but omitted data are not necessarily taken in chronological order.

\paragraph{Length parameters.}

In classical kriging, the usual cross-validation estimator of the covariance  length parameter $\boldsymbol{\theta}$  is constructed from the so-called Leave One Out (LOO) mean square error criterion. Given that the unknown function $f$ takes values $y_1, \ldots, y_n$ at points $x^{(1)}, \ldots, x^{(n)}$ (pure interpolation constraints), the cross-validation estimator $\hat{\boldsymbol{\theta}}_{CV}$ of  $\boldsymbol{\theta}$ is defined as 
\begin{equation}
\hat{\boldsymbol{\theta}}_{CV}=\arg\min_{\boldsymbol{\theta}\in\Theta}\sum_{i=1}^n\left(y_i-\hat{y}_{i,\boldsymbol{\theta}}(\boldsymbol{y}_{-i})\right)^2,
\end{equation}
where $\boldsymbol{y}_{-i}=(y_1,\ldots,y_{i-1},y_{i+1},\ldots,y_n)^\top$ and $\Theta$ is a compact subset of $\mathbb{R}^d$. The estimator 
$\hat{y}_{i,\boldsymbol{\theta}}(\boldsymbol{y}_{-i})=\mathbb{E}\left(Y(x^{(i)})\suchthat Y(X^{(-i)})=\boldsymbol{y}_{-i}\right)$ is the kriging mean at point  $x^{(i)}$ obtained by removing observation $y_i$ in the estimation process (see Equation \eqref{eq:kriging_mean}). The vector $Y\left(X^{(-i)}\right)$ is the same vector as $Y(X)$ without component $Y\left(x^{(i)}\right)$.\\

%


%
 
The previous criterion cannot be used in the presence of monotonicity constraints since 
the classical kriging mean estimate does not respect such kind of constraints.  As suggested in \cite{Maatouk201538} under pure interpolation constraints,  the most likely curve (mode estimator) defined in Section \ref{ODC} can be used instead of the kriging mean since the former satisfies the monotonicity constraints. Then, a new LOO criterion adapted to monotonicity constraints can be defined as~:
\begin{equation}\label{LOO}
\hat{\boldsymbol{\theta}}_{ACV}=\arg\min_{\boldsymbol{\theta}\in\Theta}\sum_{i=1}^n\left(y_i-M^N_{K}\left(x^{(i)} \suchthat I_{n-1},\boldsymbol{y}_{-i}\right)\right)^2,
\end{equation} 
where $M^N_{K}\left(x^{(i)} \suchthat I_{n-1},\boldsymbol{y}_{-i}\right)$ is the mode estimator (defined in Equation~(\ref{propEstimator})) of $y_i$ based on all observations but $y_i$. 
As explained in  Section \ref{ODC}, the latter does not depend on the variance $\sigma^2$. The matrix
$I_{n-1}$ is the identity matrix with dimension $n-1$ (the same as the vertical dimension as the column vector $\boldsymbol{y}_{-i}$). This methodology cannot be applied immediately in our setting since the output values $y_i$ in the LOO criterion (\ref{LOO})
are not available under linear equality constraints as defined in (\ref{EqCond}). To this end, we consider the following formulation of the LOO criterion~:
\begin{equation}\label{LOOLConst}
\hat{\boldsymbol{\theta}}_{ACV}=\arg\min_{\boldsymbol{\theta}\in\Theta}\sum_{i=1}^{n}\left(b_i-\left(A\cdot M^N_{K}\left(X \suchthat A_{-i},\boldsymbol{b}_{-i}\right)\right)_i\right)^2,
\end{equation}
where the subscript $i$ refers to the $i$-th component, $A_{-i}$ and $\boldsymbol{b}_{-i}$ are respectively the matrix and the vector without the $i^{\text{th}}$ row.\\

\paragraph{Variance parameter.}

In the case of pure interpolation constraints and without monotonicity constraints, several classical criterions can be used for the estimation of $\sigma$. For instance, in \cite{Bachoc201355}, an estimator $\hat{\sigma}_{CV}$ is defined as the parameter $\sigma$ such that the following equality holds 
\begin{equation}\label{sigSCV}
\frac{1}{n}\sum_{i=1}^n\frac{\left(y_i-\hat{y}_{i,\hat{\boldsymbol{\theta}}_{CV}}(\boldsymbol{y}_{-i})\right)^2}{\mathds{E}\left(\left(Y\left(x^{(i)}\right)-\hat{y}_{i,\hat{\boldsymbol{\theta}}_{CV}}(\boldsymbol{y}_{-i})\right)^2\suchthat Y(X^{(-i)})=\boldsymbol{y}_{-i}\right)}=1.
\end{equation}
Given our linear equality conditions and monotonicity constraints, we instead propose to define $\hat{\sigma}_{ACV}$ as the parameter $\sigma$ such that the following equality holds
\begin{equation}\label{sigSCV_2}
\frac{1}{n}\sum_{i=1}^n\frac{\left(b_i-\left(A\cdot M^N_{K}\left(X \suchthat A_{-i},\boldsymbol{b}_{-i}\right)\right)_i\right)^2}{\mathds{E}\left(\left(AY\left(X\right)-A M^N_{K}\left(X\suchthat A_{-i},\boldsymbol{b}_{-i}\right)\right)_i^2 \suchthat \mathcal{D}_i \right)}=1,
\end{equation}
where $\mathcal{D}_i$ is the set of monotonicity constraints on the whole domain and the additional linear constraints without the $i^{\text{th}}$ one, i.e. $A_{-i} Y(X^{(-i)}) = \boldsymbol{b}_{-i}$ where $A_{-i}$ is the matrix $A$ without the $i$-th row.
The expectation is estimated by simulation, approximating the process $Y$ by $Y^N$.



\section{Empirical investigation}\label{EmpInv}

In this section, the construction method developed in Section \ref{IntModIConst} is illustrated in different financial applications. Based on  market quotes observed at different  dates, we construct curves together with confidence intervals for term-structures of OIS discount rates, term-structure of zero-coupon swaps rates and term-structure of CDS implied default probabilities. Using the bivariate setting of Section \ref{TDC}, we also build interest-rate and default probability surfaces by considering time (quotation dates) as an additional dimension.



\subsection{Interest-rate curves based on Swaps versus Euribor}\label{CD}

We apply the kriging procedure to construct zero-coupon swap curves  based on market quotes of fixed-vs-floating interest-rate swaps   for different standard maturities.
Market observations are given as par rates of Swaps vs Euribor 6M. We consider the 10 quotation dates given in Table~\ref{EPOV}. For each quotation date, the term-structure is built  on $14$ swap rates $S_1, \ldots, S_{14}$ associated with standard maturities in years belonging to the set $E:=\{1,\ldots,10,15,20,30,40\}^\top$. As explained in Section~\ref{Sec:term-structure}, each observed par rate provides an indirect  information on the curve. This information takes the
  form of a linear relation given by~\eqref{eq:OIS}. For each standard maturity $T\in E$, this relation involves the value $P(t,k)$ of the curve at time horizon $k=1, \ldots, T$.
As a result, the curves are compatible with observed quotes if, for each observation date $t$, the vector of discount factors $P(t, X) := (P(t, 1), \ldots, P(t,40))^{\top}$ satisfies a linear system of the form 
\begin{equation}
\label{eq:market_fit_Swaps_Euribor}
A_t \cdot P(t, X)=\boldsymbol{b}_t,
\end{equation}
where $A_t$ is  a $14\times 40$ real matrix and $\boldsymbol{b}_t= (1,\ldots,1)^\top\in \mathbb{R}^{14}$. In this case, we have $n=14$ products whose value depends on $m=40$ points of the curve.\\

We consider that the associated discount factor curve $P(t,X)$ is one realization of a decreasing spatial process which starts from $1$ ($P(t,0)=1$) and satisfies the linear condition~\eqref{eq:market_fit_Swaps_Euribor}.  Section~\ref{IntModIConst} explains how to construct and simulate a  process with monotonicity and linear equality constraints. This construction involves a finite-dimensional approximation $Y^N$ of Gaussian processes as defined in~\eqref{propModel}.  The latter depends on a $N+2$-dimensional Gaussian vector $\boldsymbol{\xi}$ and  a set of basis functions $\phi$ defined on a subdivision of the input domain $D=[0,40]$. We consider here a regular subdivision $u_j=u_0+j\delta$, $j=0, \ldots, N$, where $u_0=0$, $N=50$ sub-intervals and $\delta=1$.
Since the curve $T\rightarrow P(t,T)$ is known to start from 1, i.e., $P(t,0)=1$, the linear equality condition on $\boldsymbol{\xi}$ defined in (\ref{EqCond}) can be reformulated as follows~:
\begin{equation}\label{EqXiCond}
\left(\begin{matrix}
1 & \phi_0(0) & \ldots & \phi_N(0)\\
& & A_t\cdot \Phi& &  
\end{matrix}
\right)\boldsymbol{\xi}=\left(\begin{matrix}
1\\
\boldsymbol{b}_t
\end{matrix}\right).
\end{equation}
Then the simulation of the GP $Y^N$ conditionally to the linear equality condition and non-increasing constraints is equivalent to generating a truncated zero-mean Gaussian vector $\boldsymbol{\xi}$ restricted to Equation~(\ref{EqXiCond}) and to non-positive components.

\paragraph{Parameters estimation.}
We first consider the correlation-length hyper parameter $\theta$. It has been estimated for each quotation dates in Table~\ref{EPOV}  by using the adapted cross-validation (ACV) method described in Section \ref{PE}. We consider covariance functions of the form $K(x,x')=\sigma^2 C(x-x',\theta)$ and we discuss two alternative kernels $C$, that is the Gaussian and the Mat\'ern 5/2 kernel (see Table~\ref{kernel}).
We denote  by $\hat{\theta}_G$ (resp. $\hat{\theta}_M$) the estimated length parameter of the Gaussian (resp. Mat\'ern 5/2) covariance function. 
As shown in Table~\ref{EPOV}, the estimated length parameter $\hat{\theta}$ remains stable for both Gaussian and Mat\'ern 5/2 covariance functions (value around $25$ for $\hat{\theta}_G$ and $30$ for $\hat{\theta}_M$). Note also that the minimal value of the objective function in the LOO criterion \eqref{LOOLConst} is slightly smaller when using a Gaussian covariance function. However, as can be seen in Figure~\ref{LOOMatGauss}, the objective function (cross-validation error)  reach $0.12$ for the Gaussian covariance function whereas it never goes above $0.02$ for the Mat\'ern 5/2 covariance function. Additionally, the global minimum is easier to find in the Mat\'ern 5/2 case.

\begin{table}[hptb]
\centering
\caption{Parameters estimation using ACV methods (Swap versus Euribor 6M).}
\label{EPOV}
\begin{tabular}{ccccc}
\hline \hline
Date  &  $\hat{\theta}_G$ &  $\hat{\theta}_M$ & Gaussian optimal value & Mat\'ern 5/2 optimal value  \\   
\hline
02/06/2010   &  25.8   &  30.8  & 4.0e-06   &  \textbf{1.1e-06}   \\
05/07/2010   &  26.2   &  25.2  & \textbf{3.8e-06}   &  1.1e-05   \\
03/08/2010   &  26.6   &  27.0  & \textbf{4.1e-06}   &  4.7e-06   \\
29/11/2010   &  23.5   &  39.9  & \textbf{5.4e-07}   &  2.3e-05   \\ 
30/12/2010   &  27.5   &  28.0  & 4.5e-06   &  \textbf{1.5e-06}   \\
31/01/2011   &  29.2   &  29.2  & 9.2e-06   &  \textbf{1.2e-06}   \\
10/05/2011   &  28.8   &  27.6  & 2.5e-06   &  \textbf{1.2e-06}   \\
10/06/2011   &  26.0   &  30.7  & \textbf{5.0e-07}   &  2.8e-06   \\
30/12/2011   &  20.3   &  30.0  & 6.8e-06   &  \textbf{3.6e-06}   \\
\hline
\end{tabular}
\end{table}

\begin{figure}[hptb]
\begin{minipage}{.5\linewidth}
\centering
\includegraphics[scale=.4]{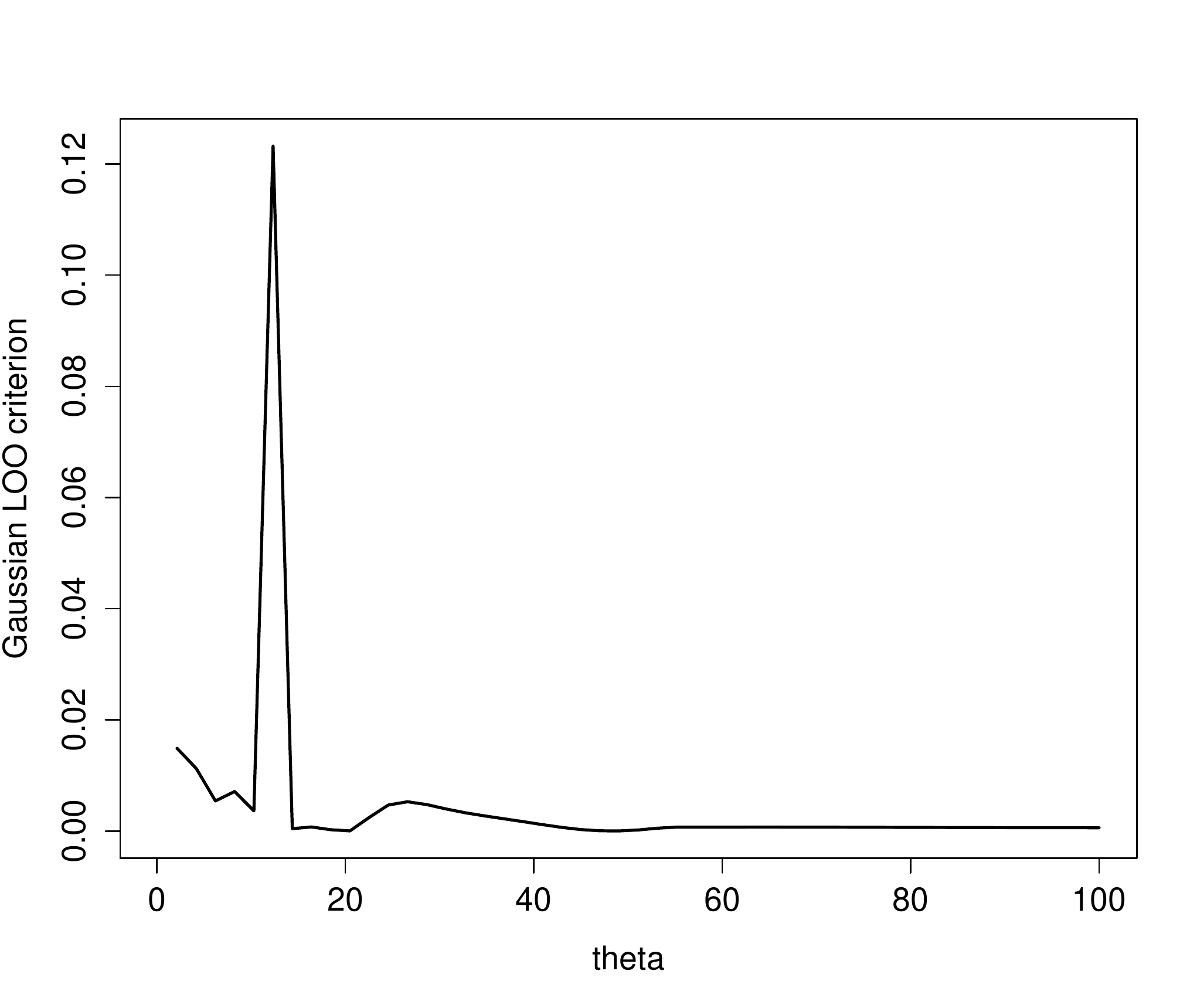}
\end{minipage}%
\begin{minipage}{.5\linewidth}
\centering
\includegraphics[scale=.4]{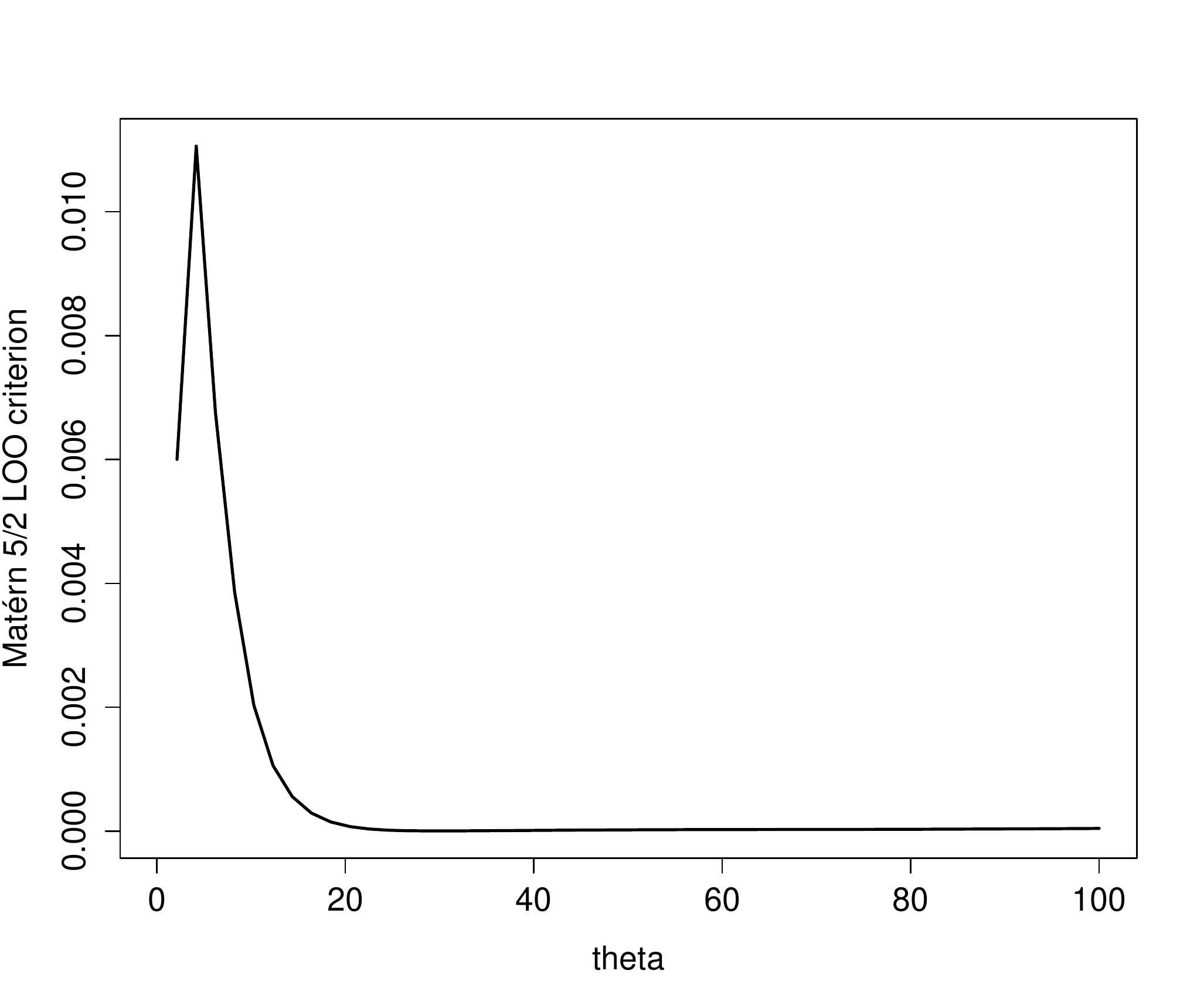}
\end{minipage}
\caption{The function to be optimized in the LOO criterion (\ref{LOOLConst}) using the Gaussian (left) and the Mat\'ern 5/2 covariance function (right). Swap versus Euribor 6M on $30/12/2011$.}
\label{LOOMatGauss}
\end{figure}

Once the length parameter $\theta$ has been estimated, the standard deviation parameter $\sigma$ is estimated using Equation~\eqref{sigSCV_2}. For instance, at quotation date $30/12/2011$, we obtain for respective Gaussian and Mat\'ern 5/2 covariance kernel $\hat{\sigma}_G=2.89$ and $\hat{\sigma}_M=0.93$.

\paragraph{Curve construction at a single quotation date.}
We now illustrate the curve construction method described in Section \ref{ODC} in a one-dimensional setting. The construction is based on market quotes as of 
$30/12/2011$. For this particular date, the estimated length parameter for the Gaussian and the Mat\'ern 5/2 kernel are given in Table~\ref{EPOV} ($\hat{\theta}_G=20.3$ and $\hat{\theta}_M=30.0$). In that case, using the method described in Section~\ref{PE}, the estimated variance parameter is equal to $\hat{\sigma}_G=2.89$ for the Gaussian kernel and to $\hat{\sigma}_M=0.93$ for the Mat\'ern 5/2 kernel. 
Figure~\ref{GaussMatParEst} compares  the sample paths of discount factors for the Gaussian and the Mat\'ern 5/2 covariance function using the corresponding estimated parameters. In both cases, we generate 100 sample paths taken from model~\eqref{propModel} conditionally  to  linear equality constraints~\eqref{EqXiCond} and non-increasing constraints. 
Note that the simulated curves (gray lines) are non-increasing in the entire domain. 
Additionally, the black solid line represents the most likely curve, i.e., the mode of the conditional GP. Recall that, by construction, this curve satisfies the given constraints. The black dashed-lines represent the 95\% point-wise confidence intervals quantified by simulation. 
Figure~\ref{SRGaussMatEURO} and Figure~\ref{TFGaussMatEURO} give the corresponding spot rates and instantaneous forward curves.\\

In order to compare our results with some models commonly used in most central banks, all figures are given together with the associated best-fitted Nelson-Siegel curves  \citep[see][]{nelson1987parsimonious} and the associated best-fitted Svensson curves \citep[see][]{Svensson:1994}. Parameters have been estimated by minimizing the sum of squared errors between market and model prices. We use a gradient descent algorithm with randomly chosen starting values as described in \cite{gilli2010calibrating}. The optimal parameters are given in Table~\ref{ParamNS_SWAP}. Discount factors and forward rates have been deduced from Nelson-Siegel and Svensson yield curves.\\

\begin{table}[hptb]
\centering
\caption{Parameters estimation for Nelson-Siegel and Nelson-Siegel-Svensson model (Swap versus Euribor 6M on 30/12/2011).}
\label{ParamNS_SWAP}
\begin{tabular}{l|cccccc}
\hline \hline
 & $\lambda_1$ & $\lambda_2$ & $\beta_1$ & $\beta_2$ & $\beta_3$ & $\beta_4$ \\ 
\hline 
Nelson-Siegel          & 7.4615 &  - & 0.0189 &  -0.0160 & 0.0487 & - \\ \hline 
Nelson-Siegel-Svensson & 4.0486 & 28.4285 & 0.1719 & -0.1590 & -0.1101 & -0.4093 \\ 
\hline 
\end{tabular}
\end{table}

\begin{figure}[H]
\begin{minipage}{.5\linewidth}
\centering
\includegraphics[scale=.37]{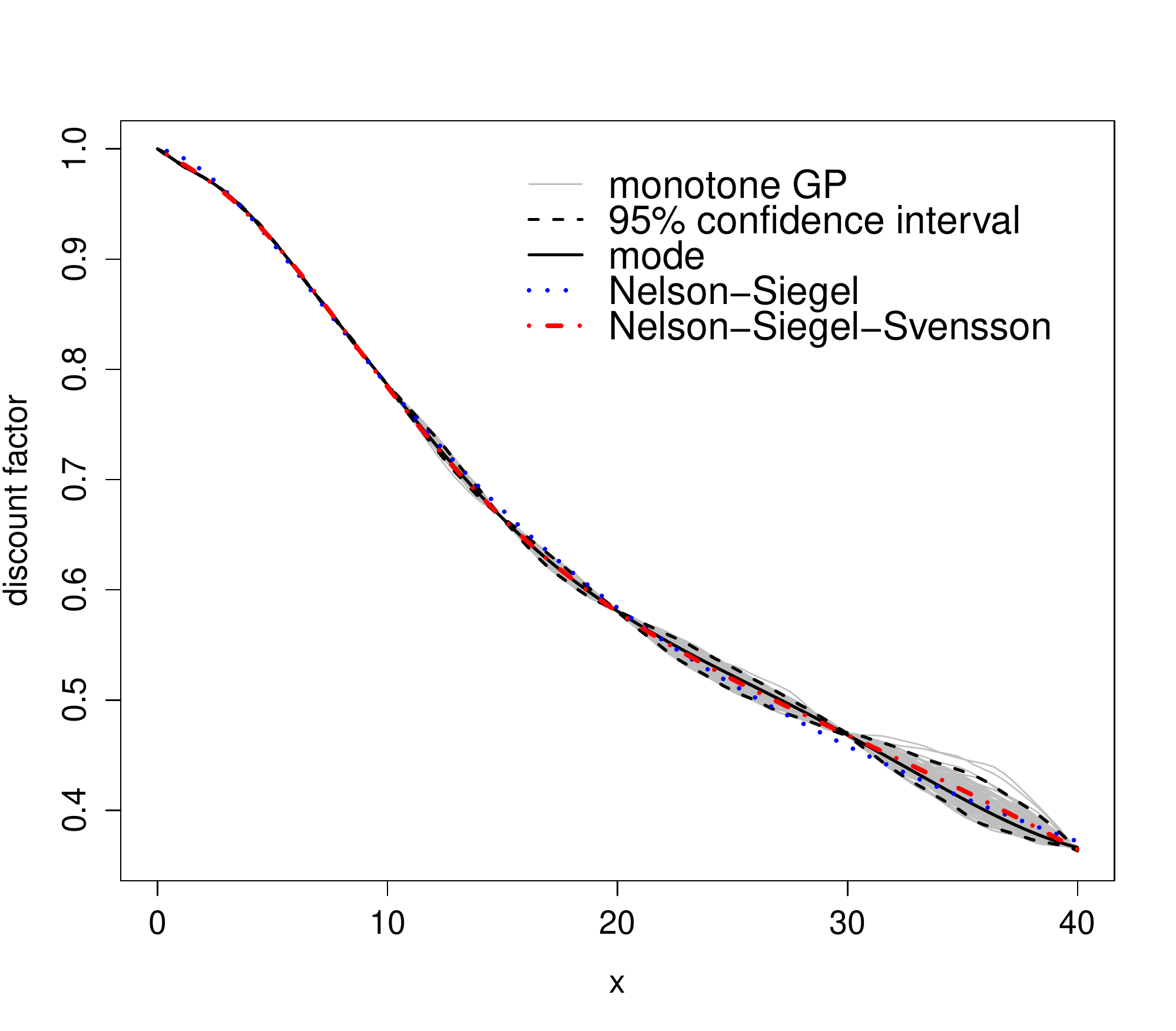}
\end{minipage}%
\begin{minipage}{.5\linewidth}
\centering
\includegraphics[scale=.37]{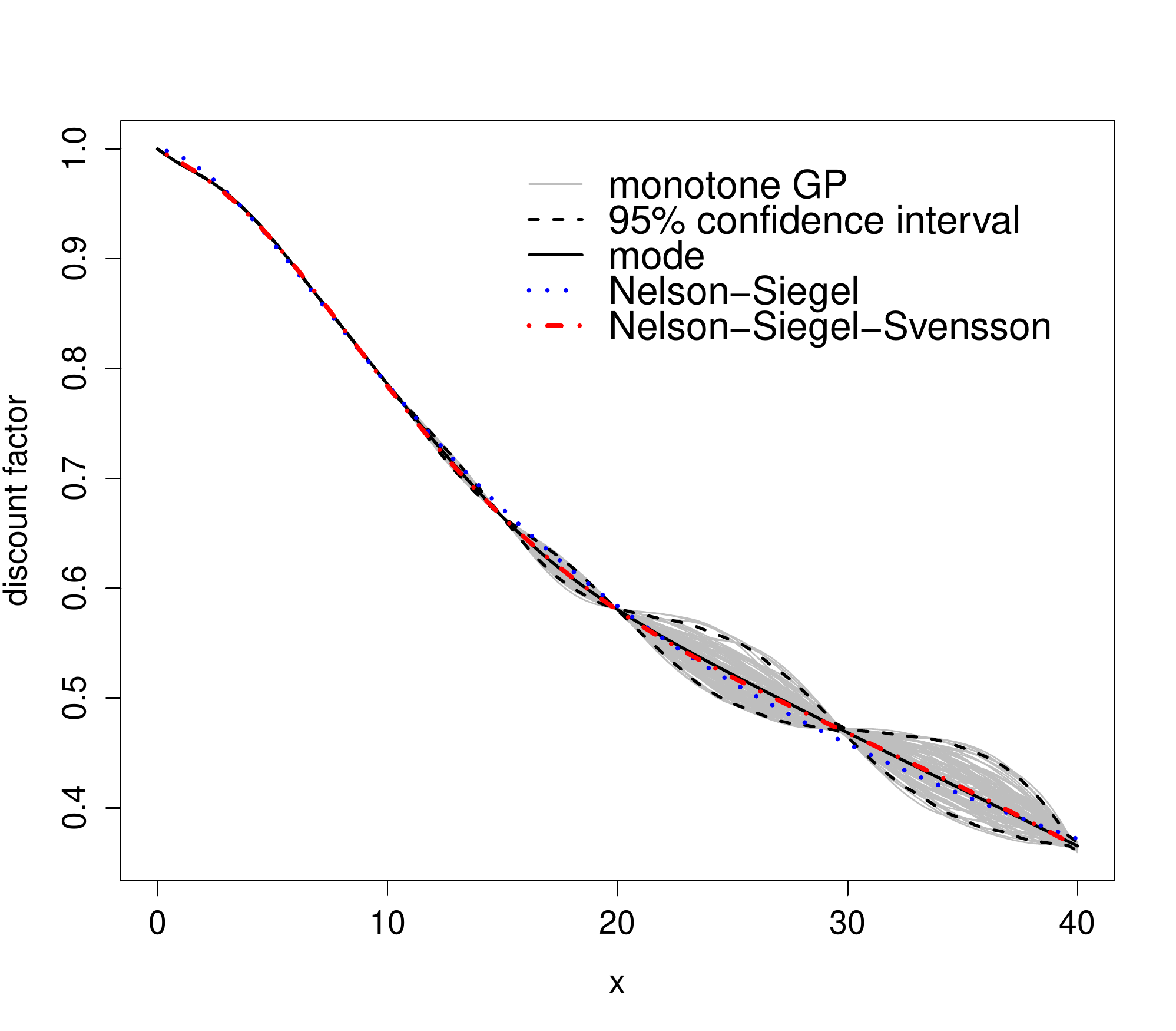}
\end{minipage}
\caption{Simulated paths (gray lines) taken from the conditional GP with non-decreasing constraints and market-fit constraints using the Gaussian covariance function with nugget equal to $10^{-5}$ (left) and the Mat\'ern 5/2 covariance function without nugget (right).
Swap vs Euribor 6M market quotes as of $30/12/2011$.}
\label{GaussMatParEst}
\end{figure}

\begin{figure}[H]
\begin{minipage}{.5\linewidth}
\centering
\includegraphics[scale=.37]{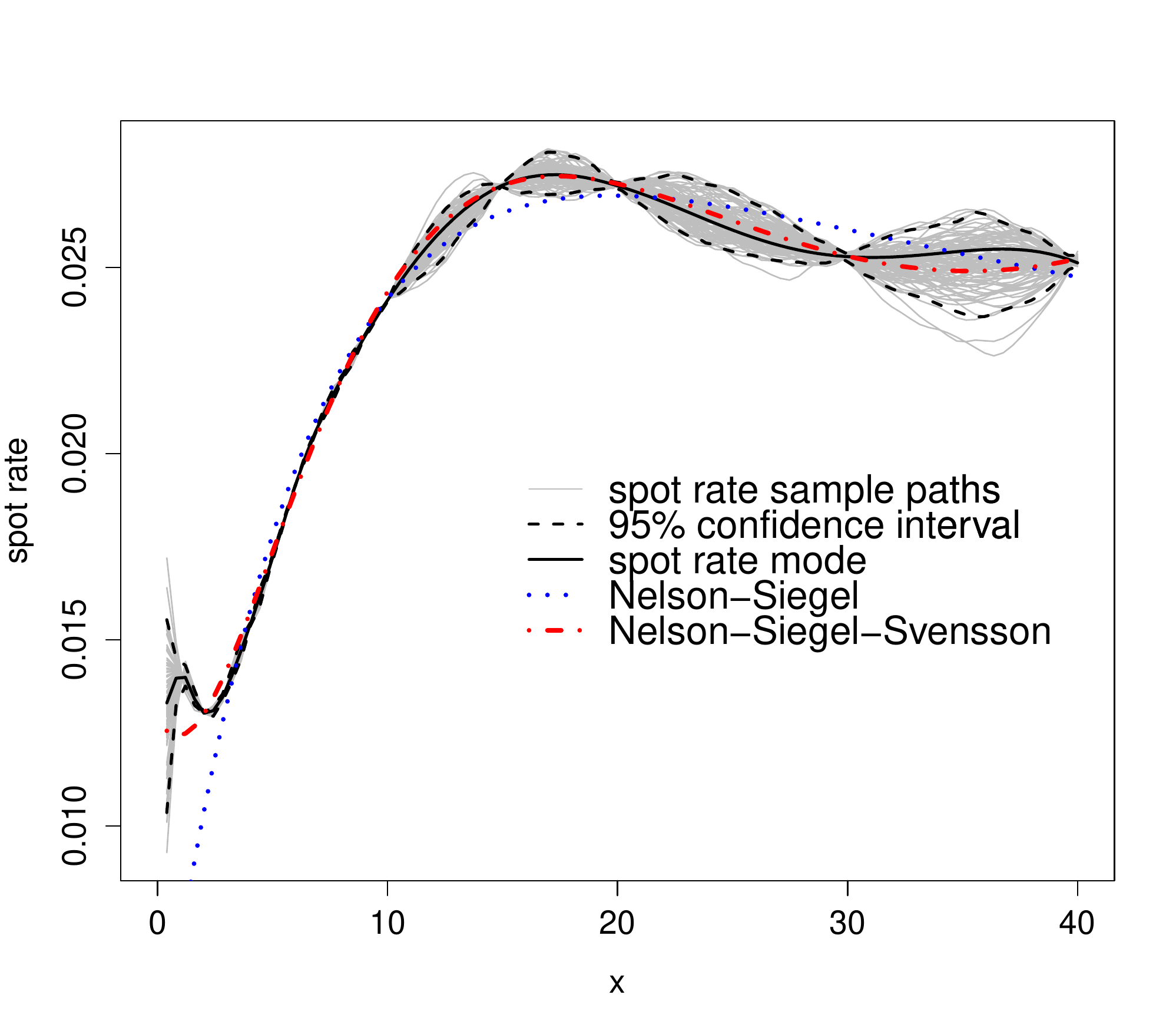}
\end{minipage}%
\begin{minipage}{.5\linewidth}
\centering
\includegraphics[scale=.37]{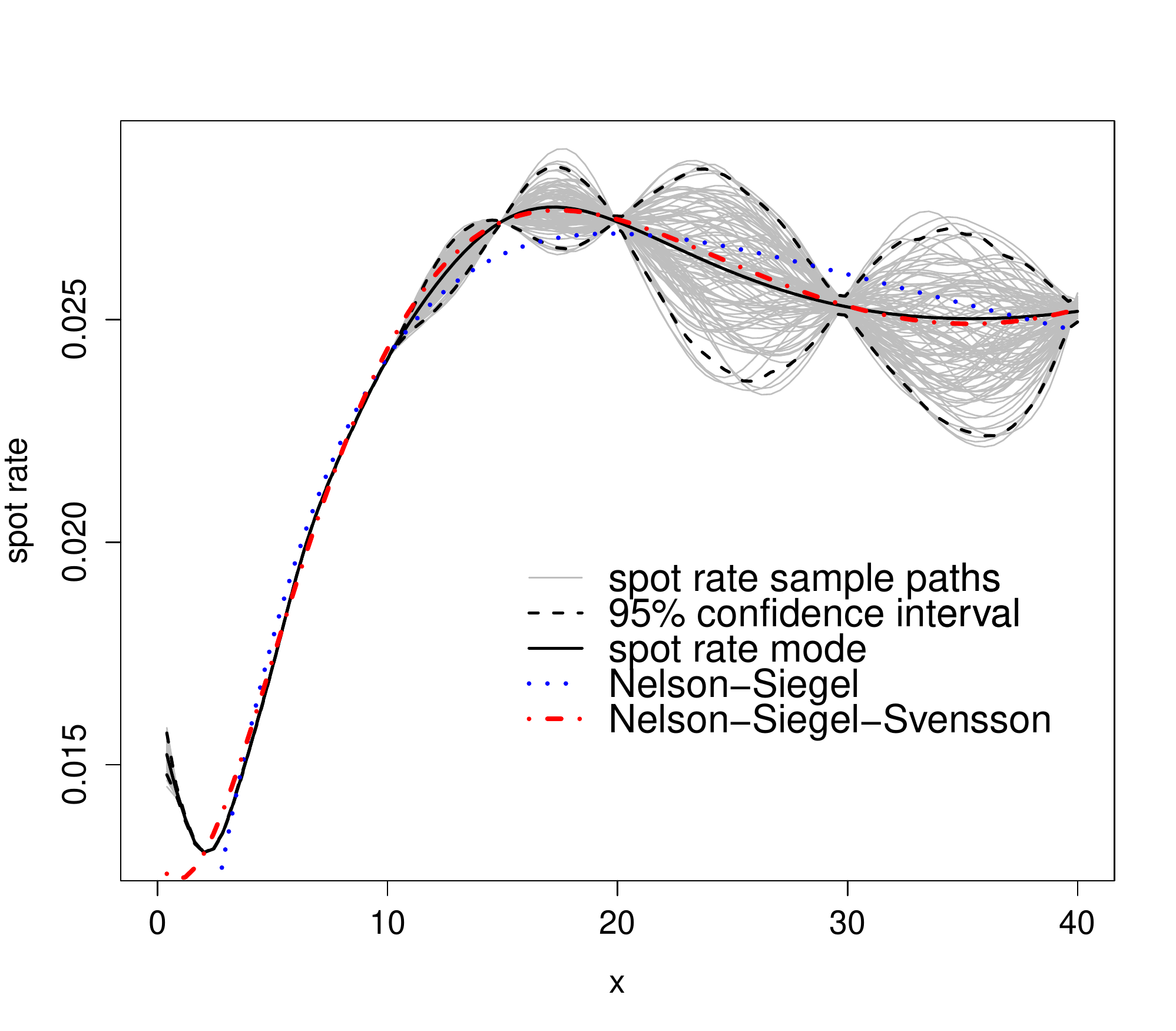}
\end{minipage}
\caption{Spot rates obtained from sample paths of Figure \ref{GaussMatParEst} with Gaussian covariance function (left) and Mat\'ern 5/2 covariance function (right). Gray lines represent $-\frac{1}{x}\log Y^N(x)$ for each sample path. The black solid line is the most likely spot rate curve $-\frac{1}{x}\log M^N_K\left(x \suchthat A,\boldsymbol{b}\right)$.}
\label{SRGaussMatEURO}
\end{figure}

\begin{figure}[H]
\begin{minipage}{.5\linewidth}
\centering
\includegraphics[scale=.37]{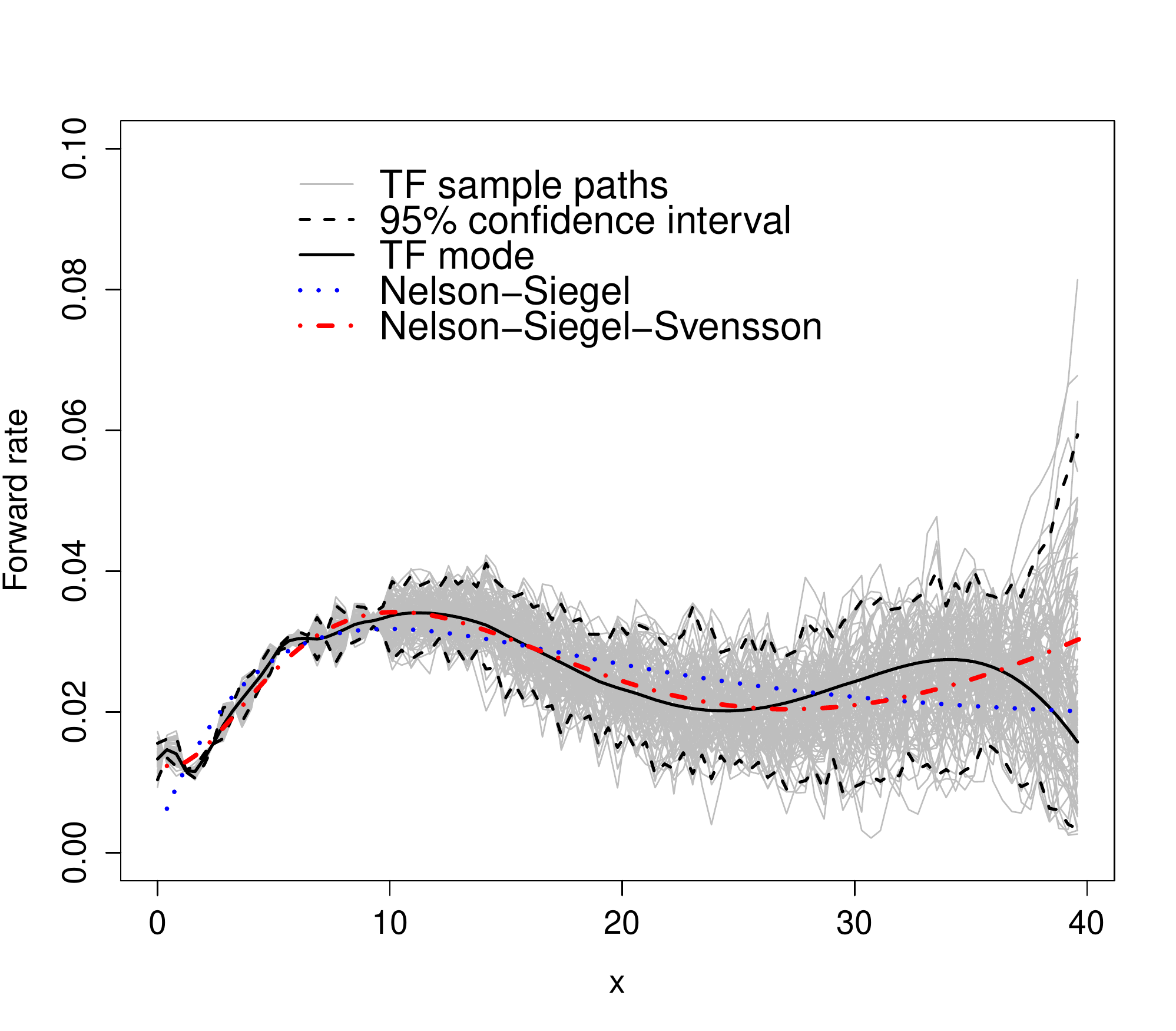}
\end{minipage}%
\begin{minipage}{.5\linewidth}
\centering
\includegraphics[scale=.37]{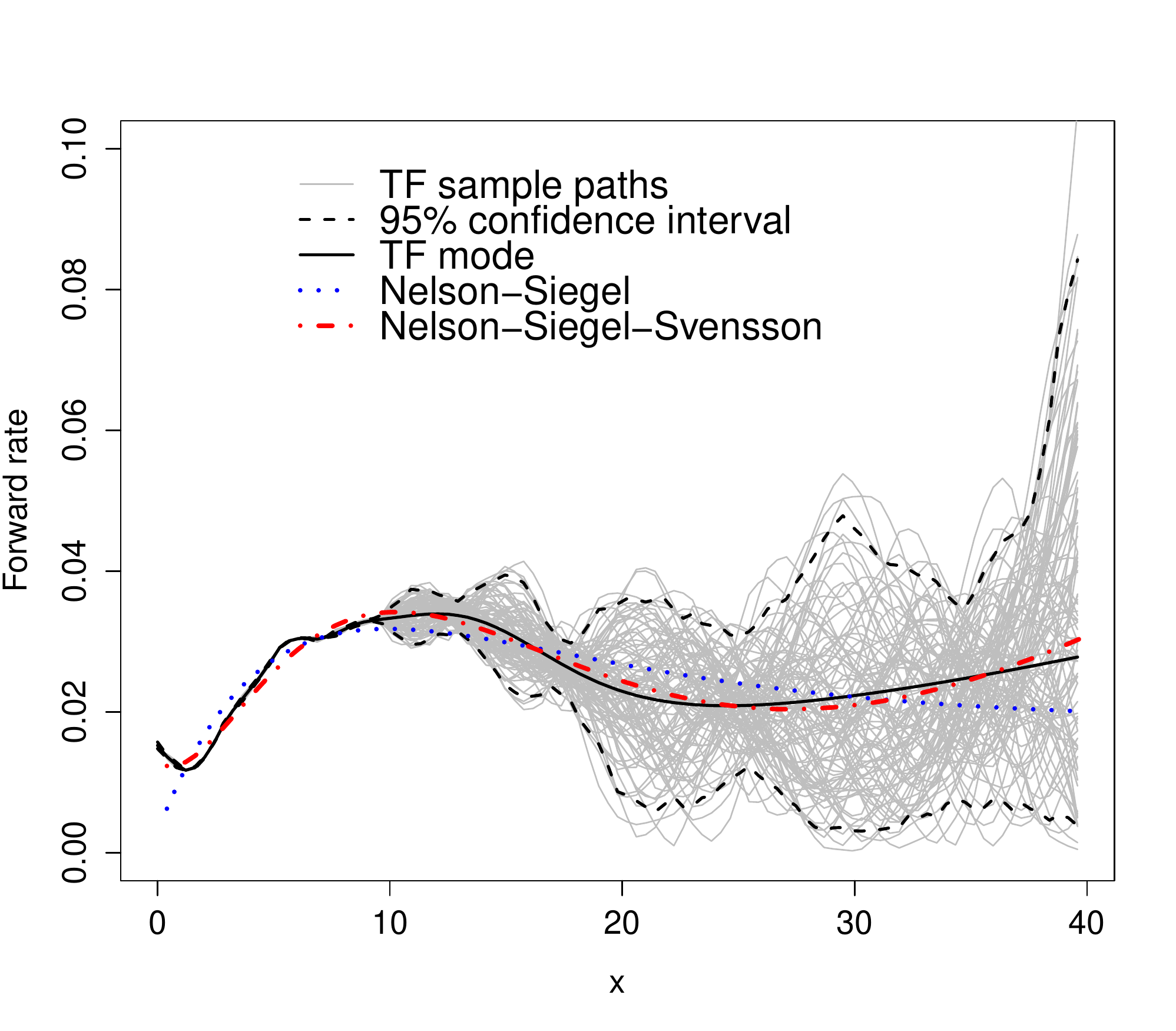}
\end{minipage}
\caption{Forward rates obtained from sample paths of Figure \ref{GaussMatParEst}  with Gaussian covariance function (left) and Mat\'ern 5/2 covariance function (right). Gray lines represent $-\frac{d}{dx}\log Y^N(x)$ for each sample path. The black solid line is the most likely forward rate curve $-\frac{d}{dx}\log M^N_K\left(x\suchthat A,\boldsymbol{b}\right)$.} 
\label{TFGaussMatEURO}
\end{figure}

\begin{figure}[hptb]
\begin{minipage}{.5\linewidth}
\centering
\includegraphics[trim={0 5cm 0 0}, scale=.37]{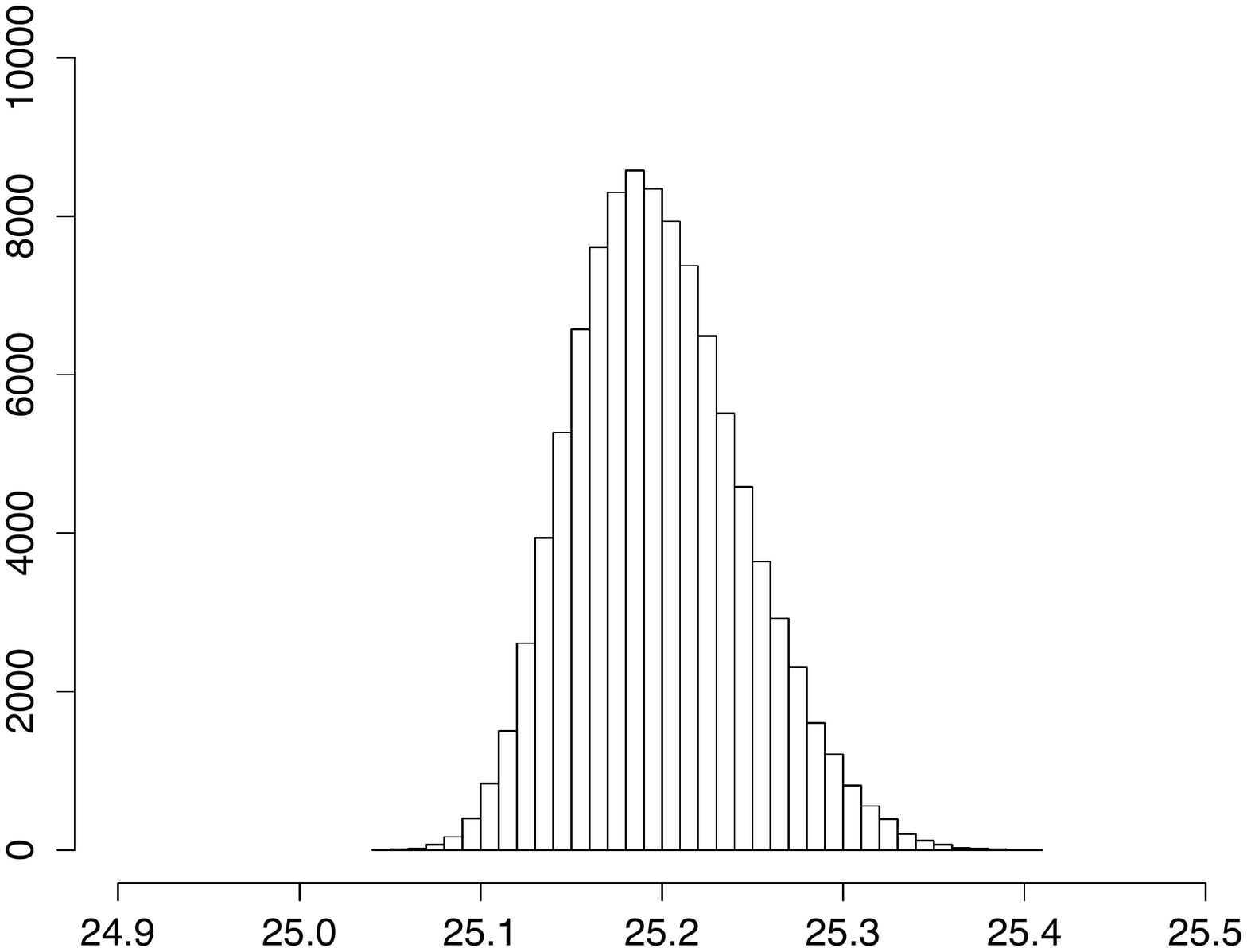}
\end{minipage}%
\begin{minipage}{.5\linewidth}
\centering
\includegraphics[trim={0 5cm 0 0}, scale=.37]{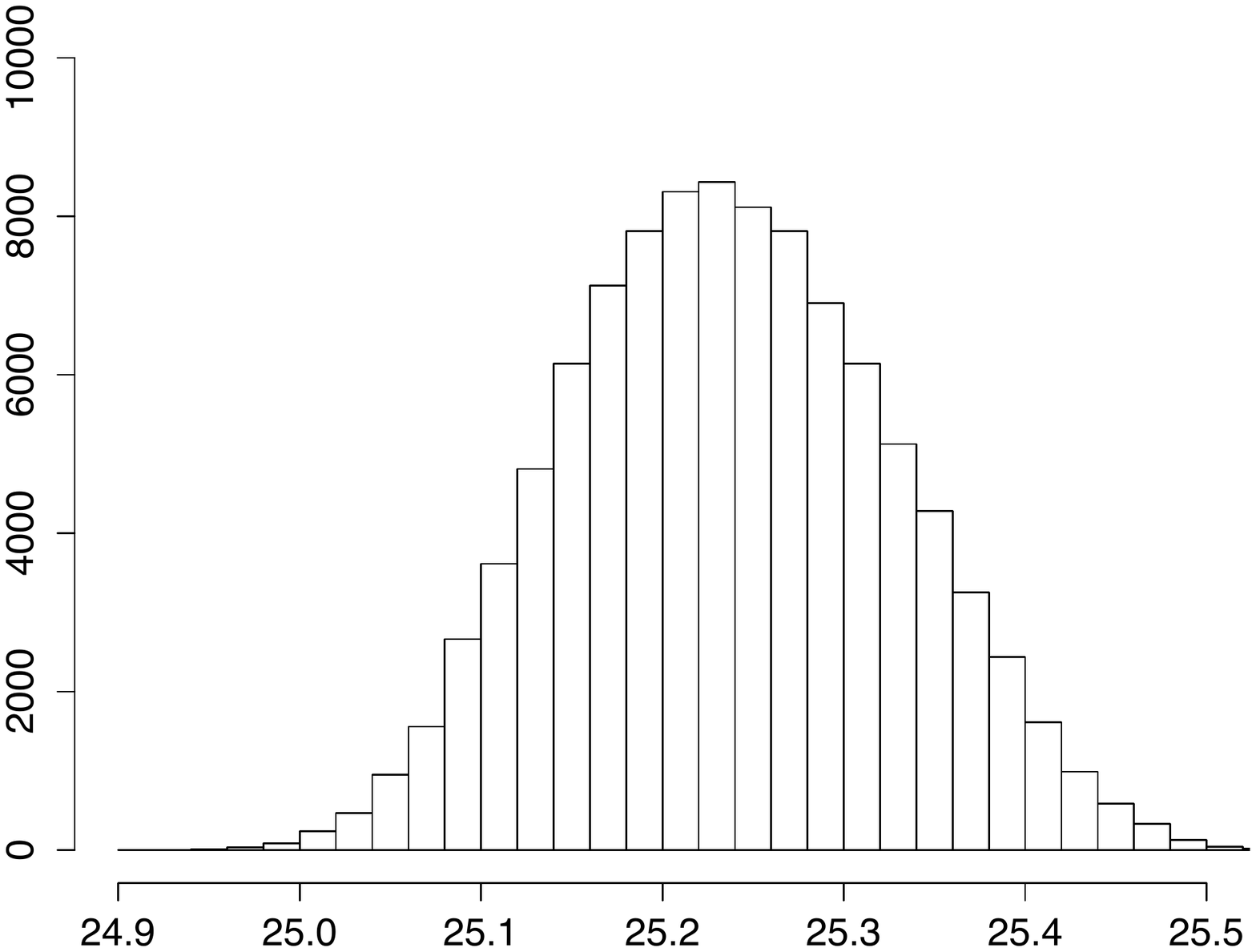}
\end{minipage}
\caption{ Histogram of a periodic annuity-due present values constructed from 100 000 independent simulations of discount factor curves under a Gaussian covariance function (left) and under a  Matérn covariance function (right). 
}
\label{renteGaussMatSwap}
\end{figure}

It can be seen on Figures~\ref{GaussMatParEst},~\ref{SRGaussMatEURO} and~\ref{TFGaussMatEURO} that the Nelson-Siegel model and the Nelson-Siegel-Svensson model are close to the monotonic kriging mode, this is especially the case for the Svensson extension. However, these models do not always fulfil the given market constraints (especially for the Nelson-Siegel model). Furthermore, they do not provide confidence intervals.

\begin{rem}
\label{rem:comparison_kernels}
As can be seen on Figure~\ref{TFGaussMatEURO}, the shape of the 
most probable forward curves (black solid lines) differs  under the two considered covariance kernels. When using a Gaussian covariance kernel, the mode curve features an extraneous hump for maturities greater than 30 years  whereas, in the Mat\'ern 5/2 case, it is slightly increasing and follows the fitted Nelson-Siegel-Svensson curve. This difference comes from the fact that $\mathcal{C}^{\infty}$ forward curves generated by the Gaussian kernel are  less flexible to adjust to local tendency of the data than  $\mathcal{C}^{1}$ forward curves generated by the Mat\'ern 5/2 kernel. Note moreover that, in the Gaussian kernel case, inverting the covariance matrix without a nugget effect  may involve numerical instabilities due to the restriction of infinite differentiability. This issue has also been reported in \cite{doi:10.1137/140976741}, where the authors state in particular that ``\textit{this choice of the Mat\'ern covariance function over the commonly used squared exponential family (Sacks et al., 1989) avoids numerical instability, often observed when inverting the covariance matrix, by removing the restriction of infinite differentiability''}.

 \end{rem}

The previous simulations can be used to estimate the distribution of other  financial assets whose values depend on the curve. As an example, in Figure~\ref{renteGaussMatSwap}, we plot an histogram of the present value $\text{\"a}_{\annu{n}}^{(p)}=\sum_{k=0}^{pn-1}\frac{1}{p}Y^N(k/p)$ of a periodic annuity-due using 100 000 simulations of discount factors sample paths, where $n=40$ and $p=12$ (monthly payments). Notice that, despite the variability of the simulated sample paths of the conditional Gaussian process, the present value of the periodic annuity-due remains stable with 95\% confidence intervals $[25.12, 25.30]$ using the Gaussian kernel and $[25.07, 25.41]$ using the Mat\'ern kernel.%



\paragraph{Several quotation dates.}

We now illustrate the building procedure in dimension two, when  data observed at different quotation dates are incorporated. We then construct a surface representing the evolution of discount factors as a function of time-to-maturities and quotation dates. To do this, we use the approach described in Section~\ref{TDC}. In Figure~\ref{Dim2Gauss}, the  surface represents the mode estimator of the conditional GP in dimension two. 
The construction relies on the swap quotations at the 9 dates given in Table \ref{EPOV}. We choose $N_x=40$ and $N_t=20$ and we consider a two-dimensional Gaussian kernel written as
\begin{equation*}
K(\boldsymbol{x},\boldsymbol{x}')=\exp\left(-\frac{(x-x')^2}{2\theta_1}-\frac{(t-t')^2}{2\theta_2}\right),
\end{equation*}
where $\boldsymbol{x}=(x,t)$ and $\boldsymbol{x}'=(x',t')$. For each vector $\boldsymbol{x}=(x,t)$, the first component $x$ represents a time-to-maturity and the second component $t$ represents a quotation date. Without loss of generality, the distance $t'-t$ between two quotation dates has been expressed in percentage of the length between the two extreme dates of the sample. The parameters $\theta_1$ and $\theta_2$ are fixed respectively to $25$ and $0.5$. Notice that the constructed discount factor surface is non-increasing with respect to time-to-maturities.

\begin{figure}[H]
\centering
\includegraphics[scale=0.35]{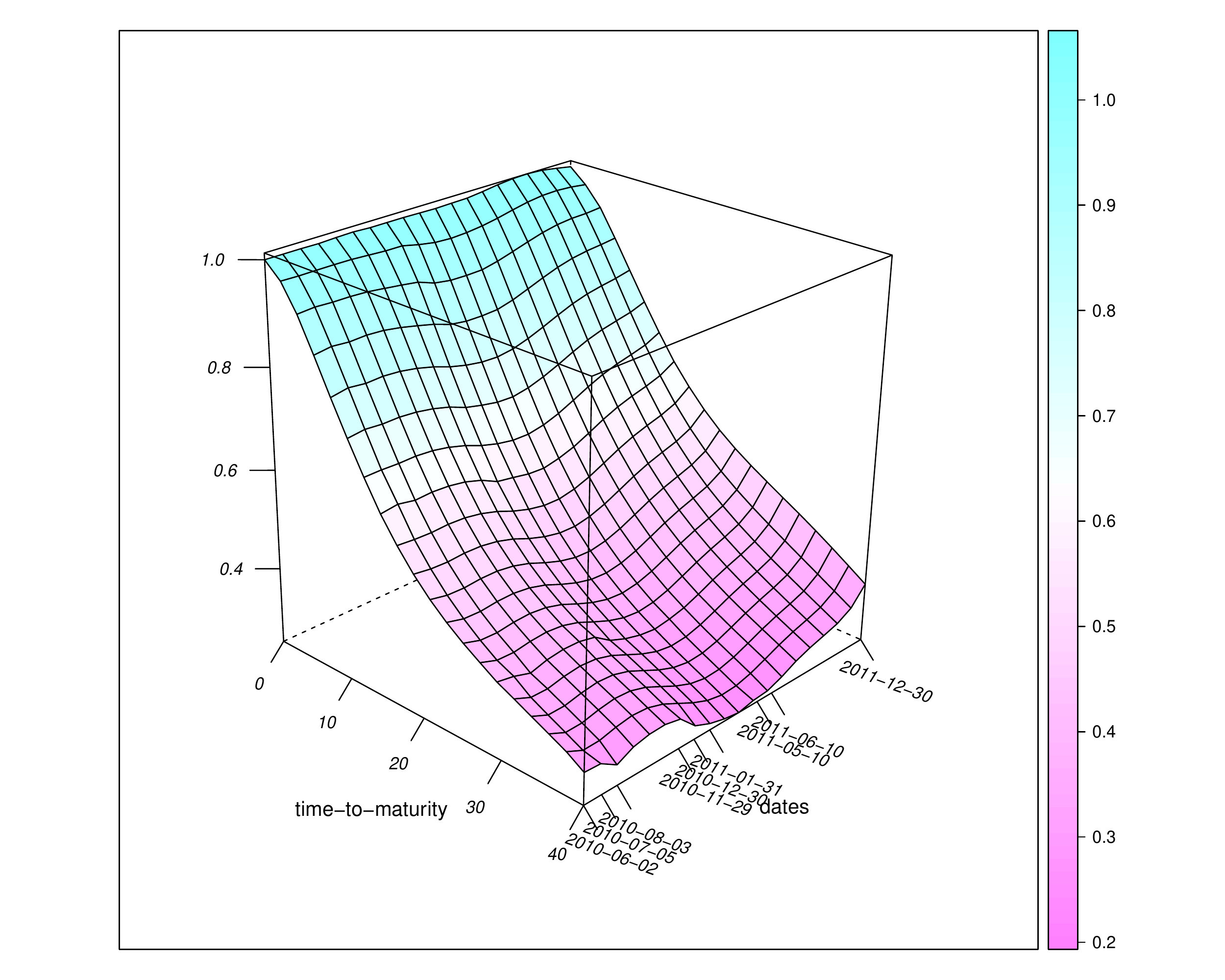}
\caption{Swap-vs-Euribor discount factors as a function of  time-to-maturities and quotation dates.}
\label{Dim2Gauss}
\end{figure}

\subsection{OIS discount curves}\label{sec:OIS}

We now apply the kriging method to construct OIS discount curves. The aim is to construct, at different quotation dates $t$, discount curves $T\rightarrow P^D(t,T)$ based on market quotes of Overnight Indexed Swaps associated with different standard maturities.
We consider par rates of OIS at the 10 quotation dates given in Table~\ref{EPOVOIS}. 
For each of these quotation dates, the term-structure is built  from $14$ swap rates $S_1, \ldots, S_{14}$ associated with standard maturities in the set $E:=\{1,\ldots,10,15,20,30,40\}^\top$. 
For each standard maturity $T\in E$, the value $P(t,k)$ of the curve at time horizons $k=1, \ldots, T$ are linked through the linear relation~\eqref{eq:OIS}. Then, the curve is compatible with market quotes if the vector of discount factors $P^D(t, X) := (P^D(t, 1), \ldots, P^D(t,40))^{\top}$ satisfies a linear system of the form 
\begin{equation}
\label{eq:market_fit_OIS}
A_t \cdot P(t, X)=\boldsymbol{b}_t,
\end{equation}
where $A_t$ is  a $14\times 40$ real matrix and $\boldsymbol{b}_t= (1,\ldots,1)^\top\in \mathbb{R}^{14}$. In this case, we have $n=14$ observations which depends on $m=40$ points of the curve. Note that the market fit condition has exactly the same form as for the previous example in Subsection~\ref{CD}.
\\

\paragraph{Parameters estimation.}
In Table~\ref{EPOVOIS}, we estimate the length hyper parameter $\hat{\theta}_G$ and $\hat{\theta}_M$ associated respectively to Gaussian and Mat\'ern 5/2 covariance functions (see Table~\ref{kernel}). The optimal values in the last two columns of Table~\ref{EPOVOIS} correspond to the values of the LOO criterion defined in (\ref{LOO}) at the global optimum. Note that the estimated parameters  $\hat{\theta}_G$ and $\hat{\theta}_M$ remain stable across the considered quotation dates.  In addition, the obtained optimal values  are close  for the two covariance functions, even if they are slightly smaller for the Gaussian covariance function. Figure~\ref{GaussMatLOOOIS} represents the function to be optimized in criterion (\ref{LOO}) using the OIS data on $03/06/2010$. Given the shape of the functions, the estimation procedure has turned to be much more straightforward using a Mat\'ern 5/2 covariance function.

%


\begin{table}[hptb]
\centering
\caption{Parameters estimation using ACV methods (OIS data).}
\label{EPOVOIS}
\begin{tabular}{ccccc}
\hline \hline
Date  &  $\hat{\theta}_G$ &  $\hat{\theta}_M$ & Optimal value Gaussian & Optimal value Mat\'ern 5/2  \\   
\hline
03/06/2010 & 26.2 & 19.1 & \textbf{2.5e-05} & 9.5e-05 \\
04/10/2010 & 27.8 & 20.6 & \textbf{1.2e-05} & 7.4e-05 \\
31/12/2010 & 26.1 & 18.7 & \textbf{2.6e-05} & 8.4e-05 \\
04/03/2011 & 28.2 & 19.0 & \textbf{1.5e-05} & {4.7e-05} \\
15/06/2011 & 27.3 & 18.2 & \textbf{1.2e-05} & 6.9e-05 \\
10/10/2011 & 26.5 & 23.8 & \textbf{1.3e-05} & 5.3e-05  \\
14/11/2011 & 26.1 & 23.8 & \textbf{1.4e-05} & 7.4e-05 \\
15/12/2011 & 25.8 & 24.2 & \textbf{1.9e-05} & 7.6e-05 \\
\hline
\end{tabular}
\end{table}

\begin{figure}[hptb]
\begin{minipage}{.5\linewidth}
\centering
\includegraphics[scale=.4]{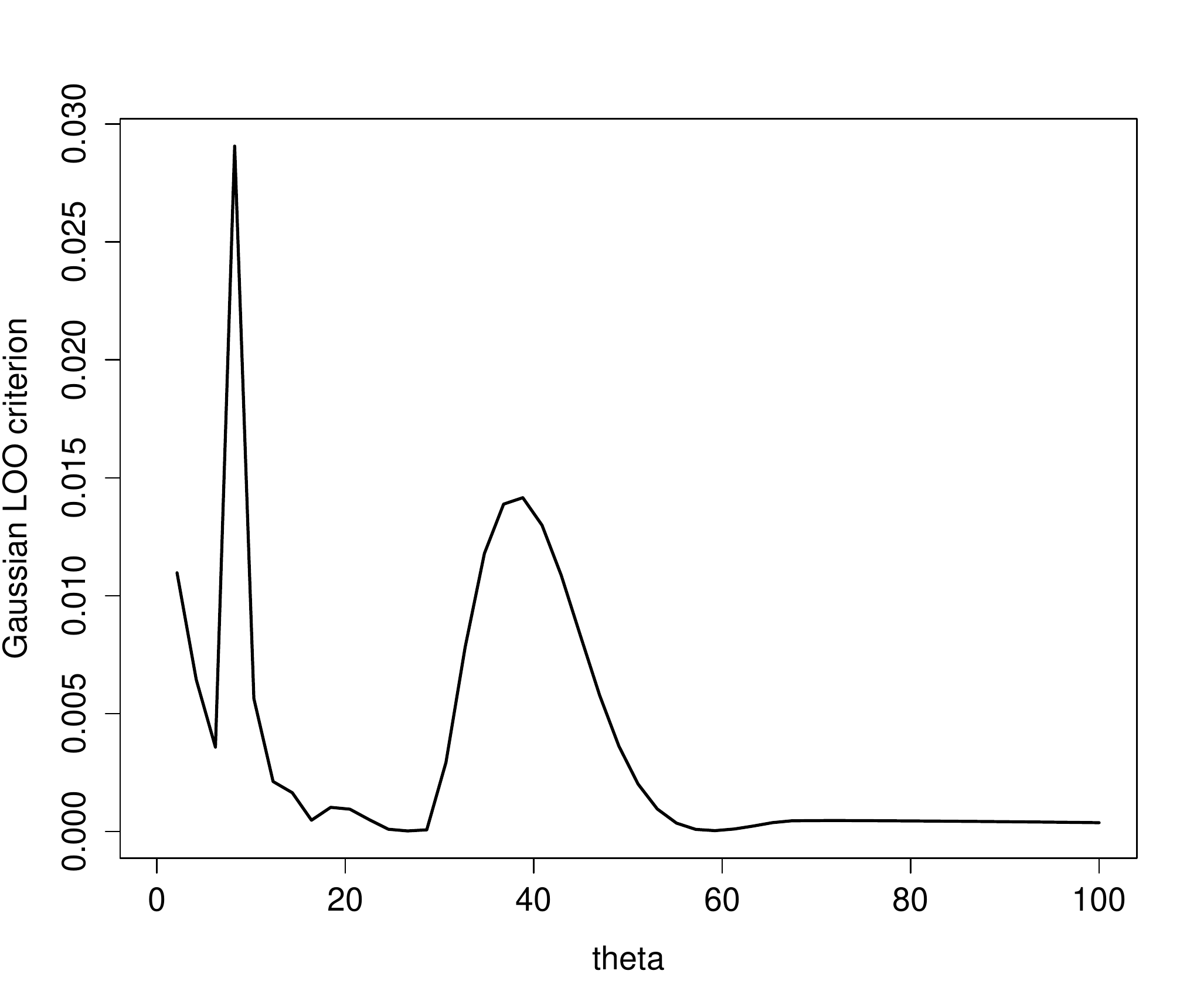}
\end{minipage}%
\begin{minipage}{.5\linewidth}
\centering
\includegraphics[scale=.4]{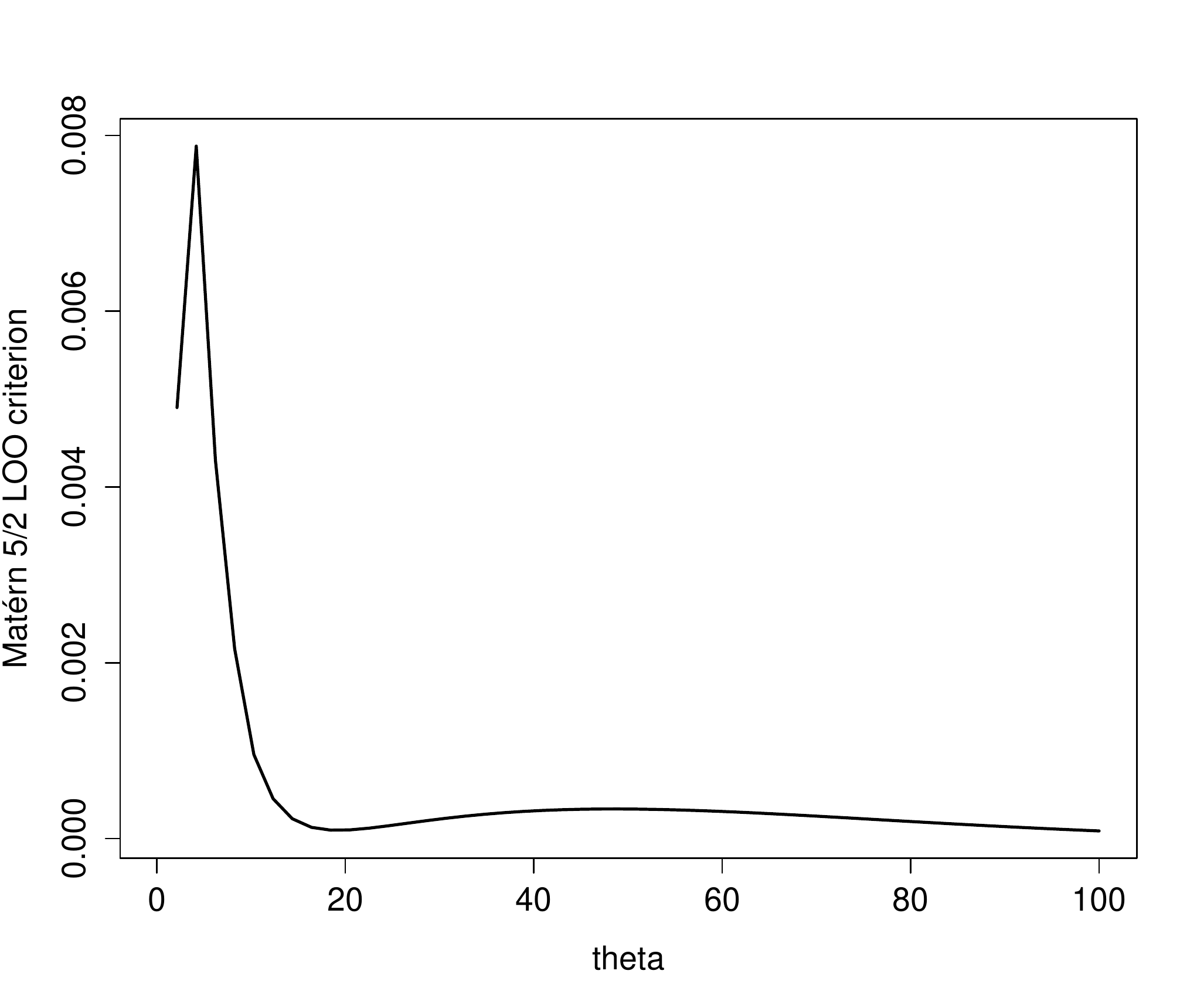}
\end{minipage}
\caption{The function to be optimized in LOO criterion (\ref{LOOLConst}) using the Gaussian covariance function (left) and the Mat\'ern 5/2 covariance function (right). OIS data on 03/06/2010.
}
\label{GaussMatLOOOIS}
\end{figure}

As previously, once the length parameter $\theta$ has been estimated, the standard deviation parameter $\sigma$ is estimated using Equation~\eqref{sigSCV_2}.

%
%

\paragraph{One single quotation date.}
In Figure~\ref{GaussMatParEstOIS}, we choose $N=50$ and generate 100 sample paths of discount factors constructed from model~\eqref{propModel} when using a Gaussian covariance function (left graph) and a Mat\'ern 5/2 covariance function (right graph). All the curves are non-increasing with respect to time-to-maturities. In addition, they all are perfectly compatible with  OIS data as of $03/06/2010$. The Gaussian process hyper-parameters have been estimated by the ACV method described in Section~\ref{PE}.
The estimated hyper-parameters are given by $(\hat{\theta}_G, \hat{\sigma}_G)=(26.2,\, 4.24)$ when using a Gaussian covariance function and by $(\hat{\theta}_M, \hat{\sigma}_M)=(19.1,\,0.24)$ when using a Mat\'ern 5/2 covariance function.
The black solid line represents the most likely curve, i.e., the mode of the conditional GP. Recall that, by construction, this curve satisfies the given constraints. The black dashed-lines represent the 95\% point-wise confidence intervals quantified by simulation. Figure~\ref{SRGaussMat} and Figure~\ref{TFGaussMat} gives the corresponding spot rate and forward curves.\\

Again, all figures~\ref{GaussMatLOOOIS},~\ref{SRGaussMat} and~\ref{TFGaussMat} are given together with the associated best-fitted Nelson-Siegel curves  \citep[see][]{nelson1987parsimonious} and the associated best-fitted Svensson curves \citep[see][]{Svensson:1994}.  Parameters have been estimated by minimizing the sum of squared errors between market and model prices. We use a gradient descent algorithm with randomly chosen starting values as described in \cite{gilli2010calibrating}. The optimal parameters are given in Table~\ref{ParamNS_OIS}.
\begin{table}[hptb]
\centering
\caption{Parameters estimation for Nelson-Siegel and Nelson-Siegel-Svensson model (OIS data, 03/06/2010).}
\label{ParamNS_OIS}
\begin{tabular}{l|cccccc}
\hline \hline
 & $\lambda_1$ & $\lambda_2$ & $\beta_1$ & $\beta_2$ & $\beta_3$ & $\beta_4$ \\ 
\hline 
Nelson-Siegel          & 1.0890 &  - & 0.0341 &  -0.0171 & -0.0601 & - \\ 
\hline 
Nelson-Siegel-Svensson & 1.0938 & 15.1891 & 0.0502 & -0.0164 & -0.1074 & -0.0494 \\ 
\hline 
\end{tabular}
\end{table}

\begin{figure}[hptb]
\begin{minipage}{.5\linewidth}
\centering
\includegraphics[scale=.4]{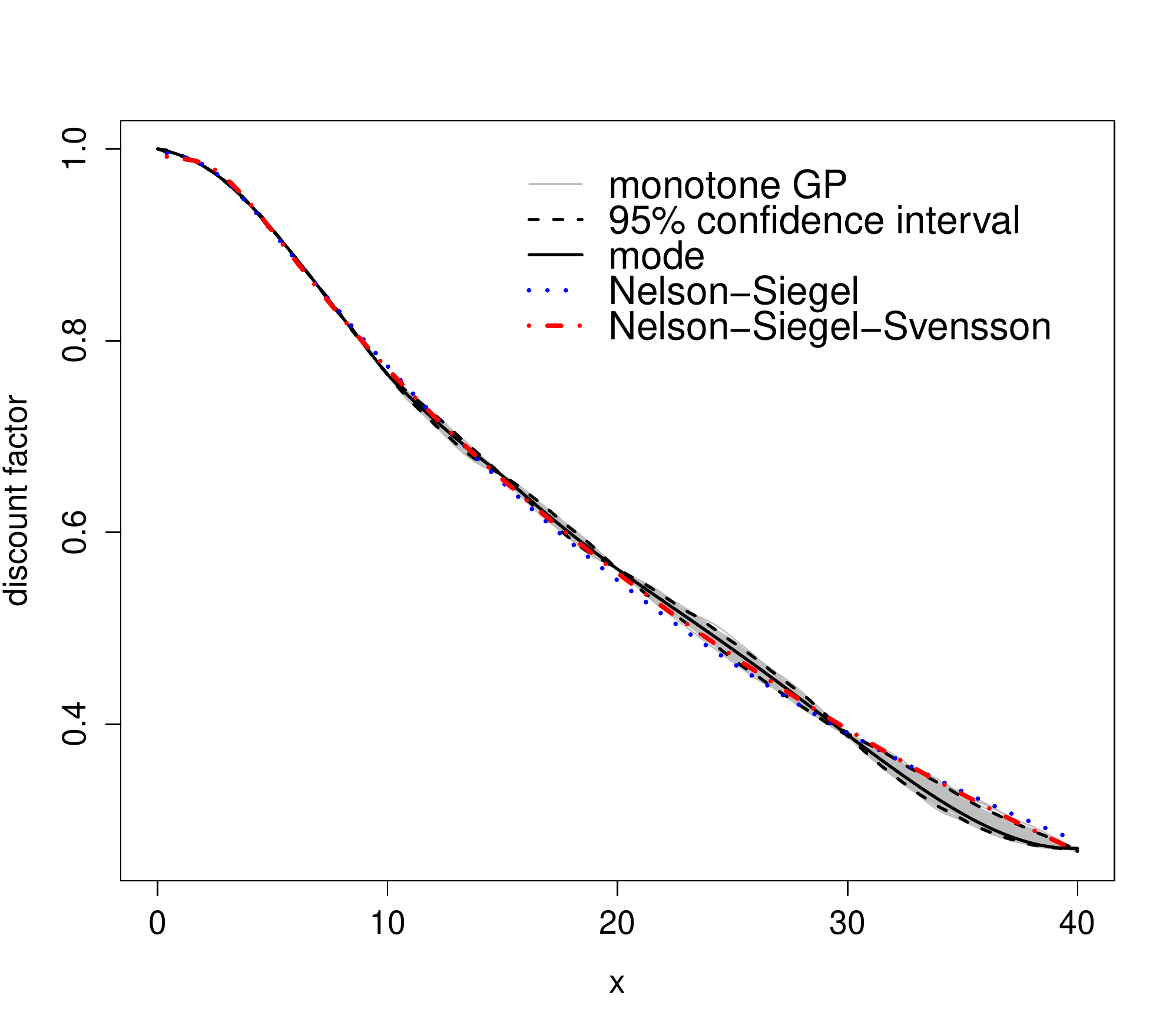}
\end{minipage}%
\begin{minipage}{.5\linewidth}
\centering
\includegraphics[scale=.4]{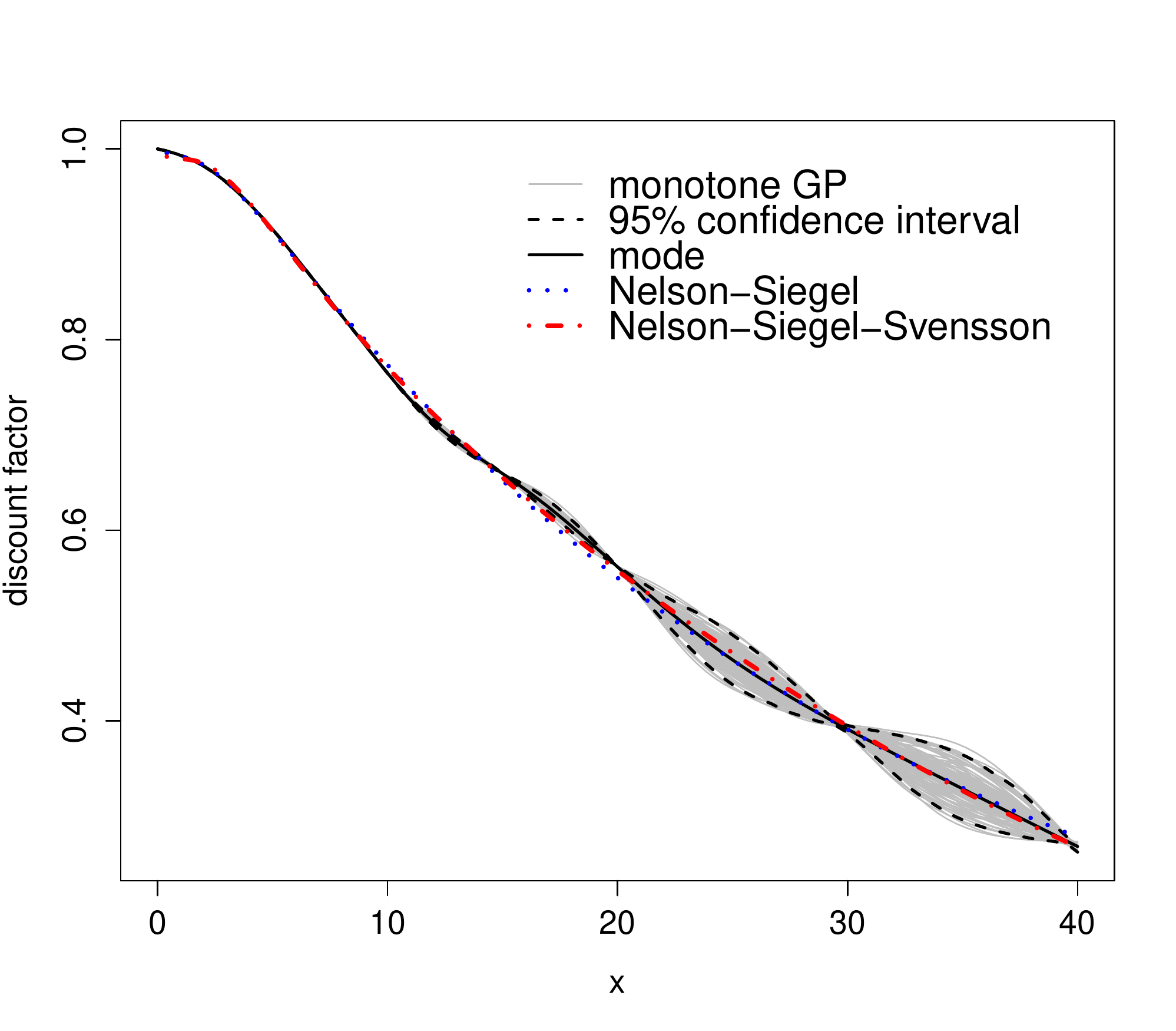}
\end{minipage}
\caption{OIS discount factor curves (gray lines) given as simulated paths of a conditional GP with non-increasing constraints using a Gaussian covariance function with nugget equal to $10^{-5}$ (left) and a Mat\'ern 5/2 covariance function without nugget (right). OIS data of 03/06/2010.
}
\label{GaussMatParEstOIS}
\end{figure}

\begin{figure}[hptb]
\begin{minipage}{.5\linewidth}
\centering
\includegraphics[scale=.4]{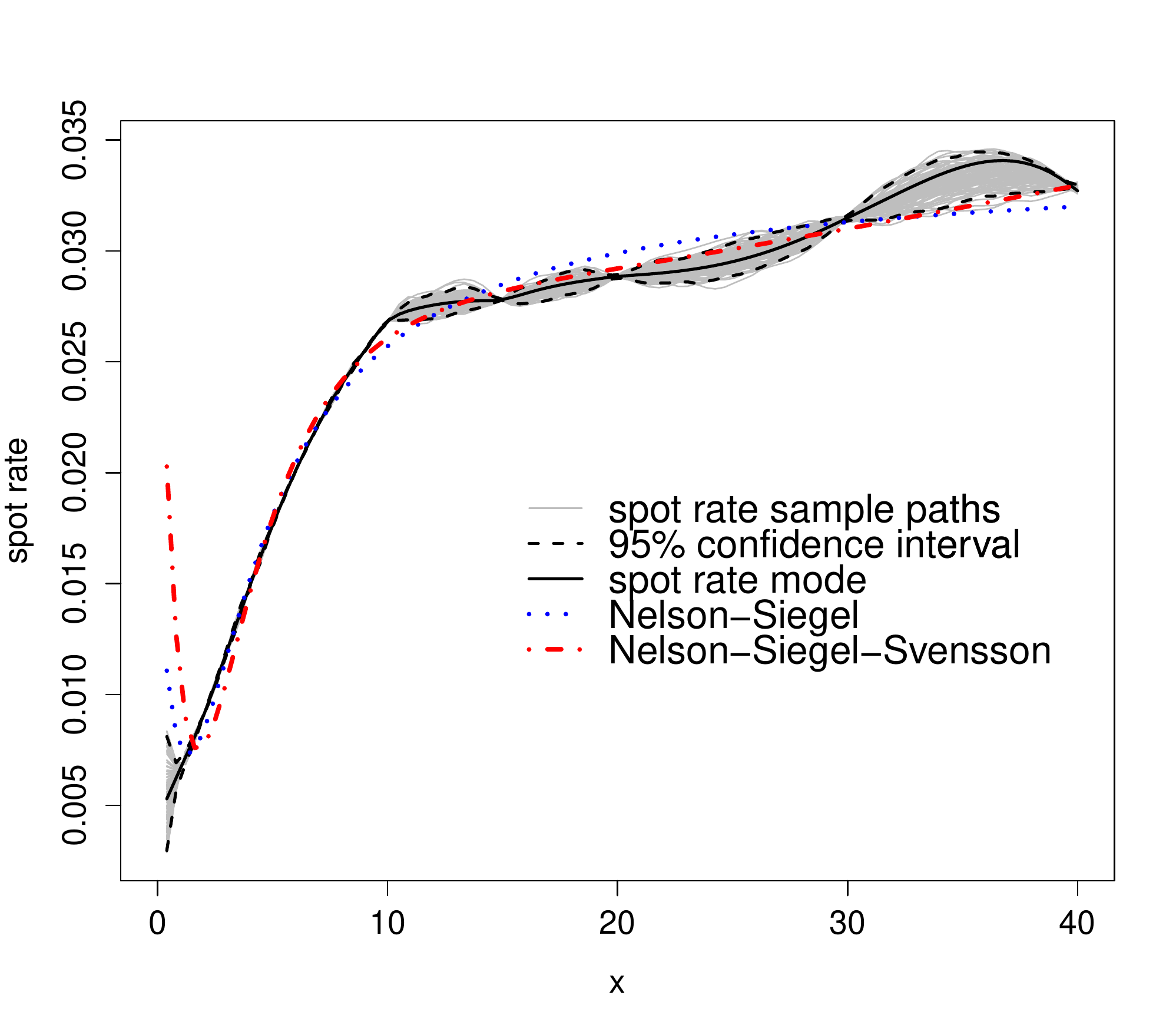}
\end{minipage}%
\begin{minipage}{.5\linewidth}
\centering
\includegraphics[scale=.4]{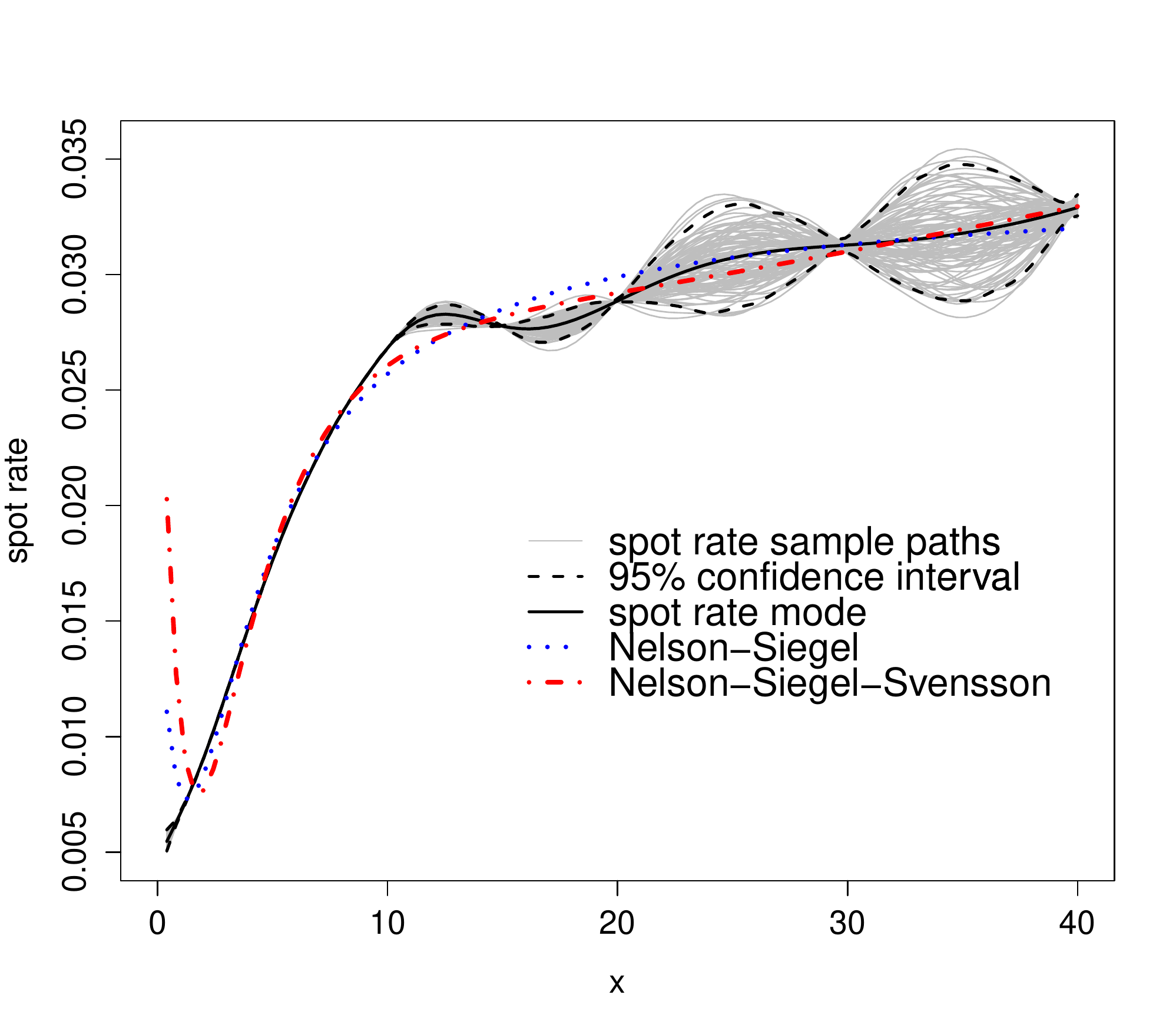}
\end{minipage}
\caption{
Spot rates obtained from sample paths of Figure \ref{GaussMatParEstOIS} with Gaussian covariance function (left) and Mat\'ern 5/2 covariance function (right). Gray lines represent $-\frac{1}{x}\log Y^N(x)$ for each sample path. The black solid line is the most likely spot rate curve $-\frac{1}{x}\log M^N_K\left(x \suchthat A,\boldsymbol{b}\right)$.
}
\label{SRGaussMat}
\end{figure}

\begin{figure}[hptb]
\begin{minipage}{.5\linewidth}
\centering
\includegraphics[trim={0 5cm 0 0}, scale=.4]{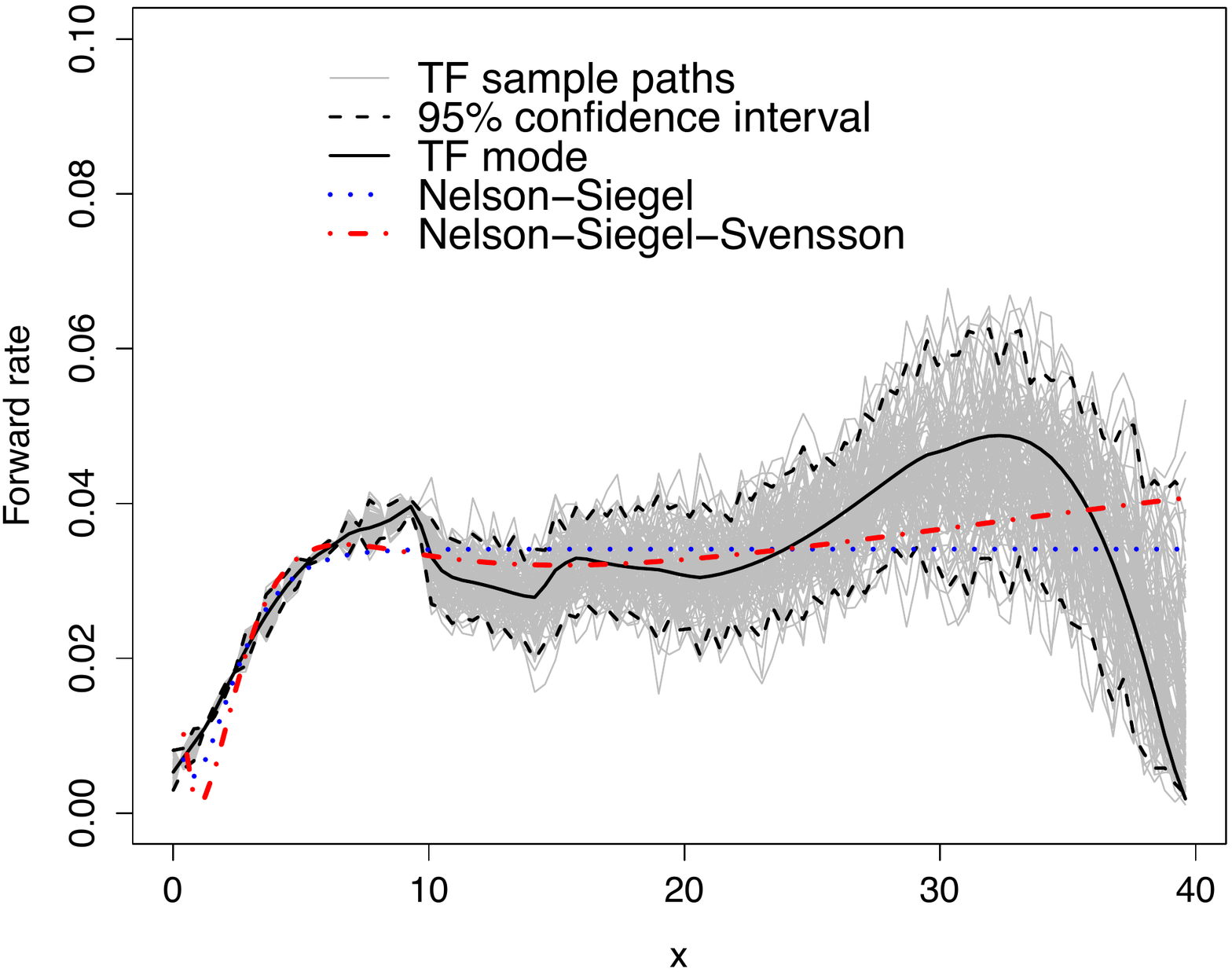}
\end{minipage}%
\begin{minipage}{.5\linewidth}
\centering
\includegraphics[trim={0 5cm 0 0}, scale=.4]{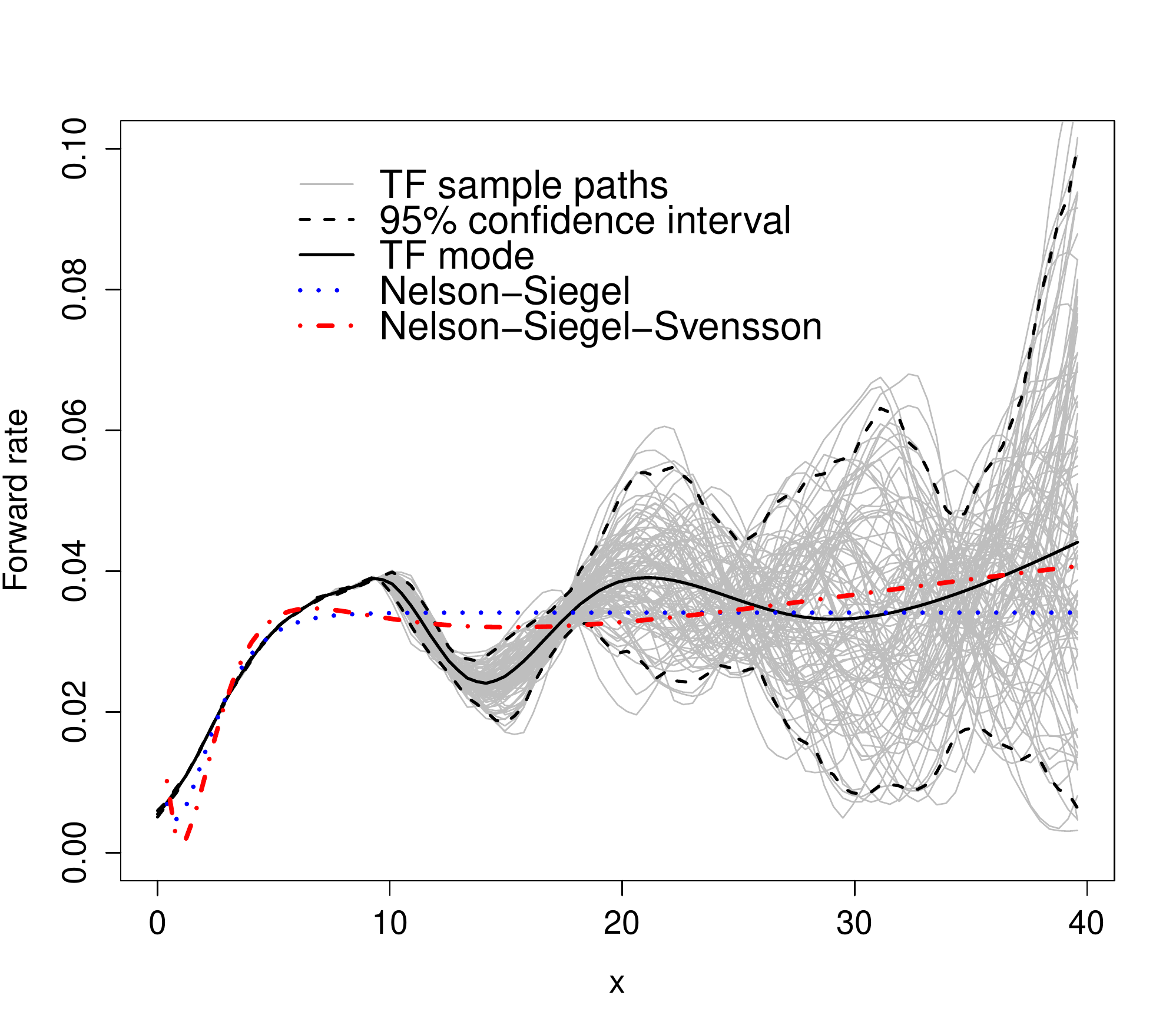}
\end{minipage}
\caption{
Forward rates obtained from sample paths of Figure  \ref{GaussMatParEstOIS}   with Gaussian covariance function (left) and Mat\'ern 5/2 covariance function (right). Gray lines represent $-\frac{d}{dx}\log Y^N(x)$ for each sample path. The black solid line is the most likely forward rate curve $-\frac{d}{dx}\log M^N_K\left(x\suchthat A,\boldsymbol{b}\right)$.
}
\label{TFGaussMat}
\end{figure}

Regarding comparison of curves constructed by the two considered covariance kernels, Remark \ref{rem:comparison_kernels} also applies here.

It appears that Nelson-Siegel and Nelson-Siegel-Svensson models do not fulfil all market constraints, contrary to the proposed methodology, which also gives confidence intervals.


\paragraph{Several quotation dates.}
Using the two dimensional approach described in Section~\ref{TDC}, we build in Figure~\ref{Dim2GaussOIS}  a surface representing OIS discount factors with respect to time-to-maturities and quotation dates. The construction relies on OIS quotations at the 8 dates given in Table \ref{EPOVOIS}. The surface corresponds to the mode of the conditional GP given market-fit equality constraints and non-increasing constraints in the direction of time-to-maturities.  We choose $N_x=40$ and $N_t=20$ and we consider a two-dimensional Gaussian kernel written as
\begin{equation*}
K(\boldsymbol{x},\boldsymbol{x}')=\exp\left(-\frac{(x-x')^2}{2\theta_1}-\frac{(t-t')^2}{2\theta_2}\right),
\end{equation*}
where $\boldsymbol{x}=(x,t)$ and $\boldsymbol{x}'=(x',t')$. For each vector $\boldsymbol{x}=(x,t)$, the first component $x$ represents a time-to-maturity and the second component $t$ represents a quotation date.
Without loss of generality, the distance  $t-t'$ between two quotation dates has been expressed in percentage of the length between the two extreme dates of the sample. The parameters $\theta_1$ and $\theta_2$ are fixed respectively to $25$ and $0.5$. Observe that the constructed discount factor surface is non-increasing with respect to time-to-maturities.

\begin{figure}[H]
\centering
\includegraphics[scale=0.35]{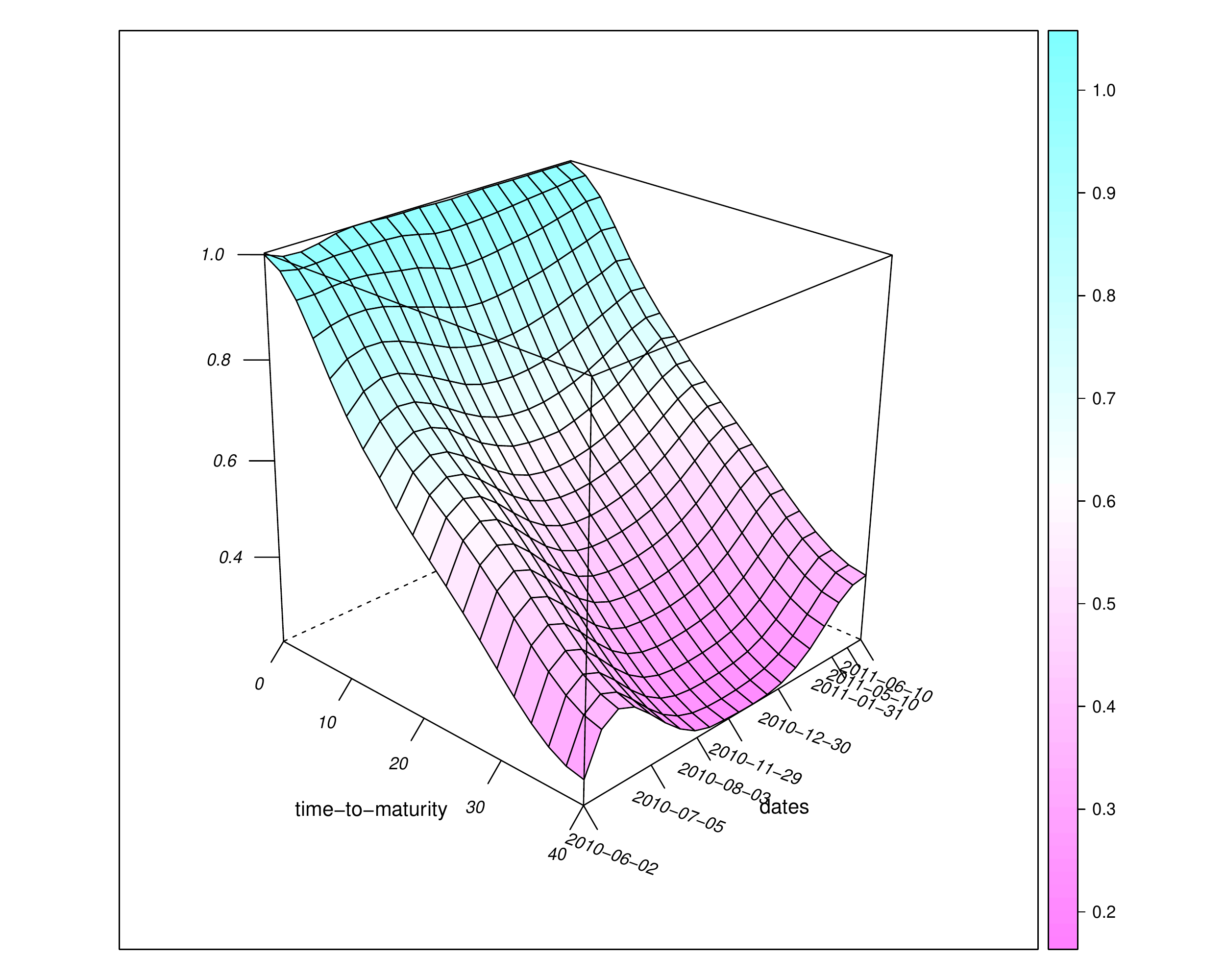}
\caption{
OIS discount factors as a function of  time-to-maturities and quotation dates.
}
\label{Dim2GaussOIS}
\end{figure}

\subsection{CDS implied default distributions}

We now apply the kriging method to build CDS implied default distributions. Indeed, CDS can be seen as an insurance product that covers the  loss of a particular debt issuer if the latter defaults over a certain protection period or maturity. The market quotes (CDS spreads) of these products provide information on the current cost for protection and, in turn,  on how the market assesses default probabilities of the underlying entity at different time horizons. 
\\

Our aim is to construct, at different quotation dates $t$, implied survival functions $T\rightarrow Q(t,T)$ of a particular debt issuer by observing the corresponding term-structure of CDS spreads (CDS spreads at increasing protection horizons). The quantity $Q(t,T)$ gives the probability at time $t$ that the underlying debt issuer does not default before the time horizon $T$. In this numerical illustration, we consider CDS  on the Russian sovereign debt for the 10 quotation dates given in Table~\ref{EPOVCDS}. 
For each of these quotation dates, this corresponds to $7$ CDS spreads $S_1, \ldots, S_{7}$ associated with protection maturities (in years) in the set $E:=\{1,2,3,4,5,7,10\}^\top$. 
For each standard maturity $T\in E$, the values of the survival probabilities at each premium payment dates $\ta_1<\cdots<\ta_n=T$ are linked through the linear relation~\eqref{eq:discrete_linear_CDS}. Note that the  premium payment dates are separated by a quarter period and these dates coincide for all quoted CDS maturities. 
In our numerical illustrations, the expected recovery rate $R$ is fixed to $40\%$ and the discount factors $P^D(t, \tau_k)$ are constructed by linear interpolation of treasury constant maturity rates given (for all considered quotation dates) by the Federal Reserve Bank of St. Louis. Then, if $\delta = 1/4$ represents a quarter period, the implied default distribution is compatible with market quotes if the vector of survival probabilities  $Q(t, X) := (Q(t, \delta), Q(t, 2\delta), \ldots, Q(t,10))^{\top}$ satisfies a linear system of the form 
\begin{equation}
\label{eq:market_fit_CDS}
A_t \cdot Q(t, X)=\boldsymbol{b}_t
\end{equation}
where $A_t$ is  a $7\times 40$ real matrix and $\boldsymbol{b}_t= (1-R,\ldots,1-R)^\top\in \mathbb{R}^{7}$. In this case, we have $n=7$ observations which depends on $m=40$ points of the curve.

\paragraph{Parameters estimation.}
In Table~\ref{EPOVCDS}, we compare the estimation of  the length hyper parameter  for a Gaussian kernel ($\hat{\theta}_G$) and for a Mat\'ern 5/2 kernel  ($\hat{\theta}_M$). The estimation has been performed using the ACV method described in Section~\ref{PE}. 
As can be seen on Figure \ref{LOOMatGaussCDS}, the LOO objective function looks similar for these two covariance functions. 

\begin{table}[hptb]
\centering
\caption{Length parameters estimation using ACV methods (CDS data).}
\label{EPOVCDS}
\begin{tabular}{ccccc}
\hline \hline
Date  &  $\hat{\theta}_G$ &  $\hat{\theta}_M$ & Optimal value Gaussian & Optimal value Mat\'ern 5/2   \\   
\hline
06/01/2005 & 6.6  & 10.5  & 6.3e-06 & \textbf{4.9e-07} \\
02/02/2006 & 4.9  & 15.5  & {2.7e-07} & \textbf{1.8e-07} \\
20/03/2007 & 4.5  & 19.8  & 7.6e-06 & \textbf{7.1e-07} \\
04/04/2008 & 9.7  & 11.4  & 5.0e-06 & \textbf{2.6e-06} \\
11/05/2009 & 4.9 & 10.1  & 1.3e-04  & \textbf{4.2e-05} \\
21/06/2010 & 11.1 & 10.1  & 1.1-05  & \textbf{3.4e-06}  \\
14/07/2011 & 14.5 & 17.1  & \textbf{1.2e-06} & 1.8e-06 \\
23/08/2012 & 4.1  & 10.8  & 8.2e-06 & \textbf{1.5e-06} \\
\hline
\end{tabular}
\end{table}

\begin{figure}[hptb]
\begin{minipage}{.5\linewidth}
\centering
\includegraphics[scale=.4]{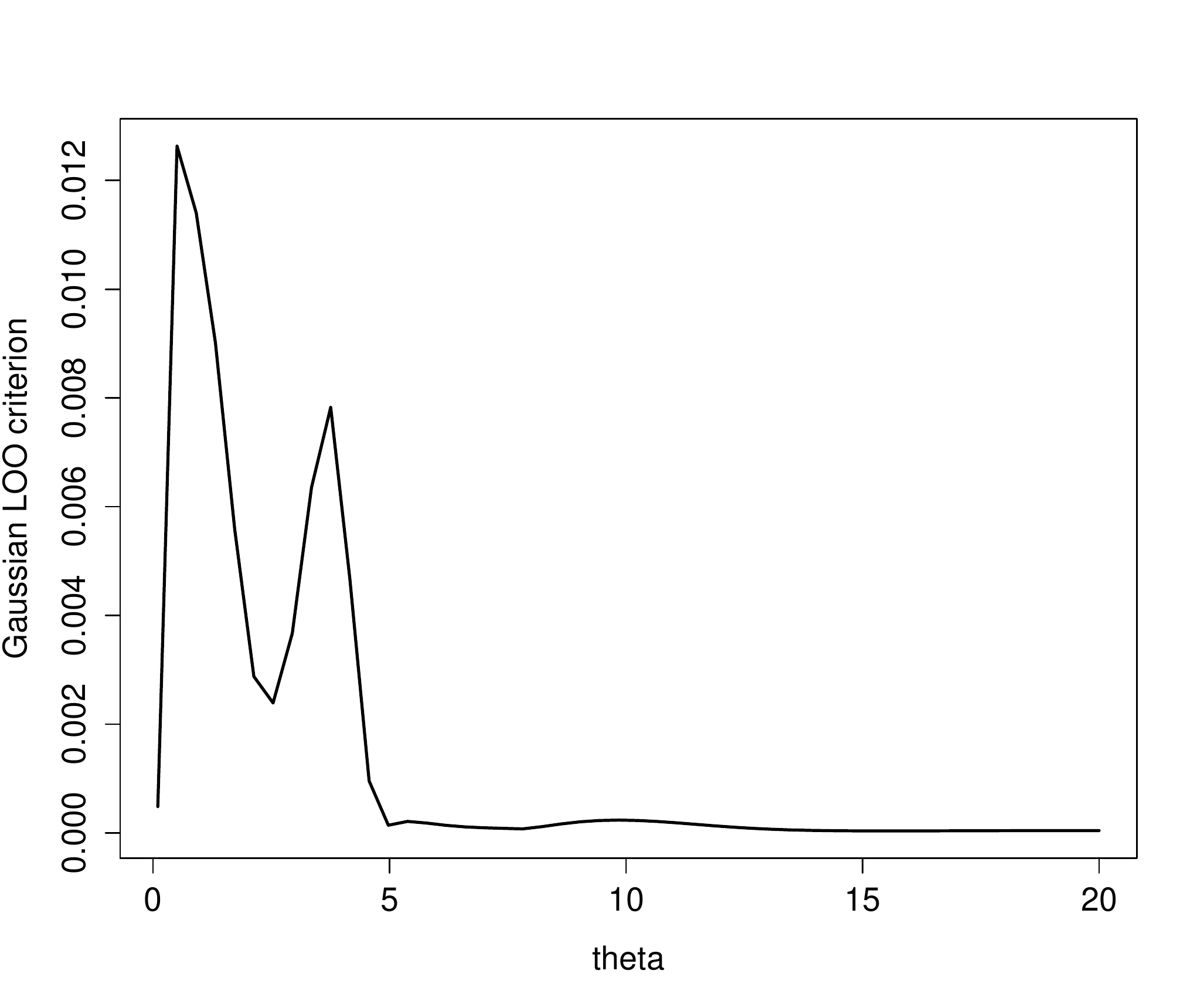}
\end{minipage}%
\begin{minipage}{.5\linewidth}
\centering
\includegraphics[scale=.4]{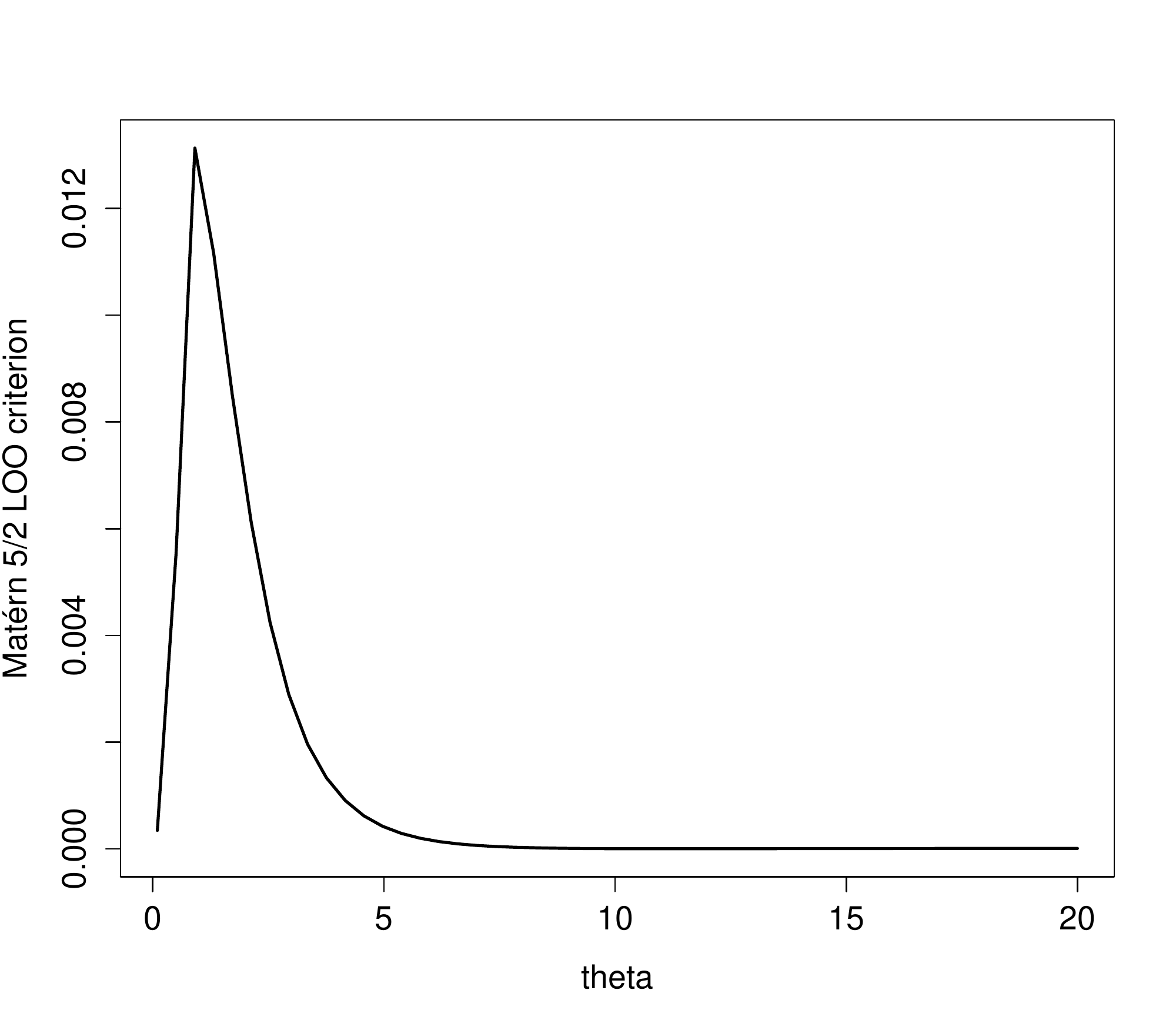}
\end{minipage}
\caption{The function to be optimized in the LOO criterion (\ref{LOOLConst}) using the Gaussian (left) and the Mat\'ern 5/2 (right) covariance function. CDS quotes as of $06/01/2005$. 
}
\label{LOOMatGaussCDS}
\end{figure}


As previously, once the length parameter $\theta$ has been estimated, the standard deviation parameter $\sigma$ is estimated using Equation~\eqref{sigSCV_2}.

\paragraph{One single quotation date.}

In Figure~\ref{GaussMatCDS}, we choose $N=50$ and generate 100 sample paths of CDS implied survival curves, constructed from model~\eqref{propModel} when using a Gaussian covariance function (left graph) and a Mat\'ern 5/2 covariance function (right graph). All the curves are non-increasing with respect to time horizons. In addition, they all are perfectly compatible with  CDS data as of $06/01/2005$. 
The Gaussian process hyper-parameters have been estimated by the ACV method described in Section \ref{PE}.
The estimated hyper-parameters are given by  $(\hat{\theta}_G, \hat{\sigma}_G)=(4.9,\, 0.09)$ when using a Gaussian covariance function and by $(\hat{\theta}_M,\hat{\sigma}_M)=(10.5,\,0.22)$ when using a Mat\'ern 5/2 covariance function.
The black solid line represents the most likely curve, i.e., the mode of the conditional GP. Recall that, by construction, this curve satisfies the given constraints. The black dashed-lines represent the 95\% point-wise confidence intervals quantified by simulation. 

\begin{figure}[hptb]
\begin{minipage}{.5\linewidth}
\centering
\includegraphics[trim=0 5cm 0 0 0, scale=.4]{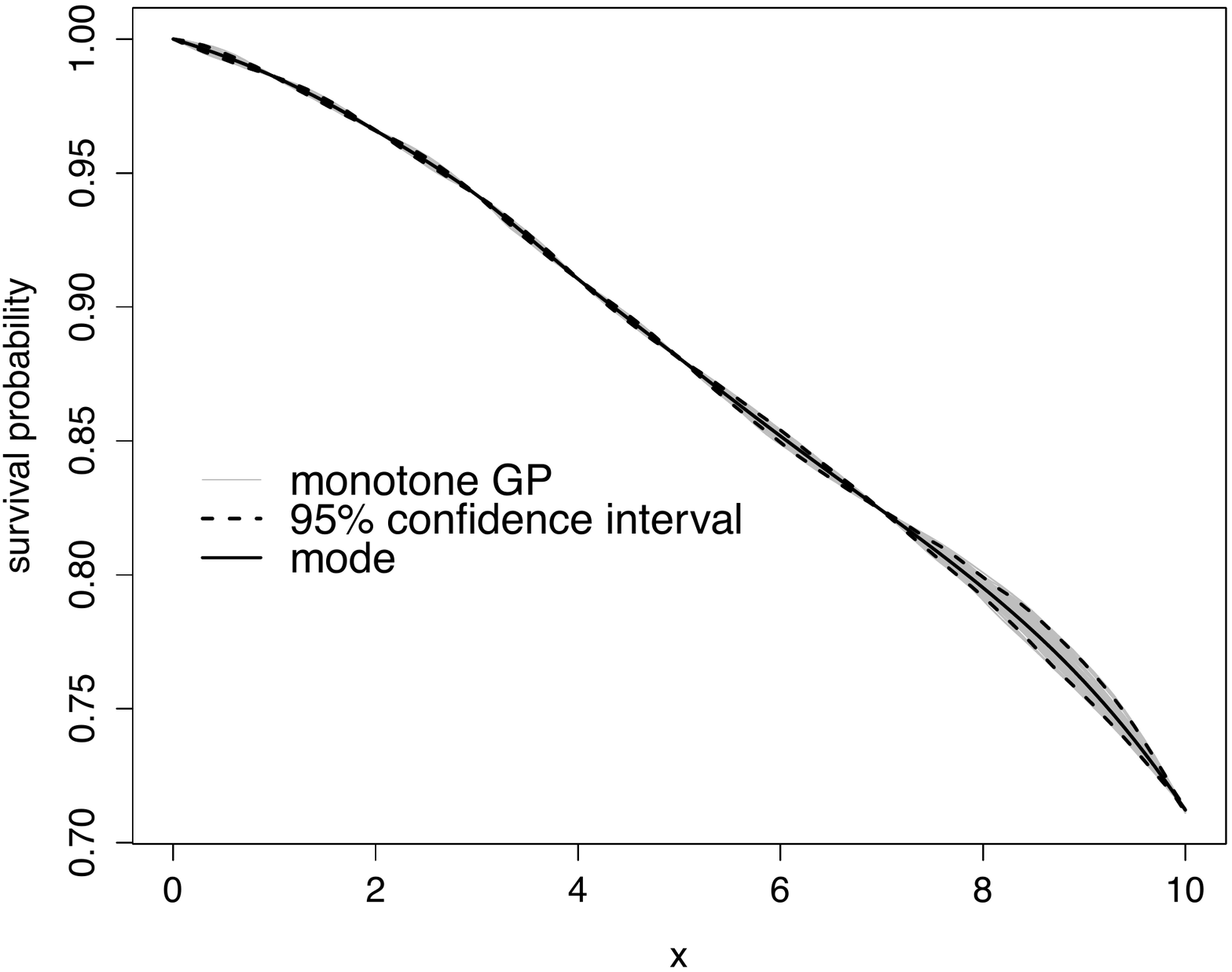}
\end{minipage}%
\begin{minipage}{.5\linewidth}
\centering
\includegraphics[trim=0 5cm 0 0 0, scale=.4]{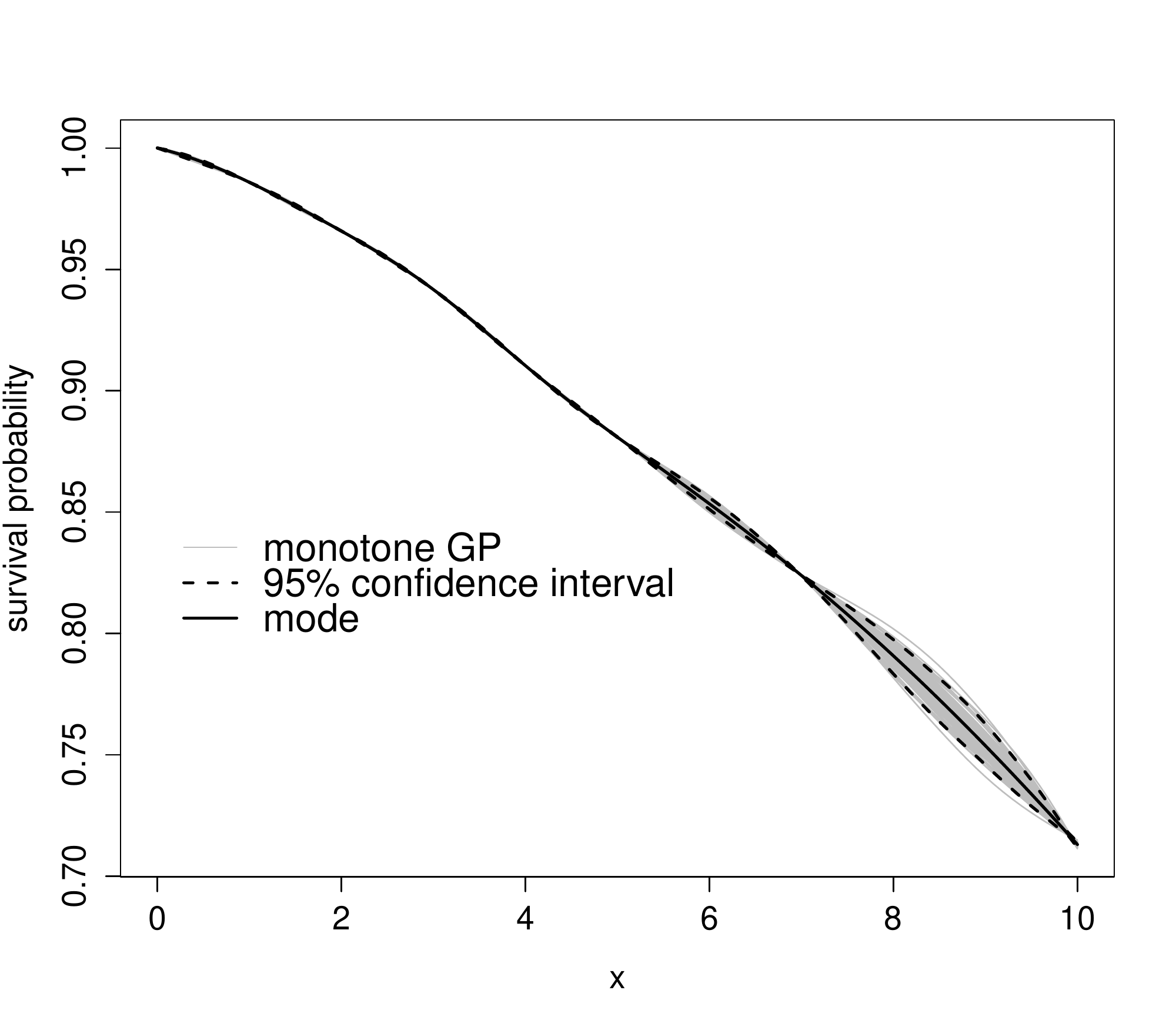}
\end{minipage}
\caption{
CDS implied survival curves (gray lines) given as simulated paths of a conditional GP with non-increasing constraints using a Gaussian covariance function (left)  or a  Mat\'ern 5/2 covariance function (right). 
}
\label{GaussMatCDS}
\end{figure}


\paragraph{Several quotation dates.}
Using the two dimensional approach described in Section~\ref{TDC}, we build in Figure~\ref{Dim2GaussCDS} a surface representing CDS implied survival curves as a function of time horizons and quotation dates.  The construction relies on CDS quoted spreads at the 8 dates given in Table~\ref{EPOVCDS}. 
The surface corresponds to the mode of the conditional GP given market-fit equality constraints and non-increasing constraints in the direction of time-to-maturities.  We choose $N_x=40$ and $N_t=20$ and we consider a two-dimensional Gaussian kernel written as
\begin{equation*}
K(\boldsymbol{x},\boldsymbol{x}')=\exp\left(-\frac{(x-x')^2}{2\theta_1}-\frac{(t-t')^2}{2\theta_2}\right),
\end{equation*}
where $\boldsymbol{x}=(x,t)$ and $\boldsymbol{x}'=(x',t')$. For each vector $\boldsymbol{x}=(x,t)$, the first component $x$ represents a time-to-maturity and the second component $t$ represents a quotation date.
Without loss of generality, the distance  $t-t'$ between two quotation dates has been expressed in percentage of the length between the two extreme dates of the sample. The parameters $\theta_1$ and $\theta_2$ are fixed respectively to $8$ and $1.7$. Notice that the constructed discount factor surface is non-increasing with respect to time-to-maturities. Considering several quotation dates simultaneously offers the advantage of increasing the data set for a better estimation of hyper-parameters  and of creating a consistent interpolation procedure across two directions (time horizons and quotation dates).


\begin{figure}[H]
\centering
\includegraphics[scale=0.35]{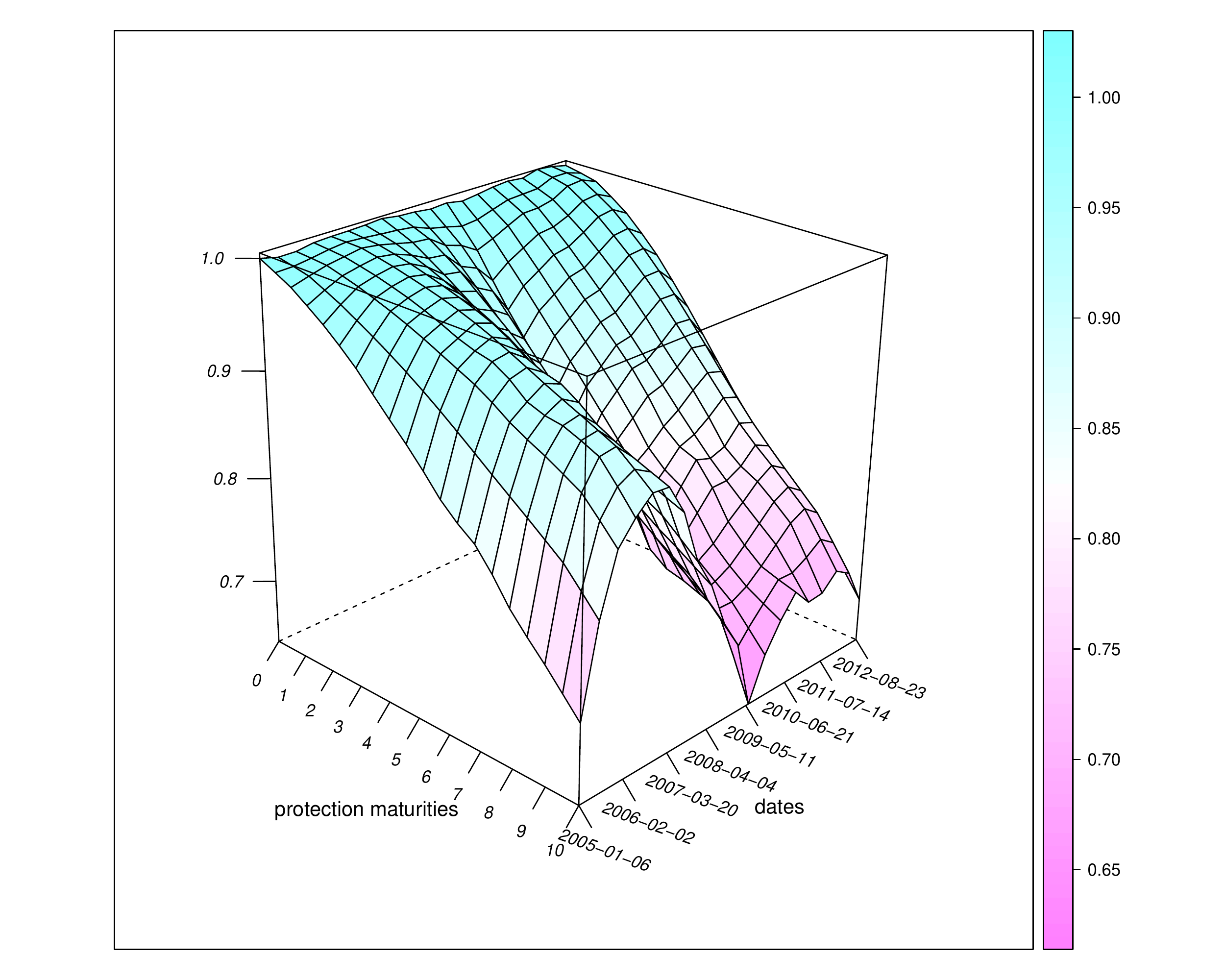}
\caption{
CDS implied survival probabilities as a function of  time-to-maturities and quotation dates.
}
\label{Dim2GaussCDS}
\end{figure}

\subsection{Need for monotonic kriging techniques}\label{sec:motivMonotonie}


At the beginning of the Section~\ref{IntModIConst}, we justified monotonic techniques developed in this paper by some no-arbitrage constraints.
A natural question is whether a monotonic condition is required to construct realistic term-structures. Let us stress that the construction problems we have considered rely on a market information summarized by an ill-posed system of equations of the form
\begin{equation*}
A \mathbf{Y}(\mathbf{X})=\mathbf{b}\, ,
\end{equation*}
where $A$ is not necessarily a square matrix. Obviously, these constraints do not lead to a unique set of possible values for $Y(x^{(i)})$, $i=1, \ldots, m$, so that spot rates or implied default probabilities are not directly available at point $x^{(1)}, \ldots, x^{(m)}$.
In the absence of arbitrage opportunities, the process $Y$ representing default-free zero-coupon bonds, discount factors or implied survival probabilities shall be monotone. It is thus natural to consider a constrained interpolation technique to construct curves based on such quantities.\\

Consider here the previous swap curve construction problem but relax the monotonicity constraint on the default-free zero-coupon bond process $Y$. As can be seen on Figure~\ref{fig:motivMonotoneDF} (left panel), the resulting \emph{kriging mean  is not  a decreasing function}. In addition, even when the average curve is monotonic, some \textit{sample curves are clearly not monotonic}, leading to wide and \textit{unrealistic confidence intervals}. On the right panel of Figure~\ref{fig:motivMonotoneDF} one can see that corresponding spot rates may seem reasonable in average, but omitting the monotonicity constraints leads to wide confidence intervals.
Notice that Figure~\ref{fig:motivMonotoneDF} does just aim at illustrating some problems that can be encountered when using unconstrained techniques, so that other choices of kernel or parameters in an unconstrained setting are not discussed here.\\

\begin{figure}[hptb]
\begin{minipage}{.5\linewidth}
\centering
\includegraphics[scale=0.4]{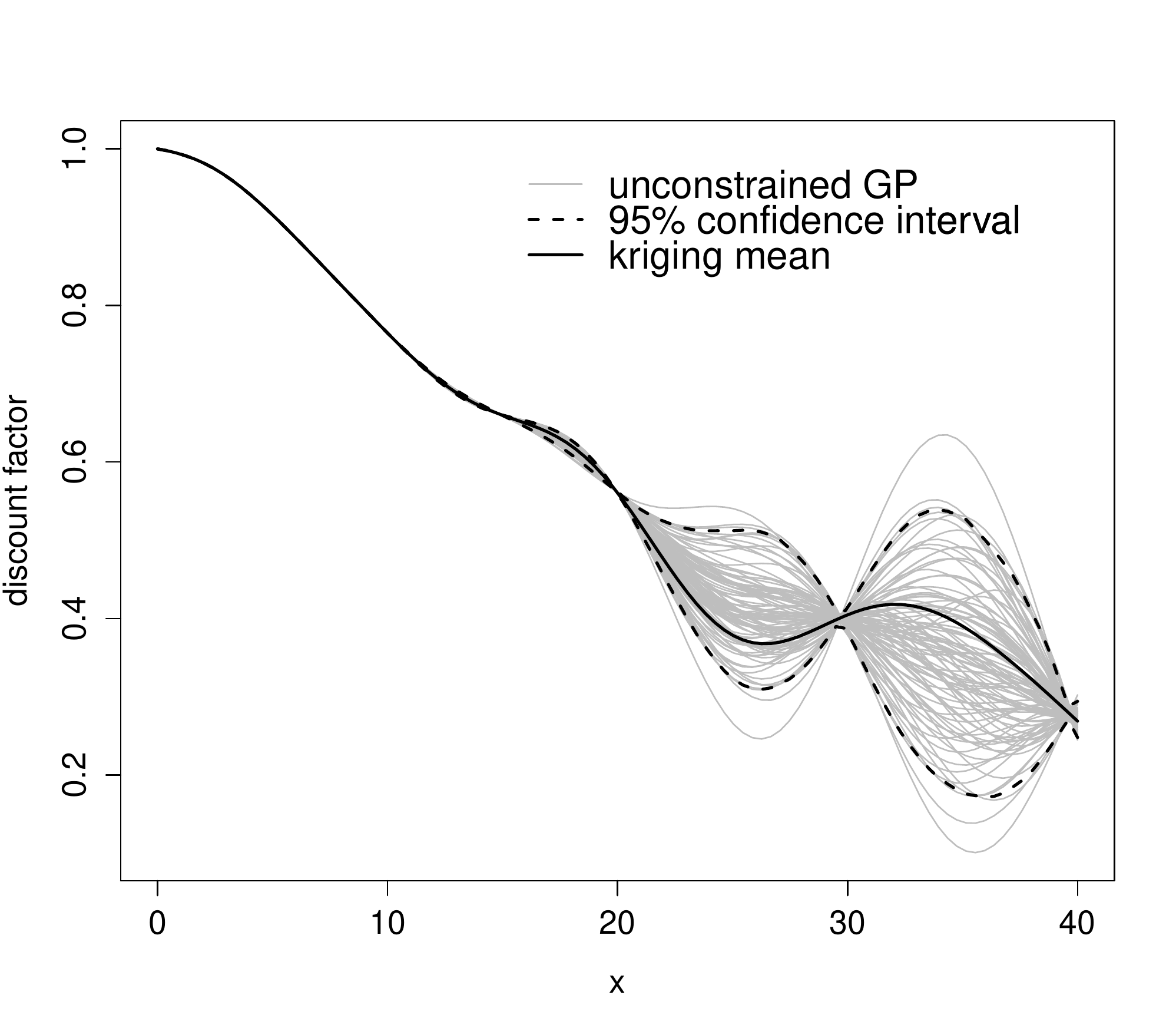}
\end{minipage}%
\begin{minipage}{.5\linewidth}
\centering
\includegraphics[scale=0.4]{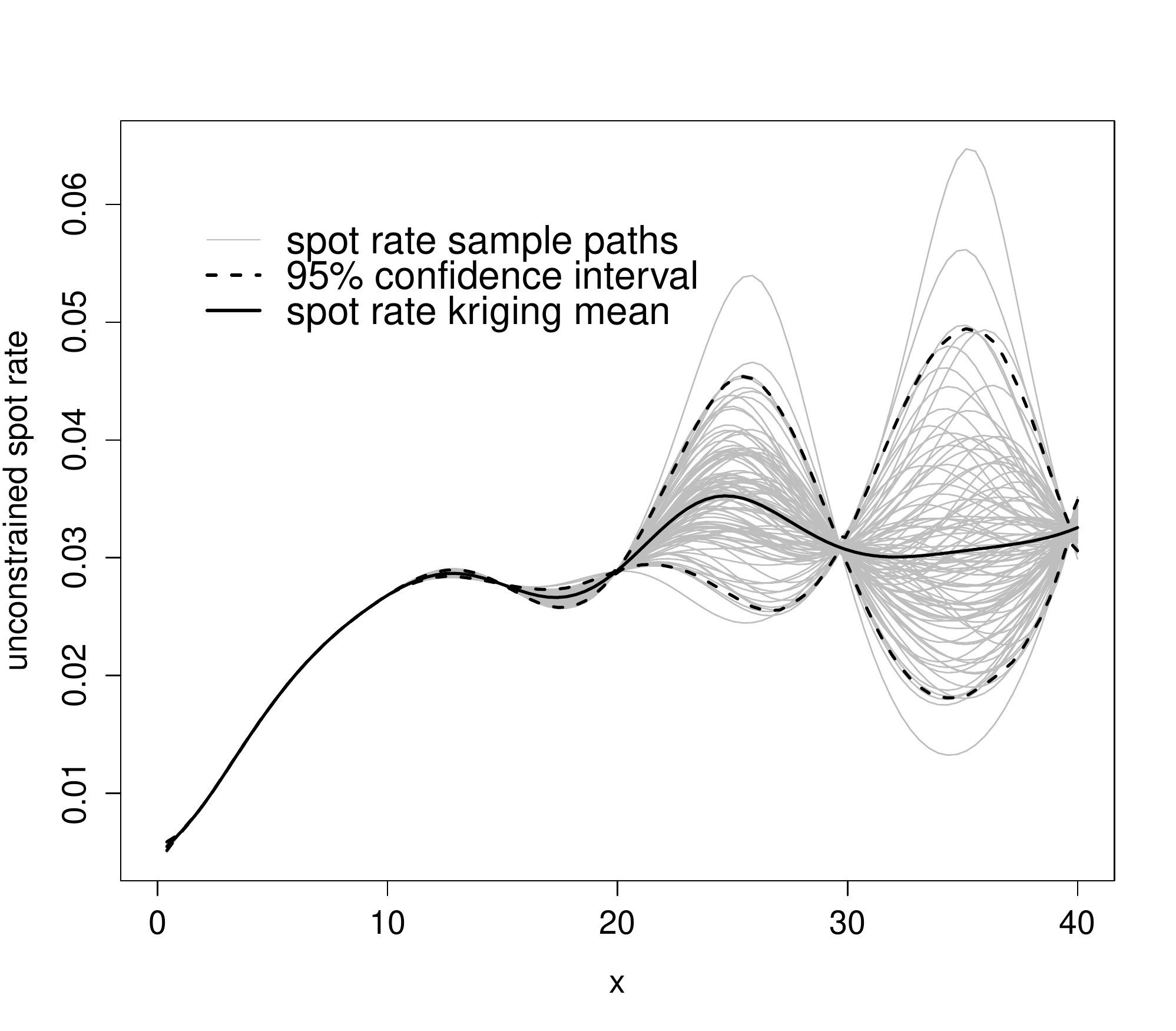}
\end{minipage}
\caption{Interpolation of some discount factors (left panel) using standard kriging techniques without monotonicity constraints. OIS Data, 03/06/2010, Gaussian kernel with nugget $10^{-7}$
. Corresponding spot rates are given in the right panel.
In this figure only, linear constraints $A \mathbf{Y}(\mathbf{X})=\mathbf{b}$ has been approximated by point-wise constraints $Y(x^{(i)})= M_K^N(x^{(i)}|A,\mathbf{b})$, $i=1, \ldots, m$, in order to use standard unconstrained R package \textit{DiceKriging} \citep[see][]{Roustant:Ginsbourger:Deville:2012:JSSOBK:v51i01}. Estimated parameters from this package are $(\theta,\sigma^2)=(5.8, 0.22)$.}
\label{fig:motivMonotoneDF}
\end{figure}


If one wants to avoid using monotone interpolation techniques, a natural idea is to interpolate some quantities that are deduced from $Y(x)$ and that are not necessarily monotonic.
As an example, one can work on a function $\zeta(x)$ deduced in a bijective way from each path of $Y(x)$. For example in a financial setting, if $Y(x)$ are discount factors, the function $\zeta(x)$ could be spot rates, $\zeta(x)= -\frac{1}{x}\log Y(x)$, or forward rates, $\zeta(x)= -\frac{d}{dx} \log Y(x)$. Obviously these functions does not suffer from monotonicity constraints, so that this could avoid using monotone interpolation. In the literature, some studies interpolate interest rates without considering any monotonic constraints. \cite{Steeley:2008} concludes that it is better to directly interpolates spot rates rather than discount factors: ``\textit{better yield curves estimates are obtained by fitting to the yield curve directly rather than fitting first to the discount function}''. 
 Some authors also propose to use kriging techniques to some non-monotonic curves, as~\cite{Benth:2015} who applies kriging to futures curves from energy futures prices. In~\cite{Kanevski2008}, tools from spatial statistics and machine learning are applied to produce some interest rate mapping  in
a two dimensional feature space (maturity, time). However, in these studies, rates or prices to be interpolated are assumed to be directly observed. In our setting, we do not necessarily observe  spot rates or even discount factors, and one aims at fitting market data in a non-parametric setting. We can make general objections to the interpolation of non-monotonic deduced quantities : 
\begin{itemize}
\item The first objection is that even if $\zeta(x)$ is not monotone, it is \emph{still constrained}: as an example, quantities such as discount rates or forward rates  are expected to be positive in absence of arbitrage opportunity. As for the discount rates, interpolating spot rates could lead to locally negative average spot rates, and confidence intervals obtained by kriging can reach the threshold 0, which is not desirable: the constraints are translated from monotonicity constraints to positivity constraints, and classical kriging techniques cannot handle these constraints.
\item A second objection is that in our setting, observations are $A \mathbf{Y}(\mathbf{X})=\mathbf{b}$, where $A$ is not necessarily a square and invertible matrix. Thus \emph{we cannot directly observe}, with our data, spot rates or forward rates. Even in the very special case where $A$ is a square and invertible matrix, spot rates can be deduced from $\mathbf{Y}(\mathbf{X})$, but not forward rates: knowing a function at some abscissas does not give straightforward constraints on its derivative.
\item At last, using kriging interpolation, even if omitting the positivity of $\zeta(x)$, would lead to \emph{non linear conditions} in terms of $\zeta(X)$. One has seen that conditionally to $A \mathbf{Y}(\mathbf{X})=\mathbf{b}$, $Y(x)$ is still a Gaussian Process, but it would not be the case any more for the process $\zeta(x)$.
Even in the simple case where $A$ is invertible, the process $\zeta(x)$ must be positive and given $A \mathbf{Y}(\mathbf{X})=\mathbf{b}$, it is not Gaussian any more, so that suited interpolation techniques still have to be introduced.
\end{itemize}


\section*{Conclusion}
In this paper, we show how suitable kriging techniques can be used to quantify model uncertainty embedded in the construction of financial term-structures. We consider that the curve under construction is an unobservable path of a conditional (spatial) Gaussian process satisfying some linear equality constraints and monotone properties. A suitable cross-validation method is proposed to estimate the Gaussian process covariance  parameters that control the level of uncertainty. We then investigate the efficiency of the proposed approach on some illustrative examples in one and two dimensions. The generated curves are all compatible with market quotes and respect non-arbitrage conditions. The conditional Gaussian process also allows to derive confidence bands for financial term-structures and related quantities. We compare the Gaussian  and the Mat\'ern 5/2 covariance kernels over different data sets : swaps vs Euribor, overnight indexed swaps, credit default swaps. We conclude that, for these applications, the Mat\'ern 5/2 covariance kernel seems to be the more appropriate since it generates more realistic forward curves in comparison with fitted  Nelson-Siegel or Svensson models.
In this work, we do not fully investigate the impact of curve uncertainty on the assessment of related  products and their hedging strategies. Moreover, in the illustrative example we have presented, we  consider that  market information is observed without uncertainty. It may be the case  that, due to  lack of liquidity, market quotes cannot be considered to be reliable. The kriging techniques we have proposed could then be adapted to the presence of noisy observations  (see section \ref{subsec:noisy}).
Improvements on the rejection sampling algorithm could also be useful, especially for the estimation of hyper-parameters in dimension $2$ where a large number of quotation dates has to be considered. These points are leaved for future research.

\section*{Acknowledgement}
We would like to thank the three anonymous reviewers and the editor for the time spent on this article and for useful suggestions. The authors also thank Xavier Bay (EMSE) and Nicolas Durrande (EMSE) for helpful discussions. This work has been conducted within the frame of the ReDice Consortium,
gathering industrial (CEA, EDF, IFPEN, IRSN, Renault) and academic (Ecole des Mines de Saint-Etienne, INRIA, and the University of Bern) partners around advanced methods for Computer Experiments. It also benefits from the support of the GRI in Financial Services and the Louis Bachelier Institute.  One author also thanks the ANR Lolita research project and DAMI research project.
\bibliography{bibliographie}
\bibliographystyle{apalike}

\end{document}